\documentclass[11pt]{article}
\pdfoutput=1
\usepackage[T1]{fontenc}
\usepackage{lmodern}
\usepackage{slantsc}
\usepackage[protrusion=true,expansion=true]{microtype}
\usepackage{amsmath,amssymb,amsfonts,amsthm}
\usepackage{subcaption}
\usepackage{graphicx}
\usepackage{fullpage}
\usepackage{setspace}
\usepackage[backref=page]{hyperref}
\usepackage{color}
\usepackage{wrapfig}
\usepackage{tikz}
\usetikzlibrary{decorations.pathreplacing}
\usepackage{algorithm}
\usepackage[noend]{algpseudocode}
\usepackage[framemethod=tikz]{mdframed}
\usepackage{xspace}
\usepackage{pgfplots}
\usepackage{framed}
\usepackage{enumitem}
\usepackage{thmtools}
\usepackage{thm-restate}
\usepackage{tabu}
\usepackage{fancyhdr}
\pgfplotsset{compat=1.5}
\usepackage{placeins}

\newtheorem{theorem}{Theorem}[section]
\newtheorem{corollary}[theorem]{Corollary}
\newtheorem{lemma}[theorem]{Lemma}

\newtheorem{definition}[theorem]{Definition}

\newtheorem{fact}[theorem]{Fact}

\newenvironment{proofof}[1]{\begin{trivlist} \item {\bf Proof
#1:~~}}
  {\qed\end{trivlist}}

\newcommand{\namedref}[2]{\hyperref[#2]{#1~\ref*{#2}}}
\newcommand{\thmlab}[1]{\label{thm:#1}}
\newcommand{\thmref}[1]{\namedref{Theorem}{thm:#1}}
\newcommand{\lemlab}[1]{\label{lem:#1}}
\newcommand{\lemref}[1]{\namedref{Lemma}{lem:#1}}

\newcommand{\corlab}[1]{\label{cor:#1}}
\newcommand{\corref}[1]{\namedref{Corollary}{cor:#1}}
\newcommand{\seclab}[1]{\label{sec:#1}}
\newcommand{\secref}[1]{\namedref{Section}{sec:#1}}

\newcommand{\factlab}[1]{\label{fact:#1}}
\newcommand{\factref}[1]{\namedref{Fact}{fact:#1}}

\newcommand{\figlab}[1]{\label{fig:#1}}
\newcommand{\figref}[1]{\namedref{Figure}{fig:#1}}
\newcommand{\alglab}[1]{\label{alg:#1}}
\renewcommand{\algref}[1]{\namedref{Algorithm}{alg:#1}}

\newcommand{\deflab}[1]{\label{def:#1}}
\newcommand{\defref}[1]{\namedref{Definition}{def:#1}}

\newcommand{\frameref}[1]{\namedref{Famework}{frame:#1}}
\newcommand{\framelab}[1]{\label{frame:#1}}

\newcommand{\Zhao}[1]{{}}

\newcommand{\SetPL}{\textbf{\textsf{PL}}}

\newcommand{\PPr}[1]{\ensuremath{\mathbf{Pr}\left[#1\right]}}

\newcommand{\Ex}[1]{\ensuremath{\mathbb{E}\left[#1\right]}}

\renewcommand{\O}[1]{\ensuremath{\mathcal{O}\left(#1\right)}}
\newcommand{\tO}[1]{\ensuremath{\tilde{\mathcal{O}}\left(#1\right)}}
\newcommand{\eps}{\varepsilon}
\newcommand{\Var}[1]{\ensuremath{\text{Var}\left[#1\right]}}

\def \calA    {\mdef{\mathcal{A}}}
\def \calB    {\mdef{\mathcal{B}}}

\def \calE    {\mdef{\mathcal{E}}}

\def \calL    {\mdef{\mathcal{L}}}

\def \calP    {\mdef{\mathcal{P}}}
\def \calQ    {\mdef{\mathcal{Q}}}

\def \bA    {\mdef{\mathbf{A}}}

\def \bE    {\mdef{\mathbf{E}}}

\def \be    {\mdef{\mathbf{e}}}

\def \bu    {\mdef{\mathbf{u}}}

\def \bv    {\mdef{\mathbf{v}}}

\def \bi    {\mdef{\mathbf{i}}}

\newcommand{\mdef}[1]{{\ensuremath{#1}}\xspace}  
\DeclareMathOperator*{\median}{median}

\DeclareMathOperator*{\argmax}{argmax}
\DeclareMathOperator*{\polylog}{polylog}
\DeclareMathOperator*{\poly}{poly}

\newcommand{\ignore}[1]{}

\newif\ifnotes\notestrue 
\ifnotes
\newcommand{\samson}[1]{\textcolor{blue}{{\bf (Samson:} {#1}{\bf ) }} \marginpar{\tiny\bf
             \begin{minipage}[t]{0.5in}
               \raggedright S:
            \end{minipage}}}
\newcommand{\shenghao}[1]{\textcolor{red}{{\bf (Shenghao:} {#1}{\bf ) }} \marginpar{\tiny\bf
             \begin{minipage}[t]{0.5in}
               \raggedright S:
            \end{minipage}}}
\newcommand{\david}[1]{\textcolor{purple}{{\bf (David:} {#1}{\bf ) }} \marginpar{\tiny\bf
             \begin{minipage}[t]{0.5in}
               \raggedright D:
            \end{minipage}}} 
\else
\newcommand{\samson}[1]{}
\newcommand{\shenghao}[1]{}
\newcommand{\david}[1]{}
\fi

\makeatletter
\renewcommand*{\@fnsymbol}[1]{\textcolor{mahogany}{\ensuremath{\ifcase#1\or *\or \dagger\or \ddagger\or
 \mathsection\or \triangledown\or \mathparagraph\or \|\or **\or \dagger\dagger
   \or \ddagger\ddagger \else\@ctrerr\fi}}}
\makeatother

\providecommand{\email}[1]{\href{mailto:#1}{\nolinkurl{#1}\xspace}}

\definecolor{mahogany}{rgb}{0.75, 0.25, 0.0}
\definecolor{darkblue}{rgb}{0.0, 0.0, 0.55}
\definecolor{darkpastelgreen}{rgb}{0.01, 0.75, 0.24}
\definecolor{darkgreen}{rgb}{0.0, 0.2, 0.13}
\definecolor{darkgoldenrod}{rgb}{0.72, 0.53, 0.04}
\definecolor{darkred}{rgb}{0.55, 0.0, 0.0}
\definecolor{forestgreenweb}{rgb}{0.13, 0.55, 0.13}
\definecolor{greencss}{rgb}{0.0, 0.5, 0.0}
\definecolor{bleudefrance}{rgb}{0.19, 0.55, 0.91}
\definecolor{forestgreen(traditional)}{rgb}{0.0, 0.27, 0.13}

\hypersetup{
     colorlinks   = true,
     citecolor    = forestgreenweb,
	 linkcolor	  = forestgreenweb
}

\AtBeginDocument{%
  \DeclareFontShape{T1}{lmr}{m}{scit}{<->ssub*lmr/m/scsl}{}%
}
\begin{document}
\allowdisplaybreaks

\begin{titlepage}
\title{\texorpdfstring{$L_p$}{Lp} Sampling in Distributed Data Streams with Applications to Adversarial Robustness}
\author{
Honghao Lin \\ Carnegie Mellon University \\ \email{honghaol@andrew.cmu.edu} 
\and
Zhao Song \\ UC Berkeley \\ \email{magic.linuxkde@gmail.com} 
\and 
David P. Woodruff \\ Carnegie Mellon University \\ \email{dwoodruf@andrew.cmu.edu}
\and
Shenghao Xie \\ Texas A\&M University \\ 
\email{xsh1302@gmail.com}
\and
Samson Zhou \\ Texas A\&M University \\ 
\email{samsonzhou@gmail.com}
}
\date{}
\maketitle

\begin{abstract}
In the distributed monitoring model, a data stream over a universe of size $n$ is distributed over $k$ servers, who must continuously provide certain statistics of the overall dataset, while minimizing communication with a central coordinator. In such settings, the ability to efficiently collect a random sample from the global stream is a powerful primitive, enabling a wide array of downstream tasks such as estimating frequency moments, detecting heavy hitters, or performing sparse recovery. Of particular interest is the task of producing a perfect $L_p$ sample, which given a frequency vector $f \in \mathbb{R}^n$, outputs an index $i$ with probability $\frac{f_i^p}{\|f\|_p^p}+\frac{1}{\mathrm{poly}(n)}$. In this paper, we resolve the problem of perfect $L_p$ sampling for all $p\ge 1$ in the distributed monitoring model. Specifically, our algorithm runs in $k^{p-1} \cdot \mathrm{polylog}(n)$ bits of communication, which is optimal up to polylogarithmic factors. 

Utilizing our perfect $L_p$ sampler, we achieve adversarially-robust distributed monitoring protocols for the $F_p$ moment estimation problem, where the goal is to provide a $(1+\varepsilon)$-approximation to $f_1^p+\ldots+f_n^p$. Our algorithm uses $\frac{k^{p-1}}{\varepsilon^2}\cdot\mathrm{polylog}(n)$ bits of communication for all $p\ge 2$ and achieves optimal bounds up to polylogarithmic factors, matching lower bounds by Woodruff and Zhang (STOC 2012) in the non-robust setting. Finally, we apply our framework to achieve near-optimal adversarially robust distributed protocols for central problems such as counting, frequency estimation, heavy-hitters, and distinct element estimation. 
\end{abstract}
\setcounter{page}{0}
\thispagestyle{empty}
\end{titlepage}

\section{Introduction}
As modern applications often require evolving data to be collected and stored across multiple servers, the \emph{distributed monitoring model} has emerged as an increasingly important setting. 
In this model, there exist $k$ servers and a central coordinator, with two-way communication links between the coordinator and each server. 
The network iteratively receives updates so that at each time $t$, server $j\in[k]$ has a dataset $f^{(t)}(j)$. 
The goal of the coordinator is to aggregate statistics of the overall dataset, which comprises of the union of the $k$ datasets $f^{(t)}(1),\ldots,f^{(t)}(k)$ at each time $t$, while minimizing the total amount of communication between the sites and the servers, e.g., due to power constraints in sensors networks~\cite{EstrinGHK99,juang2002energy,madden2005tinydb}, or bandwidth constraints in distributed systems~\cite{ribeiro2006bandwidth,ribeiro2006bandwidth2}. 
Motivated by its applications to the increasing prevalence of distributed systems, the distributed monitoring model has recently received significant study~\cite{DilmanR02,BabcockO03,CormodeGMR05,KeralapuraCR06,ArackaparambilBC09,CormodeMY11,TirthapuraW11,CormodeMYZ12,HuangYZ12,WoodruffZ12,YiZ13,ChenZ17,WuGZ20}. 

\paragraph{Perfect $L_p$ sampler.}
Sampling is a powerful and widely used technique for analyzing large-scale datasets, particularly in streaming and distributed settings~\cite{Vitter85,GemullaLH08,CohenDKLT11,CohenDKLT14}. 
By selecting a small representative subset of the data, sampling enables efficient computation and communication, and thus serves as a fundamental tool in applications such as data summarization~\cite{Vitter85,GemullaLH08, CohenDKLT11,CohenDKLT14}, network traffic monitoring~\cite{GilbertKMS01, EstanV03,MaiCSYZ06,HuangNGHJJT07,ThottanLJ10}, database management~\cite{HaasS92, LiptonN95,GibbonsM98,Haas16,CohenG19}, and data summarization~\cite{FriezeKV04,DeshpandeRVW06,AggarwalDK09,IndykMMM14,MahabadiIGR19, IndykMGR20,MahabadiRWZ20,cohen2020wor}. 
In this paper, we consider the task of sampling in the distributed monitoring model, where a distributed data stream is processed across $k$ servers. 
At each time step $t$, an update $i_t\in[n]$ arrives at some server $j_t\in[k]$. 
Each server $j$ maintains a local frequency vector $f^{(t)}(j) \in \mathbb{R}^n$, where $f^{(t)}_i(j)$ is the number of times $i$ has appeared in the data stream local to server $j$ up to time $t$. 
These collectively induce the global frequency vector $f^{(t)} = \sum_{j \in [k]} f^{(t)}(j)$. 
The goal is to sample an index $i\in[n]$ with probability proportional to some fixed weight function $G(x_i)$ and possibly obtain an approximation of the value $x_i$. 
Among the most impactful and extensively studied choices for the function $G(\cdot)$ is the family of functions $G(z) = |z|^p$ for $p>0$, referred to as an \emph{$L_p$ sampler}, which have been used in tasks such as identifying heavy hitters~\cite{WoodruffZ21b,JayaramWZ24,WoodruffXZ25}, low-rank approximation~\cite{DeshpandeRVW06,DeshpandeV06,MahabadiRWZ20}, $L_p$ norm and $F_p$ moment estimation~\cite{WoodruffZ21b,JayaramWZ24,WoodruffXZ25}, finding duplicates~\cite{JowhariST11}, mini-batch sampling in stochastic gradient descent~\cite{MahabadiWZ22}, cascaded norm approximation~\cite{JayramW09}, and data summarization, projective clustering, and volume maximization~\cite{DeshpandeRVW06,DeshpandeV07,CivrilM09,IndykMMM14,MahabadiIGR19,IndykMGR20,MahabadiRWZ20}. 
Hence, an actively growing body of work has studied $L_p$ sampling in the streaming model and its variants~\cite{MonemizadehW10,AndoniKO11,JowhariST11,BravermanOZ12,JayaramW18,cohen2020wor,JayaramWZ22,SwartworthWZ25,WoodruffXZ25}.  
We recall the definition of $L_p$ sampler, formalized to the distributed monitoring setting as follows: 
\begin{definition}[Continuous $L_p$ sampler]
\deflab{def:lp}
At each time $t\in[m]$, a continuous $L_p$ sampler returns an index $\bi^*$, such that for all $j \in [n]$, we have
\[\PPr{\bi^* = j} = \left(1\pm\eps\right)\cdot\frac{(f^{(t)}_i)^p}{\|f^{(t)}\|_p^p} \pm n^{-C},\]
where $f^{(t)}\in \mathbb{R}^{n}$ is the frequency vector at each time $t$ and $C>0$ can be chosen to be any constant. 
When $\eps = 0$, the sampler is said to be \emph{perfect}; otherwise, it is said to be \emph{approximate}. 
\end{definition}
Due to the distortions in the sampling probabilities, approximate $L_p$ samplers can cause an $\eps$-fractional bias, which may compound over the course of more intricate algorithms built upon these samplers, causing downstream inaccuracies.
Perfect $L_p$ samplers reduce these issues by ensuring small total variation distance from the desired joint distribution across even a polynomial number of samples. 

This is especially important in privacy-preserving applications, where a sampled index $i\in [n]$ is exposed to an external adversary $\calA$, with the goal of limiting leakage of global information about the dataset. 
It is possible for an approximate sampler to bias the sampling probabilities for a large set $S$ of indices by $(1+\eps)$ if a certain global property $\calP$ holds, and instead bias the sampling probabilities of $S$ by $(1-\eps)$ if $\calP$ does not hold. 
In this case, $\calA$ can deduce $\calP$ from $\O{\frac{1}{\eps^2}}$ samples. 
By contrast, perfect $L_p$ samplers can protect such global properties even with the release of a number of samples polynomial in $n$. 

Surprisingly, despite a long line of research in the distributed monitoring model, previous works have only been able to achieve perfect $L_p$ samplers in this setting for $p=1$~\cite{JayaramSTW19}. 
However, their techniques rely on the additive nature of the $L_1$ norm and do not generalize to other regimes of $p$. 
In fact, to the best of our knowledge, there do not even exist constructions of \emph{approximate} $L_p$ samplers for $p\neq 1$ that do not use non-trivial communication. 
This gap raises a fundamental question:
\begin{quote}
Is it possible to design continuous perfect $L_p$ samplers using total communication sublinear in the length of the data stream?
\end{quote}

\subsection{Our Contributions}
In this paper, we resolve the above question in the affirmative, introducing the first perfect $L_p$ sampler in the distributed monitoring setting.
\begin{restatable}{theorem}{thmlpsamplerub}
\thmlab{thm:lp:sampler:ub}
Given a distributed stream of length $m=\poly(n)$ on a universe of size $n$ across $k$ servers, there exists a continuous perfect $L_p$ sampler for $p\ge 1$ that uses $k^{\max(1,p-1)}\cdot\polylog(n)$ bits of communication. 
\end{restatable}
We emphasize that by comparison, previous constructions were not even able to achieve \emph{approximate} $L_p$ samplers for $p\neq 1$ using $o(m)$ communication, whereas we are able to achieve perfect $L_p$ samplers. 
Thus our work represents a substantial contribution across the total communication of the protocol, the distortion of the $L_p$ sampling probabilities, and the regime of $p$. 
Moreover, our $L_p$ sampler can be implemented to also output an unbiased estimate to $f_{\bi^*}^{(t)}$ for the sampled index $\bi^*$ at each time in the stream, with variance at most $(f_{\bi^*}^{(t)})^2$. 
We give a lower bound showing that any perfect $L_p$ sampler for a static dataset distributed across $k$ servers must necessarily use $\Omega(k^{p-1})$ bits of communication. 
\begin{restatable}{theorem}{thmlpsamplerlb}
\thmlab{thm:lp:sampler:lb}
Any perfect $L_p$ sampler for $p\ge 1$ must use $\Omega(k^{\max(1,p-1)})$ bits of communication. 
\end{restatable}
Hence, \thmref{thm:lp:sampler:lb} shows that our continuous perfect $L_p$ sampler in \thmref{thm:lp:sampler:ub} is optimal up to $\polylog(n)$ factors. 
Given its role as a core algorithmic primitive, our result enables a range of new protocols in the distributed monitoring model, which we now describe. 

\paragraph{Moment estimation and adversarial robustness.} 
A particular application where $L_p$ sampling is often used is $F_p$ moment estimation, where the goal is to output a $(1+\eps)$-approximation to $\|f^{(t)}\|_p^p:=\sum_{i\in[n]}(f^{(t)}_i)^p$. 
The $F_p$ moment estimation problem has a large number of applications. 
For $p\in[1,2]$, it has been used for computer networking in the context of using entropy estimation to identify port scanning attacks~\cite{FeigenbaumKSV02,LakhinaCD05,XuZB05,WoodruffZ21b} and to measure the entropy of origin-destination pairs~\cite{ZhaoLOSWX07,HarveyNO08}. 
For $p=2$, it has been used to estimate join and self-join sizes~\cite{AlonGMS02} and to detect network anomalies~\cite{KrishnamurthySZC03,ThorupZ04}. 
Furthermore, the frequency moment estimation problem is closely related to the norm estimation problem, where the goal is to compute a $(1+\eps)$-approximation to $\|f\|_p$, as a $(1+\eps)$-approximation algorithm to one of these problems can be used to obtain a $(1+\eps)$-approximation to the other for constant $p>0$, after an appropriate rescaling of $\eps$. 
In particular, larger values of $p$ provide a smooth interpolation toward $L_\infty$, which captures the maximum frequency and is useful for identifying the most dominant elements in a dataset. 
Thus, beginning with the seminal work of \cite{AlonMS99}, there has been a rich body of literature studying $F_p$ moment estimation in various sublinear algorithm models, e.g.,~\cite{Indyk00,ChakrabartiKS03,Bar-YossefJKS04,Woodruff04,IndykW05,Li08,KaneNW10,KaneNPW11,AndoniKO11,CormodeMY11,WoodruffZ12,KannanVW14,BlasiokDN17,WoodruffZ21a,WoodruffZ21b,EsfandiariKMWZ24,JayaramWZ24}.

Recent works observed that in many big data applications, future inputs may not necessarily be independent of previous outputs, e.g., in repeated interactions with a database, future queries may depend on answers to previous queries. 
As a result, \cite{BenEliezerJWY20} introduced the adversarially robust streaming model, which has garnered a significant amount of attention~\cite{Ben-EliezerY20,HassidimKMMS20,AlonBDMNY21,BravermanHMSSZ21,KaplanMNS21,WoodruffZ21b,AjtaiBJSSWZ22,BeimelKMNSS22,Ben-EliezerEO22,Ben-EliezerJWY22,ChakrabartiGS22,AvdiukhinMYZ19,AssadiCGS23,AttiasCSS23,CherapanamjeriS23,DinurSWZ23,GribelyukLWYZ24,WoodruffZ24,GribelyukLWYZ25,AndoniHKO25,CohenSS25,Cohen-AddadXWZ25} and was recently adapted by \cite{XiongZH23} to the distributed monitoring setting. 
The model can be captured by the following two-player game between a distributed monitoring protocol $\calP$ and a source $\calA$ of adaptive or adversarial input to $\calP$. 
After a fixed query function $\calQ$ is given, the game proceeds over $m$ rounds, and in the $t$-th round:
\begin{enumerate}
\item
The source $\calA$ computes a pair $s_t=(u_t,i_t)$, where $u_t\in[n]$ and $i_t\in[k]$, which increments coordinate $u_t$ at server $i_t$.
\item
The protocol $\calP$ processes $(u_t,i_t)$ and outputs its current answer $Z_t$.
\item
The source $\calA$ observes and records the response $Z_t$.
\end{enumerate}
The goal of the protocol $\calP$ is to produce a correct answer $Z_t$ to the query function $\calQ$ on the dataset induced by the adaptive distributed stream $\{s_1,\ldots,s_t\}$ induced by the source $\calA$, across all times $t\in[m]$. 
\cite{XiongZH23} observed that the existing analysis of many randomized distributed monitoring protocols fail because the randomness is assumed to be independent of the input, an assumption that no longer due to the adaptive nature of the input. 
In fact, \cite{XiongZH23} was only able to achieve a robust protocol for distributed counting, using a novel formulation of partial differential privacy and significant technical analysis. 

We use our continuous perfect $L_p$ sampler to achieve an adversarially robust distributed monitoring protocol for $F_p$ estimation: 
\begin{theorem}
\thmlab{thm:fp:robust}
Given a distributed stream of length $m=\poly(n)$ on a universe of size $n$ across $k$ servers, a parameter $p\ge 2$, and an accuracy parameter $\eps\in(0,1)$, there exists an adversarially robust algorithm that outputs a $(1+\eps)$-approximation to the $F_p$ moment at all times, using $\frac{k^{p-1}}{\eps^2}\cdot\polylog(n)$ bits of communication.
\end{theorem}
\thmref{thm:fp:robust} is the first result to achieve adversarial robustness for $F_p$ moment estimation in the distributed monitoring setting. 
Moreover, we remark that in the one-shot setting, where each server has a static vector, it is known that $\Omega\left(\frac{k^{p-1}}{\eps^2}\right)$ bits of communication is necessary~\cite{EsfandiariKMWZ24}. 
Therefore, not only is \thmref{thm:fp:robust} optimal up to polylogarithmic factors, but it shows that perhaps surprisingly, there is no overhead in either the $\frac{1}{\eps}$ or the $k$ factors to achieve adversarial robustness. 

Crucially, we achieve \thmref{thm:fp:robust} by using our perfect $L_p$ sampler to design a difference estimator for $F_p$ moment estimation. 
For the purposes of moment estimation, a difference estimator is an algorithmic primitive that achieves an estimate of $F_p(\bu+\bv)-F_p(\bu)$ with additive error $\eps\cdot F_p(\bu)$, given two frequency vectors $\bu,\bv\in\mathbb{R}^n$, under the promise that $F_p(\bu+\bv)-F_p(\bu)\le\gamma\cdot F_p(\bu)$. 
Difference estimators are useful because, due to the upper bound on the difference, it is sometimes possible to achieve more efficient protocols based on the value of $\gamma$. 
However, their constructions and analysis are often intricate, e.g., the design of efficient constructions for $F_p$ moment estimation in the streaming model for non-integer $p>2$ currently remains to be an open question~\cite{WoodruffZ21b}. 
Thus, our techniques may be of independent interest. 

Along the way, we develop a framework for integrating our difference estimator toward achieving adversarial robustness in the distributed monitoring model. 
In addition to $F_p$ moment estimation, our framework can also be used to achieve robust protocols for the counting problem, i.e., $p=1$, the distinct element estimation problem, i.e., $p=0$, and the $L_2$-heavy hitters problem: 
\begin{theorem}
\thmlab{thm:main_three}
Given an accuracy parameter $\eps\in(0,1)$, there exists an adversarially robust distributed monitoring protocol that achieves an $(1+\eps)$-approximation to:
\begin{enumerate}
\item 
The distinct elements problem, i.e., $F_0$ estimation, using $\frac{k}{\eps^2}\cdot\polylog(n)$ bits of communication. 
\item 
The counting problem, i.e., $F_1$ estimation, using $\left(k+\frac{\sqrt{k}}{\eps}\right)\cdot\polylog(n)$ bits of communication.  
\item 
$L_2$-heavy hitters, using $\frac{k}{\eps^2}\cdot\polylog(n)$ bits of communication. 
\end{enumerate}
\end{theorem}
We remark that previous works have shown communication lower bounds of $\Omega\left(\frac{k}{\eps^2}\right)$ bits for distinct element estimation~\cite{WoodruffZ12}, $\Omega\left(k+\frac{\sqrt{k}}{\eps}\log n\right)$ bits for counting~\cite{HuangYZ12}, and $\Omega\left(\frac{k}{\eps^2}\right)$ bits for $L_2$-heavy hitters. 
Thus our results in \thmref{thm:main_three} conceptually show that there is no overhead to achieve adversarial robustness for these central problems, up to polylogarithmic factors. 
Our framework is quite general, utilizing difference estimators and existing non-robust distributed monitoring protocols in a black-box and intuitive manner, and thus may be amenable to a number of additional other problems. 
Previously, there were no known non-trivial robust distributed monitoring protocols for distinct element estimation and $L_2$-heavy hitters. 
For the problem of counting, our robust algorithm matches the communication cost of the protocol of \cite{XiongZH23}, though our algorithm and analysis may arguably be much simpler. 
We summarize these results in \figref{fig:results:ar}.

\begin{figure}[!htb]
\centering
{
\tabulinesep=1.1mm
\resizebox{\columnwidth}{!}{
\begin{tabu}{|c|c|c|c|}\hline
{\bf Problem} &  {\bf LB} & {\bf Prior non-robust UB} & {\bf Our Robust UB}  \\\hline\hline
Distinct elements & $\Tilde{\Omega}\left(\frac{k}{\eps^2}\right)$ \cite{WoodruffZ12,WoodruffZ14} & $\Tilde{{\cal O}}\left(\frac{k}{\eps^2}\right)$ \cite{CormodeMY11} & $\Tilde{{\cal O}}\left(\frac{k}{\eps^2}\right)$  \\\hline
Count-tracking & $\Tilde{\Omega}\left(k+\frac{\sqrt{k}}{\eps}\right)$ \cite{HuangYZ12} & $\Tilde{{\cal O}} \left(k+\frac{\sqrt{k}}{\eps}\right)$ \cite{HuangYZ12} & $\Tilde{{\cal O}} \left(k+\frac{\sqrt{k}}{\eps}\right)$  \\\hline
$F_p$ estimation, $p\ge 2$ & $\Tilde{\Omega}\left(\frac{k^{p-1}}{\eps^2}\right)$ \cite{WoodruffZ12} & $\Tilde{{\cal O}}\left(\frac{k^{p-1}}{\eps^2}\right)$ \cite{HuangXZW25} &  $\Tilde{{\cal O}}\left(\frac{k^{p-1}}{\eps^2}\right)$ \\\hline
$L_2$ heavy-hitters & $\Tilde{\Omega}\left(\frac{k}{\eps^2}\right)$ \cite{WoodruffZ12} & $\Tilde{{\cal O}}\left(\frac{k}{\eps^2}\right)$ \cite{WoodruffZ12} &  $\Tilde{{\cal O}}\left(\frac{k}{\eps^2}\right)$ \\\hline
\end{tabu}
}
}
\caption{Comparison between non-robust and robust $(1+\eps)$-approximation algorithms for distributed monitoring. Here, $\Tilde{{\cal O}}$ and $\Tilde{\Omega}$ hide the $\polylog(n/\eps)$ factors. {\bf UB} denotes the upper bound and {\bf LB} denotes the lower bound. Our results are from \thmref{thm:main_three}.}
\figlab{fig:results:ar}
\end{figure}

\paragraph{Comparison with \cite{HuangXZW25}.} 
In concurrent and independent work, \cite{HuangXZW25} recently also achieved a distributed monitoring protocol for $F_p$ moment estimation using total communication $\frac{k^{p-1}}{\eps^2}\cdot\polylog(n)$ for $p\ge 2$, which matches our result in \thmref{thm:fp:robust}, but for \emph{non-adaptive} distributed streams. 
Their protocol modifies a scheme previously used by \cite{IndykW05,WoodruffZ12,BravermanO13a} by carefully tracking the frequencies of heavy-hitters from substreams induced by subsampling the universe $[n]$ at different levels using the notion of $(\alpha,\eps)$-covers. 
Both their algorithm and analysis are quite elegant and significantly simpler than our result in \thmref{thm:lp:sampler:lb}. 
However, as their approach fundamentally relies on heavy-hitter detection, they do not address the design of perfect $L_p$ samplers, which is the core contribution of our work in \thmref{thm:lp:sampler:ub}. 
Moreover, the ability to perform perfect $L_p$ sampling is essential to our difference estimator, which enables adversarial robustness in \thmref{thm:fp:robust}.  
It is not clear whether their non-robust, heavy-hitter-based techniques can be extended to yield adversarial robustness for $F_p$ moment estimation especially given that, prior to this work, no existing approach could achieve adversarial robustness for heavy-hitters in the distributed monitoring model. 
Thus, although our results overlap in achieving similar bounds for non-robust moment estimation, our work is fundamentally different, both technically and conceptually, due to its focus on perfect $L_p$ sampling and its applications to adversarial robustness.

\subsection{Technical Overview}
In this section, we describe the intuition behind our algorithms and techniques, and how we overcome significant barriers for previous approaches. 
In \secref{sec:lp:sampler}, we introduce our perfect $L_p$ sampler.
In \secref{sec:framework:tech}, we describe the difference estimator framework.
Lastly in \secref{sec:diff:lp}, we provide the $F_p$ difference estimator, which is essentially based on our $L_p$ sampler.

\subsubsection{\texorpdfstring{$L_p$}{Lp} Sampling}
\seclab{sec:lp:sampler}
A natural starting point might be the perfect $L_1$ sampler from \cite{JayaramSTW19} in the distributed monitoring model. 
However, at a high level, sampling is significantly easier due to the additivity of $L_1$. 
Thus it is not clear how to extend this approach to $p\neq 1$. 
$L_p$ sampling has also been well-studied in the streaming model, where the data stream is given in its entirety to an algorithm that is bounded in space and cannot store the entire dataset~\cite{MonemizadehW10,AndoniKO11,JowhariST11,BravermanOZ12,JayaramW18,cohen2020wor,JayaramWZ22,SwartworthWZ25,WoodruffXZ25}. 
Unfortunately, these approaches do not directly extend to the distributed monitoring model, where informing the coordinator of each update incurs communication cost linear in the length of the stream, and moreover, only \cite{WoodruffXZ25} achieves near-optimal perfect $L_p$ sampling for $p\ge 2$. 
On the other hand, the universal template for $L_p$ sampling on an input vector $f\in\mathbb{R}^n$ is to (1) perform some linear transformation on $f$ to obtain a scaled vector $g$, (2) estimate each entry of $g$ and find the maximum index $i^*$ of the estimated vector, (3) run a statistical test to decide whether to report $i^*$. 
\cite{AndoniKO11,JayaramW18,WoodruffXZ25} observed that if we define $g_i =\frac{f_i}{\be_i^{1/p}}$ for all $i\in[n]$, where each $\be_i$ is an independent exponential random variable with rate $1$, then the probability of $i^*$ being the maximum index of $g$ is precisely the $L_p$ sampling probability (see \defref{def:lp}), i.e.,
\[\PPr{i^* = \argmax_i g_i} = \frac{f_i^p}{\|f\|_p^p},\]
which follows the desired distribution. 
Therefore, it suffices to continuously track the largest index of $g$ at all times in the distributed stream. 
Unfortunately, this turns out to be a difficult task in general; even for a universe of size $n=2$, tracking the maximum frequency at all times on a stream of length $m$ provably requires $\Omega(m)$ communication in the worst case. 
Circumventing this issue forms the majority of both our algorithmic novelty and technical analysis for our perfect $L_p$ sampler.

\paragraph{Global $L_1$-of-$L_p$ sampling.}
Due to distributional properties of exponential random variables, the largest index of $g$ is a $L_p$-heavy-hitter.
Then suppose that there is a procedure that identifies the set of $L_p$-heavy-hitters at all times, it is our hope that we can keep track of the items in this set with a higher accuracy, to recover or approximate the maximum.
To begin with, we introduce the construction of this $L_p$-heavy-hitter algorithm.

For ease of notation, we fix $t\in[m]$ and simply write $g_i=g_{i,t}$. 
It is known that with high probability, the largest rescaled coordinate $g_{i}$ is a $\O{\frac{1}{\log^{2/p} n}}$-$L_p$ heavy-hitter with respect to $g$, i.e., $\max_{i\in[n]} g_i^p\ge\O{\frac{1}{\log^2 n}}\cdot\|g\|_p^p$. 
Now, consider the following hard distribution. 
Let $g_n(j)$ be the value of $i$ at server $j$, where $n=\argmax_{i\in[n]} g_i$ without loss of generality.  
Suppose $g_n(j)=1$ for all $j\in[k]$ so that $g_n=k$. 
Suppose that for each server $j$, we have $g_{(j-1)\cdot k^{p-1}+1}(j)=\ldots=g_{j\cdot k^{p-1}}(j)=1$, i.e., each server has a single instance of coordinate $n$, as well as $k^{p-1}$ additional disjoint coordinates so that $\|g\|_p^p=k^p + k\cdot(k^{p-1})=2k^p$. 
Then $n$ is a $\O{1}$-$L_p$ heavy and could be the largest coordinate after a possible rescaling of $f$, so we must identify $g_n$ in this case. 

Observe that in this case, the maximizer is a $\frac{1}{k^{p-1}}$-$L_p$ heavy-hitter at each site. 
Suppose that for each coordinate $i\in[n]$ and server $j\in[k]$, we sample $g_i(j)$ with the \emph{global $L_1$ of $L_p$} probability distribution $\min\left(k^{p-1}\cdot\frac{(g_i(j))^p}{\sum_{a\in[k]}\sum_{b\in[n]}(g_b(a))^p},1\right)$. 
In the above example, each of the $k$ servers will report the maximizer with probability roughly $\frac{1}{k}$ and thus the maximizer will be reported with probability roughly $1-\left(1-\frac{1}{k}\right)^k=\Omega(1)$ across the $k$ servers. 
Moreover, under this sampling distribution, $k^{p-1}$ samples would be returned in expectation and thus it can be shown that over the course of the stream, the total communication is at most $k^{p-1}\cdot\polylog(n)$, as desired. 

In fact, after the sampling procedure, there exists a reconstruction algorithm that estimates the value of each scaled coordinate $g_i$, eliminating all but at most $\polylog(n)$ candidate coordinates from possibly being the maximizer $\argmax_{i\in[n]}g_i$. 
In other words, for a fixed instance $t$, we can identify a set $\SetPL_t\subset[n]$ (standing for ``probably large'') with at most $\polylog(n)$ coordinates such that $\argmax_{i\in[n]}g_{i,t}\in\SetPL_t$. 

\paragraph{Statistical test and dependence on the maximum index.}

The task is then to extract the maximizer from $\SetPL_t$ at all times, which essentially requires continuously estimating the frequency of each item $g_i\in\SetPL_t$. 
A natural approach would be to apply the continuous counting algorithm of \cite{HuangYZ12}, which will provide an unbiased estimate to each $g_i\in\SetPL_t$ with variance $\O{g_i^2}$. 
However, the algorithm only gives a constant-multiplicative approximation to each $g_i$; and as we mentioned earlier, it is impossible to obtain the exact value at all times.
Fortunately, due to distributional properties of the exponential variables, the maximum of the scaled vector $g$ is at least twice the second-maximum with constant probability.
Given this observation, \cite{JayaramW18,WoodruffXZ25} performed a statistical test to accept the instances with this desired guarantee to provide an $L_p$ sample.

Formally, let $\widehat{g_{D(1)}}$ and $\widehat{g_{D(2)}}$ be $1.01$-approximations of the maximum and the second-maximum, we accept a sample if and only if $\widehat{g_{D(1)}} \ge 2 \cdot \widehat{g_{D(2)}}$.
However, the algorithm of \cite{HuangYZ12} maintains an approximate counter for each local count $g_i(j)$ at the coordinator, which sums up all approximate counters; and as a result, the estimate depends on how $g_i$ is distributed across the $k$ servers.
For example, if all updates go to a single server, the error is expected to be significantly smaller than that in the case that the updates are uniformly distributed across the $k$ servers. 
But then the statistical test is conducted on the estimated value, so the failure probability of the test depends on the distribution of the maximum coordinate across the $k$ servers.
Therefore, the probability that we fail a sampler depends on which index $i^*$ achieves the max in the scaled vector, which destroys the perfect $L_p$ sampling probability distribution because each index may have a different failure probability and thus, we cannot apply the continuous counting algorithm of \cite{HuangYZ12}. 
Hence, we need to design a novel distributed counting algorithm, whose estimation error is independent of how the coordinate is distributed across the $k$ servers.

\paragraph{Double exponential-scaling method.} 
To that end, we instead scale each $g_i(j)$ by another independent exponential variable $\be(j)$ and define $h_i(j):= \frac{g_i(j)}{\be(j)}=\frac{f_i(j)}{\be_i^{1/p}\be(j)}$. 
Due to the max-stability property of the exponential variables, $\max_j h_i(j)$ has the same distribution as $\frac{g_i}{\be'_i}$, where $\be'_i$ is independent of all $\be(j)$.
Thus, we run a separate counting algorithm at each server for each $h_i(j)$, such that we can use the estimate of a single approximate counter that reaches the max, instead of the sum of all approximate counters. 
In this way, the second exponential scaling washes away dependencies, so that the resulting estimate does not depend on the distribution of $g_i$ across all servers. 
We remark that a subsequent work \cite{CohenSS25} showed that random sampling also gives an efficient distributed counting algorithm, such that the estimate does not dependent on the distribution across the servers; it is possible that this approach also gives the functionality required from our subroutine. 

\paragraph{Amortizing the communication by anti-concentration.}
It is unclear how we control the communication cost of the aforementioned counting algorithm, because we must extract the \emph{actual} max across the $k$ servers to ensure the distributional-free property.
In particular, let $X$ be the maximum and let $X'$ be the second-maximum of the scaled vector $[h_i(1),\cdots,h_i(k)]$.
If the ratio $\frac{X}{X'}=1+\alpha$, then we run a counting algorithm with accuracy $\frac{\alpha}{100}$ at each server $j$, which outputs a $(1+\frac{\alpha}{100})$-approximation to $h_i(j)$ at each $j$, so we recover the server $j$ that achieves the max.
This might require roughly $\frac{k}{\alpha}$ communication if $X$ and $X'$ are always $\alpha$-close, which is prohibitively large for our purpose.
Fortunately, due to the anti-concentration property of exponential variables, we have $\PPr{\frac{X}{X'}<1+\alpha} = \O{\alpha}$.
Then, since the counting algorithm requires $\O{\frac{1}{\alpha}}$ communication, we only need roughly $\O{\alpha} \cdot k \cdot \frac{1}{\alpha} = \O{k}$ communication in expectation.
On the other hand, we run a separate counting algorithm at each server $j$ that outputs a $1.01$-approximation $\widehat{h_i(j)}$ at all times.
If $j^*$ is detected to be the maximum server, then we take $\widehat{h_i(j^*)}$ as the estimate of $X$.
In summary, although finding the maximum value of a frequency vector $f\in\mathbb{R}^n$ defined by a distributed stream of length $m$ requires $\Omega(m)$ communication in general, this can only happen when the largest and second largest coordinates of the frequency vector are quite close. 
Critically, the anti-concentration property shows this can only happen a small fraction of times, which allows us to amortize the additional communication cost of dealing with such occurrences over the entirety of the stream. 

\paragraph{Geometric mean estimator and truncated Taylor series.}
Now that we obtain an unbiased estimate of $X=\max_j h_i(j) = \frac{g_i}{\be'_i}$ with constant-multiplicative error.
The remaining task is to estimate $g_i$ using $X$.
Unfortunately, both the expectation and the variance of $X$ can be unbounded. 
This is because the cumulative distribution function of a random variable $Z\sim\frac{1}{\be^{1/p}}$ for an exponential random variable $\be$ is roughly $\PPr{Z\le t} = 1-\Theta\left(\frac{1}{t^p}\right)$ for $t\ge 1$, and so the expectation of $Z$ is unbounded for $p\le 1$ and the variance of $Z$ is unbounded for $p\le 2$. 
On the other hand, we have $\PPr{Z^{1/3}\ge t} = \Theta\left(\frac{1}{t^{3p}}\right)$, and so both the expectation and variance of $Z^{1/3}$ are finite for $p\ge 1$. 
Hence if we take the geometric mean $Y = \sqrt[3]{X_1\cdot X_2\cdot X_3}$ of three independent instances $X_1$, $X_2$, and $X_3$, then there exists a constant $C$ such that $\Ex{C\cdot Y}=g_i$ and $\Ex{C^2\cdot Y^2}\le\O{g_i^2}$, which gives us an unbiased estimate of $g_i$. 
However, it is not instantly clear how to provide a valid estimator for $X_r^{1/3}$, because the fractional power of $\frac{1}{3}$ might bring up undesirable covariance terms.
We use a truncated Taylor series to estimate each $X_r^{1/3}$, which requires $\O{\log n}$ independent unbiased estimates of $X_r$.
Namely, we expand $X_r^{1/3}$ by $\sum_q \alpha_q X_r^q$, so that it is easier to obtain unbiased estimates of $X_r^q$ for an integer $q$.
Then, averaging $\polylog n$ instances of $Y$ gives us a $1.01$-approximation to $g_i$. 
Thus, we can track the value for each $g_i$ in $\SetPL$ and report the sample if it passes the statistical test.
Although geometric mean estimators have been used for $p$-stable random variables when $p<2$~\cite{Li08,WoodruffZ21b}, their use to upper bound the variance of inverse exponential random variables appears to be new. 
We believe this technique is of independent interest and may find broader applications.
In fact, the combination of geometric mean estimator and truncated Taylor estimator is applied again in our subroutine to give unbiased estimate of the $F_p$ moment $\|f\|_p^p$ (see \secref{sec:diff:lp}).
 
\paragraph{Core insight.} 
The underlying input in the distributed monitoring model can be captured by an $n\times k$ matrix $\bA$, where the rows of the matrix are the coordinates $i\in[n]$ and the columns are the servers $j\in[k]$. 
Then the perfect $L_p$ sampling problems involves computing the $p$-th powers of the $L_1$ norm of the rows, similar to more complex cascaded norms such as the $L_{p,q}$ norm by
\[\|\bA\|_{p,q}=\left(\sum_{i\in[n]}\left(\sum_{j\in[d]}A_{i,j}\right)^{p/q}\right)^{1/p},\]
i.e., the $L_p$ norm of the vector whose entries are the $L_q$ norms of the rows of $\bA$. 
In this setting, our two-sided embedding can be viewed as computing the matrix $\bE_1\bA\bE_2$, where $\bE_1\in\mathbb{R}^{n\times n}$ is a diagonal matrix so that each entry is $\frac{1}{\be_i^{1/p}}$ for an independent exponential random variable $\be_i$ and $\bE_2\in\mathbb{R}^{k\times k}$ is a diagonal matrix where each entry is $\frac{1}{\be(j)^{1/q}}$ for an independent exponential random variable $\be(j)$, with $q=1$ in our setting. 
Thus our main idea can be viewed as a two-sided non-negative embedding that reduces these more complicated quantities to a two-level heavy-hitters problem through properties of exponential random variables. 
The heavy-hitters in each level can be maintained and approximated. 
Moreover, although finding the maximum coordinate is difficult in general, exponential random variables enjoy anti-concentration properties so that at most $\alpha$ fraction of the instances have their maximum and second maximum values separated by less than $(1+\alpha)$, which allows us to amortize the communication. 
On the other hand, our two-sided embedding leads to values with unbounded expectation but we can address this by averaging independent instances of geometric mean estimators. 
Given this perspective, we believe our approach could potentially extend to more general samplers based on cascaded norms.

\subsubsection{Difference Estimator Framework and Adversarial Robustness}
\seclab{sec:framework:tech}
Previous work introduced several effective frameworks to achieve adversarially robust streaming algorithms in the streaming model, where efficient space complexity is prioritized.
\cite{BenEliezerJWY20} provided the well-known sketch switching framework, which initiates several sketches simultaneously, and switches to a new sketch whenever the estimate increases by an $\eps$-fraction. 
This framework achieves adversarially robustness because, the algorithm fixes the output until it switches to a new sketch, and so the inner randomness of each sketch is hidden from the adversary before its revelation. 
However, this framework loses a $\frac{1}{\eps}$ factor in the space complexity, so its extension to the distributed monitoring setting would also lose a $\frac{1}{\eps}$ factor in the communication cost.

To achieve tight bound in both $k$ and $\frac{1}{\eps}$ factors, we adapt the difference estimator framework introduced by \cite{WoodruffZ21b} to the context of distributed monitoring. 
Firstly, we run a ``double detector'' algorithm, which flags when the target function value $F$ doubles and gives a $(1+\eps)$-approximation to $F$ whenever it flags. 
So, this divides the algorithm into $\O{\log n}$ rounds, and for a vector $u$ when the double detector flags, our goal is to estimate $F(u+v) -F(u)$ for a subsequent vector $v$, where $F(u+v)\le 2F(u)$. 

The main idea is that if $F(u+v)-F(u)$ is small, say roughly $\eps \cdot F(u)$, then we only need a $\O{1}$-approximation to $F(u+v)-F(u)$, which still ensures an additive error of $\eps \cdot F(u)$, instead of losing $\frac{1}{\eps}$ to obtain a $(1+\eps)$-approximation. 
Similarly, if $F(u+v)-F(u)$ is roughly $2\eps \cdot F(u)$, then we only need a $(1+\frac{1}{2})$-approximation to $F(u+v)-F(u)$; and more generally, if $F(u+v)-F(u)$ is roughly $\frac{1}{2^r} \cdot F(u)$, then we only need a $(1+2^r\eps)$-approximation to $F(u+v)-F(u)$.

Now, for $\beta = \O{\log \frac{1}{\eps}}$ and $F(u+v)\le 2F(u)$, we decompose $F(u+v) - F(u)$ into the binary expression 
\[F(u+v) - F(u) = \sum_{r=1}^{\beta} F(u+ v_{\leq r}) - F(u+ v_{\leq r-1} ),\] 
where we define $v_{\leq t} := \sum_{j=0}^t v_t$ and $v_1,\ldots,v_\beta$ arrives in sequential order, and each difference satisfies $F(u+v_{\leq r}) - F(u+v_{\leq r-1}) \approx \frac{1}{2^r}\cdot F(u)$. 
Thus, it suffices to give a $\left(1+\frac{2^r \eps}{\beta}\right)$-approximation for each difference, which we define to be a level-$r$ difference estimator informally. 
We implement a separate difference estimator when $F(u+v)$ grows to the corresponding level and sum up the estimation of all levels to give a $(1+\eps)$-approximation. 
Since a level-$r$ difference estimator is used each time the difference increases by $\frac{1}{2^r}\cdot F(u)$, it suffices to run roughly $2^r$ instances of the level-$r$ difference estimator. 
Thus, if there exists a level-$r$ difference estimator which uses $\frac{1}{2^r \eps^2}G(n,k)$ bits of communication for each $r \in [\beta]$, then in total we use $\frac{1}{\eps^2} G(n,k)$ bits of communication in total for each level. 
Note that $\beta = \O{\log \frac{1}{\eps}}$, so we only lose a multiplicative $\O{\log \frac{1}{\eps}}$ factor in the total communication cost. 

We observe that the difference estimator framework also achieves adversarially robustness.
This can be done by using fresh randomness for each difference estimator and not changing the output until the estimated value increases by $(1+\eps)$ multiplicatively.
Therefore, when the estimate is revealed to the adversary, we start a new difference estimator at level-$\beta$ with hidden randomness, so the adaptive input stream from the adversary is independent of any internal randomness in the new difference estimator, guaranteeing adversarial robustness. 
We refer the reader to \secref{sec:diff:est} for a more detailed outline for the difference estimator framework.

\subsubsection{Adversarially Robust \texorpdfstring{$F_p$}{Fp} Estimation, \texorpdfstring{$p\ge2$}{p>2}}
\seclab{sec:diff:lp}
With the difference estimator framework, it remains to construct a $F_p$ difference estimator, so that it can be lifted to an adversarially robust $F_p$ estimation algorithm.

\paragraph{Sampling-based method.}
We denote the difference as $S = F_p(u+v) - F_p(u) = \|u+v\|_p^p - \|u\|_p^p$, with the guarantee $S < \gamma \cdot \|u\|_p^p$. 
The core idea of our difference estimator is to sample an index $i \in [n]$ with respect to its importance in the difference term $S$, and construct the estimator by rescaling the sampling probability and estimating the entry-wise difference $(u_i+v_i)^p - u_i^p$. 
Specifically, let $p_i$ be our sampling probability. 
Then, we define our estimator to be
\[\widehat{s} = \frac{1}{p_i} \cdot ( (u_i+v_i)^p - u_i^p ).\]
It is obvious that $\widehat{s}$ is an unbiased estimation of $S$. 
Thus, our goal is to develop a sampling distribution with sufficiently small variance while bounding the total communication by $\frac{\gamma k^{p-1}}{\eps^2} \polylog(n)$   bits.

\paragraph{Mixed $L_p$ sampling.}
Intuitively, we should give larger sampling probabilities to coordinates $i$ that have large values of $(u_i+v_i)^p - u_i^p$, which implies that either $u_i$ is large or $v_i$ is large. 
Thus, a natural approach would be to split the indices into sets where $u_i>v_i$ and $u_i\le v_i$, and $L_p$ sample from $u$ for indices in the first set, and $L_p$ sample from $v$ for indices in the second set. 
In practice, we could $L_p$ sample some number of indices from both $u$ and $v$ and reject certain samples that are from the wrong set, e.g., if we acquire a coordinate $i\in[n]$ by $L_p$ sampling from $u$ and it turns out that $u_i\le v_i$, or vice versa. 
However, we cannot afford to acquire the real value of $v_i$ at all times. 
Rather, we can only run a distributed counting algorithm to estimate $v_i$ with constant-multiplicative error, but then we could misclassify the index and incorrectly reject/accept some samples, which can give substantially large bias to the resulting estimator. 
Unfortunately, these simple misclassifications are known to generally require suboptimal dependencies to handle the resulting error, e.g.,~\cite{IndykW05,WoodruffZ12,LevinSW18,BlasiokBCKY17,WoodruffZ21a}. 
Additionally, another concern is that we could reject too many samples so that we need prohibitively large number of independent instances to ensure a sufficiently large number of successful samples. 

Instead, we sample from a mixed $L_p$ distribution. 
That is, with probability $\frac{1}{2}$, we obtain an $L_p$ sample from a fixed vector $u$, otherwise we obtain an $L_p$ sample from an evolving vector $v$, so that our sampling probability is $\frac{u_i^p}{2\|u\|_p^p} + \frac{v_i^p}{2\|v\|_p^p}$, thereby sampling the ``important'' indices from both $u$ and $v$ that contribute significantly to the difference. 
This construction bypasses the undesirable misclassification and rejection procedure and ensures a small variance via a careful analysis.

\paragraph{Bi-variate Taylor estimator.} 
Then we need to give unbiased estimations for $\frac{1}{p_i}$ and $(u_i+v_i)^p - u_i^p$ with bounded variance. 
However, the main challenge with this sampling distribution is that we need to estimate
\[\frac{1}{p_i} = \left(\frac{u_i^p}{2\|u\|_p^p} + \frac{v_i^p}{2\|v\|_p^p}\right)^{-1} = \frac{2\|u\|_p^p\|v\|_p^p}{\|u\|_p^pv_i^p + \|v\|_p^pu_i^p},\]
up to high accuracy. 
We can query the servers collectively for the value of $u_i$, but it is particularly challenging to deal with the denominator, since $\|u\|_p^p$ and $\|v\|_p^p$ may require querying a large number of coordinates, and $v_i^p$ may be evolving over the duration of the stream. 
To handle this, we define $a = \|u\|_p^pv_i^p$ and $b = \|v\|_p^pu_i^p$ and consider $\frac{1}{a+b}$. 
We first suppose that we have good estimators for both $a$ and $b$; we shall subsequently discuss the construction of the estimators.  
Then, we utilize the bivariate Taylor expansion
\[f(a,b) = f(a_0,b_0) + \sum_{q=1}^{\infty} \; \sum_{\substack{d,e\ge0 \\ d+e=q}} \frac{1}{d!\,e!}\left.\frac{\partial^q f(x,y)}{\partial x^d \partial y^e}\right|_{(x,y)=(a_0,b_0)} (x-a_0)^d (y-b_0)^e,\]
of the function $f(a,b) = \frac{1}{a+b}$ expanded at $a_0$, $b_0$, which are $\left(1+\frac{1}{100p}\right)$-estimations of $a$ and $b$. 
Then, to achieve an approximation to $f(a,b)$ with $\frac{1}{\poly(n)}$-bias, it suffices to truncate the Taylor series after $Q=\O{\log n}$ terms and thus we define our truncated bivariate Taylor estimator to be 
\[ \sum_{q=0}^Q  \frac{(-1)^q}{(a_0+b_0)^{q+1} }\cdot ((a-a_0)+(b-b_0))^q,\]
and we estimate the power $((a-a_0)+(b-b_0))^q$ in each term by multiplying $q$ independent estimators. 
Here, we remark that the truncated bivariate Taylor estimator is valid because the partial derivatives of $f(a,b) = \frac{1}{a+b}$ with respect to $a$ and $b$ have the same expression; the tail error bound may not be clear if the partial derivatives are different. 
On the other hand, it is an $2$-dimensional extension of the truncated Taylor estimator applied in our perfect $L_p$ sampler.

\paragraph{Unbiased estimation of $\|u\|_p^p$.}
It remains to show the estimator for $a$ and $b$. 
Here, the most challenging part is to give unbiased estimations of $\|u\|_p^p$ and $\|v\|_p^p$.
We consider the simpler case of $\|u\|_p^p$ first.
Notice that $u$ is a fixed vector, so we can obtain the actual value of each $u_i$ by querying each server, however, we still do not afford to retrieve the actual value of $\|u\|_p^p$.
Applying the max-stability property of exponential variables stated in \secref{sec:lp:sampler}, we show that if we scale each coordinate in the original vector $u$ by $g_i = \frac{u_i}{\be_i^{1/p}}$, then we have $\max_i g_i^p = \frac{\|u\|_p^p}{\be}$, where $\be$ is an independent exponential variable.
Similar as our perfect $L_p$ sampler, we still need to deal with the problem that $\frac{1}{\be}$ has an unbounded expectation, which fails to give an unbiased estimator.
Again, we use the idea of geometric mean estimator defined by the geometric mean of three independent instances $X = \max_{i \in [n]} \frac{u_i}{\be_i^{1/p}}$:
\[Y = \sqrt[3]{X_1^{p} \cdot X_2^{p} \cdot X_3^{p}},\]
which provides an unbiased estimate to $\|u\|_p^p$.
For the evolving vector $v$, we could also use the geometric mean estimator, but we are not able to recover the maximum coordinate at all times, instead we only have access to an unbiased estimate of $X$.
This requires us to apply the truncated Taylor estimator to accurately approximate the geometric mean.

\paragraph{Estimation of the entry-wise difference $(u_i+v_i)^p - u_i^p$.} 
For the estimation of $(u_i+v_i)^p - u_i^p$, it is natural to think of using binomial series to estimate $(u_i+v_i)^p$, e.g., \cite{WoodruffZ21b} uses the finite binomial series to estimate the difference for integer $p$ in the non-distributed streaming setting. 
However, a crucial drawback of binomial series is that we need to decide the dominant variable. 
This means that, if $u_i > v_i$, then we can only transform $(u_i+v_i)^p$ to $u_i^p \cdot (1+\frac{v_i}{u_i})^p$ and expand $(1+\frac{v_i}{u_i})^p$. We cannot do it the other way or the binomial series would be divergent. 
As we mentioned before, we do not have access to the real value of $v_i$, which leads to misclassification error that could cause our algorithm to fail if we use the binomial series.

On the other hand, notice that we can compute the actual value of $u_i$, so instead it is intuitive to apply our truncated Taylor estimator again. 
Let $f(v_i) = (u_i+v_i)^p - u_i^p$, where we consider $u_i$ as a known value. 
Then, we define the truncated Taylor estimator to be $((u_i+s_i)^p-u_i^p)+\sum_{q=1}^Q \binom{p}{q}(u_i+s_i)^{p-q} (v_i-s_i)^q$, with $s_i$ being a $(1+\frac{1}{100p})$-approximation of $v_i$.

With the above subroutines, we give unbiased estimates of the mixed $L_p$ sampling-based estimator $\widehat{s}=\frac{1}{p_i}((u_i+v_i)^p-u_i^p$.
Since our estimate has small variance, averaging across several independent instances provides a valid $F_p$ difference estimator.
Then, applying the difference estimator framework in \secref{sec:framework:tech}, we achieve an adversarially robust algorithm for distributed $F_p$ estimation.

In summary, our continuous perfect $L_p$ sampler appears to be an essential procedure in constructing the $F_p$ different estimator.
The mixed-$L_p$ sampler provides a reasonable distribution that captures the important coordinates in both $u$ and $v$ with respect to the difference.
In addition, the formula of the sampling probability allows efficient unbiased estimation to its inverse, so that we can average across independent instances to reduce the variance.

\subsection{Preliminaries}
For an integer $n>0$, we use the notation $[n]$ to denote the set $\{1,\ldots,n\}$. 
We use $\poly(n)$ to denote a fixed polynomial of $n$, whose degree can be determined by setting parameters in the algorithm accordingly. 
We use $\polylog(n)$ to denote $\poly(\log n)$. 
We use the notation $\tO{f}$ for a function $f(\cdot,\ldots,\cdot)$ to denote $f\cdot\polylog(n/\eps)$. 

If an event occurs with probability $1-1/\poly(n)$, then we say the event occurs with high probability. 
For a vector $f\in\mathbb{R}^n$, we define
\begin{align*}  \|f\|_p=\left(f_1^p+\ldots+f_n^p\right)^{1/p}.\end{align*}  

\paragraph{The coordinator model of distributed computing.} 
We work in the asynchronous coordinator model, which was introduced by~\cite{DolevF92}. 
There are $k$ servers, one of which acts as a special coordinator. 
Each server has access to private randomness and communicates with the coordinator over a private, bidirectional channel; no direct communication occurs between servers. 
Communication is asynchronous and message-based: messages on each channel must alternate between the coordinator and the server and must be self-delimiting, i.e., both parties must be able to determine when a message has ended.  
Each server $i \in [k]$ can determine from its transcript $\Pi_i$ with the coordinator whether it is their turn to speak, since the previous message must have come from the coordinator.  
When it is the coordinator’s turn, it may send messages to any subset of servers.  
At the end of the protocol, the coordinator must output an answer.

The communication cost of the protocol $\Pi=\{\Pi_1,\ldots,\Pi_k\}$ is the total number of bits sent across all channels coordinator outputs the answer. 
Thus we assume without loss of generality that the protocol is sequential and round-based, i.e., in each round the coordinator speaks to some number of players and await their responses before initiating the next round. 

Our protocol actually works in the more restricted setting where each vector $f(j)\in\mathbb{R}^n$ at server $j\in[k]$ arrives in a data stream and each server must use space sublinear in $n$ to store some sketch of $f(j)$. 
In this case, each element of the data stream consists of an update to a single coordinator of a single server, which will be specified only to the appropriate server. 
After the data stream is completed, the servers run the protocol $\Pi$ given only the sketches maintained by each server $j$, rather than given the entirety of $f(j)$. 

\paragraph{Exponential random variables. }
Exponential random variables are a key part of our protocols. 
Recall the following definition for an exponential random variable:
\begin{definition}[Exponential random variables]
If $\be$ is a exponential random variable with scale parameter $\lambda>0$, then the probability density function for $\be$ is 
\[p(x)=\lambda e^{-\lambda x}.\]
We say $\be$ is a standard exponential random variable if $\lambda=1$.
\end{definition}
We first state the following tail bounds for exponential random variables.
\begin{fact}
Let $\be$ be a standard exponential random variable. 
Then for any $a,b\ge 0$, 
\begin{align*}
\Pr[ \be \ge a\log n ] = \frac{1}{n^a},~~~ \Pr[ \be\le b ] \le b.
\end{align*}
\end{fact}
We next recall that the largest scaled item is a polylogarithmic heavy-hitter with high probability. 
\begin{fact}[c.f., Lemma 2.1 in \cite{EsfandiariKMWZ24}]
\factlab{fact:max:heavy}
Let $\be_1,\ldots,\be_n$ be standard exponential random variables, let $\alpha_1,\ldots,\alpha_n \ge 0$, and let $C>0$ be a fixed constant.  
Then
\begin{align*}
\Pr \Big[ \frac{\max_{i\in[n]}\alpha_i/\be_i}{\sum_{i=1}^n \alpha_i/\be_i}\ge\frac{1}{C\log^2 n} \Big] \ge1-\frac{1}{\poly(n)}.
\end{align*}
\end{fact}
We next recall that scaling an exponential random variable inversely scales its rate parameter:
\begin{fact}[Scaling of exponentials]
\factlab{fac:exp_scaling}
    Let $t$ be exponentially distributed with rate $\lambda$, and let $\alpha>0$. Then $\alpha t$ is exponentially distributed with rate $\lambda / \alpha$.
\end{fact}
We now define the anti-rank vector, which is a vector consisting of the order statistics.
\begin{definition}[Anti-rank vector]
Let $\left(t_1, \ldots, t_n\right)$ be independent exponentials. For $k=1,2, \ldots, n$, we define the $k$-th anti-rank $D(k) \in[n]$ of $\left(t_1, \ldots, t_n\right)$ to be the values $D(k)$ such that $t_{D(1)} \leq t_{D(2)} \leq \cdots \leq t_{D(n)}$.
\end{definition}

Using the structure of the anti-rank vector, \cite{Nag06} introduces a simple form for describing the distribution of $t_{D(k)}$ as a function of $\left(\lambda_1, \ldots, \lambda_n\right)$ and the anti-rank vector.

\begin{fact}[\cite{Nag06}]
\factlab{fac:exp_max_prob}
 For any $i=1,2, \ldots, n$, we have
\[\operatorname{Pr}[D(1)=i]=\frac{\lambda_i}{\sum_{j=1}^n \lambda_j}\]
\end{fact} 

\begin{fact}[\cite{Nag06}]
\factlab{fac:exp_order_statistics}
Let $\left(\be_1, \ldots, \be_n\right)$ be independently distributed exponentials, where $\be_i$ has rate $\lambda_i>0$. Then for any $k=1,2, \ldots, n$, we have
\[\be_{D(k)}=\sum_{i=1}^k \frac{E_i}{\sum_{j=i}^n \lambda_{D(j)}},\]
where $E_1, E_2, \ldots, E_n$ are i.i.d. exponential variables with mean $1$, and are independent of the anti-rank vector $(D(1), D(2), \ldots, D(n))$.
\end{fact} 
With the above results, we are able to describe the characterization to each entry in the scaled vector.
\begin{corollary}
\corlab{cor:exp_inverse}
Let $f \in \mathbb{R}^n$ be a frequency vector. Let $\left(\be_1, \ldots, \be_n\right)$ be i.i.d. exponential random variables with rate $1$. We define $z_i=f_i / \be_i^{1 / p}$. Let $(D(1), \ldots, D(n))$ be the anti-rank vector of the vector $\left(z_1^{-p}, \cdots z_n^{-p}\right)$, we have
\[\PPr{D(1)=i} = \frac{\left|f_i\right|^p}{\|f\|_p^p}\]
As a result, the probability that $\left|z_i\right|=\arg \max _j\left\{\left|z_j\right|\right\}$ is precisely $\left|f_i\right|^p /\|f\|_p^p$, so for a perfect $L_p$ sampler it suffices to return $i \in[n]$ with $\left|z_i\right|$ maximum.
Moreover, we have
\[z_{D(k)}=\left(\sum_{i=1}^k \frac{E_i}{\sum_{j=i}^N f_{D(j)}^p}\right)^{-1 / p},\]
where $E_i$ 's are i.i.d. exponential random variables with mean $1$, and are independent of the anti-rank vector $(D(1), \ldots, D(n))$. Specifically, we have
\[\max_{i\in[n]} \frac{f_i}{\be_i} = z_{D(1)} = \frac{\|f\|_p}{E^{1/p}}.\]
\end{corollary}
\begin{proof}
By \factref{fac:exp_scaling}, the variable $z_i^{-p}=\be_i /f_i^p$ is exponentially distributed with rate $\lambda_i=f_i^p$. By \factref{fac:exp_max_prob}, we have, 
\[\PPr{D(1)=i}=\PPr{i=\arg \min \left\{\left|z_1\right|^{-p}, \ldots,\left|z_n\right|^{-p}\right\}}=\frac{\left|f_i\right|^p}{\|f\|_p^p}.\] 
Moreover, by \factref{fac:exp_order_statistics}, we have,
\[z_{D(k)}=\left(\sum_{i=1}^k \frac{E_i}{\sum_{j=i}^N \lambda_{D(j)}}\right)^{-1 / p}=\left(\sum_{i=1}^k \frac{E_i}{\sum_{j=i}^N f_{D(j)}^p}\right)^{-1 / p},\]
where $E_i$ 's are i.i.d. exponential random variables with mean $1$, and are independent of the anti-rank vector $(D(1), \ldots, D(n))$.
\end{proof}

We state the anti-concentration property of exponential random variables, which means that there is a gap between the first max and the second max of the scaled vector.
\begin{lemma}[Anti-concentration]
\lemlab{lem:anti}
Let $f \in \mathbb{R}^n$ be a frequency vector. 
Let $\left(\be_1, \ldots, \be_n\right)$ be i.i.d. exponential random variables with rate $1$. We define $z_i=f_i / \be_i^{1 / p}$. Let $(D(1), \ldots, D(n))$ be the anti-rank vector of the vector $\left(z_1^{-p}, \cdots z_n^{-p}\right)$. For parameter $\eps \in (0,1]$ we have
\[\PPr{\frac{z_{D(1)}}{z_{D(2)}} \le 1+\eps}=\O{\eps}.\]
\end{lemma}
\begin{proof}
By \corref{cor:exp_inverse}, let $F := \|f\|_p^p$ we have $z_{D(1)}^p = \frac{F}{\be_1}$ and $z_{D(2)}^p = \left(\frac{\be_1}{F} + \frac{\be_2}{F-f_{D(1)}^p}\right)^{-1}$, where $\be_1$ and $\be_2$ are independent variables. Then, we have
\[z_{D(2)}^p = \left(\frac{\be_1}{F} + \frac{\be_2}{F-f_{D(1)}^p}\right)^{-1} \le \left(\frac{\be_1}{F} + \frac{\be_2}{F}\right)^{-1} = \frac{F}{\be_1+\be_2}.\]
Thus, we have
\[\PPr{\frac{z_{D(1)}^p}{z_{D(2)}^p} \le 1+\eps} \le \PPr{\be_2 \le \eps \cdot \be_1}.\]
Now, due to the pdf of exponential variables, we have
\begin{align*}
\PPr{\be_2 \le \eps \cdot \be_1} &= \int_0^\infty e^{-s} \int_0^{\eps s} e^{-t} dtds \\
&= \int_0^\infty e^{-s} (1 - e^{-\eps s}) ds \\
&= -e^{-s} ~\Big|_0^\infty + \frac{1}{1+\eps} e^{-(1+\eps)s} ~\Big|_0^\infty \le \eps.
\end{align*}
Therefore, we have,
\[\PPr{\frac{z_{D(1)}^p}{z_{D(2)}^p} \le 1+\eps} \le \eps.\]
For parameter $\eps \in (0,1]$ and constant $p$, we have
\[\PPr{\frac{z_{D(1)}}{z_{D(2)}} \le 1+\eps} \le \O{\eps},\]
as desired.
\end{proof}

Next, we prove that the geometric mean of three exponential random variables has bounded variance, which will be used in our geometric mean estimator.
\begin{lemma}[Geometric mean]
\lemlab{lem:geo:exp:var}
\label{lem:geo:exp:var}
Let $\be_1,\be_2,\be_3$ be independent exponential random variables with rate $1$ and let $X_1=\frac{1}{\be_1}$, $X_2=\frac{1}{\be_2}$, and $X_3=\frac{1}{\be_3}$. 
Let $X=\sqrt[3]{X_1\cdot X_2\cdot X_3}$. 
Then $\Ex{X}=C_{\ref{lem:geo:exp:var}}$ for some constant $C_{\ref{lem:geo:exp:var}}\in(0,8]$ and $\Ex{X^2}\le 27$.
\end{lemma}
\begin{proof}
Since $\be_1,\be_2,\be_3\ge 0$, then we trivially have $\Ex{X}\ge 0$ and $\Ex{X^2}\ge 0$. 
Since $X_1$, $X_2$, and $X_3$ are independent random variables, then we have 
\begin{align*}  \Ex{X}=\Ex{X_1^{1/3}}\cdot\Ex{X_2^{1/3}}\cdot\Ex{X_3^{1/3}}.\end{align*}  
For $i\in\{1,2,3\}$, we have
\begin{align*}  \Ex{X_i^{1/3}}=\int_0^\infty\PPr{X_i>t^3}\,dt.\end{align*}  
As $\be_i$ is an exponential random variable, then the probability density function of $X_i$ is $p(t)=\frac{1}{t^2}\cdot e^{-\frac{1}{t}}$. 
Hence, we have $\PPr{X_i>t}\le\frac{1}{t}$ so that 
\begin{align*}
\Ex{X_i^{1/3}}&=\int_0^\infty\PPr{X>t^3}\,dt\\
&=\int_0^1\PPr{X_i>t^3}\,dt+\int_1^\infty\PPr{X_i>t^3}\,dt
\end{align*}
Now, for the first term, it satisfies
$\int_0^1\PPr{X_i>t^3}\,dt \le 1$.
Moreover, for the second term, we have
\begin{align*}
\int_1^\infty\PPr{X_i>t^3}\,dt
\le \int_1^\infty\frac{1}{t^3}\,dt < 1,
\end{align*}
and so $\Ex{X}\le 8$. 
Similarly, we have
\[\Ex{X^2}=\Ex{X_1^{2/3}}\cdot\Ex{X_2^{2/3}}\cdot\Ex{X_3^{2/3}}.\]
Moreover, for $i \in \{1,2,3\}$, we have
\begin{align*}
\Ex{X_i^{2/3}}&=\int_0^\infty\PPr{X_i>t^{3/2}}\,dt\\
&=\int_0^1\PPr{X_i>t^{3/2}}\,dt+\int_1^\infty\PPr{X_i>t^{3/2}}\,dt
\end{align*}
Again, for the first term, it satisfies
$\int_0^1\PPr{X_i>t^3}\,dt \le 1$.
Moreover, for the second term, we have
\begin{align*}
\int_1^\infty\PPr{X_i>t^{3/2}}\,dt \le \int_1^\infty\frac{1}{t^{3/2}}\,dt \le 3,
\end{align*}
and so $\Ex{X^2}\le 27$. 
\end{proof}

We next introduce an auxiliary lemma that bounds the expectation of the inverse exponential variable $\frac{1}{\be}$ after truncating its tail.
\begin{lemma}
\lemlab{lem:truncated:inverse:ex}
Let $\be$ be an exponential variable with rate $1$. We define event $\calE$ as $\frac{1}{\poly(n)} < \be < \poly(n)$. Then, we have
\[\Ex{\frac{1}{\be_i} ~|~ \calE} = \O{\log n}.\] 
\end{lemma}
\begin{proof}
Due to the probability density function of exponential variables, we have
\[\Ex{\frac{1}{\be_i} ~|~ \calE} = \int_{\frac{1}{\poly(n)}}^{\poly(n)} \PPr{\frac{1}{\be_i} > t} \mathrm{d} t 
\le \int_{\frac{1}{\poly(n)}}^{\poly(n)} \frac{1}{t} \mathrm{d} t
= \O{\log n}.\] 
\end{proof}
We note that the event $\calE$ in the lemma statement of \lemref{lem:truncated:inverse:ex} happens with probability $\frac{1}{\poly(n)}$ due to the pdf of exponential variables.
Since we generate at most $\poly(n)$ independent exponential variables in all our algorithms, in this paper we condition on that $\calE$ holds for each exponential variable we generate, which still happens with high probability by a union bound.

\section{Perfect \texorpdfstring{$L_p$}{Lp} Sampler}
In this section, we present our construction of the perfect $L_p$ sampler in the distributed monitoring setting.
Our $L_p$ samplers works by first scaling the vector $f$ by exponential random variables to obtain a scaled vector $g = [f_1/\be_1^{1/p},\ldots, f_n/\be_n^{1/p}]$. 
We then aim to find the maximum index $i^*$ of $g$. 
\secref{sec:pl:recovery} introduces a crucial subroutine that identifies a small set of coordinates $\SetPL$ that are $\frac{1}{\polylog(n)}$-heavy in $g$, so that we can track each item in thte distributed data stream.
In \secref{sec:truncated:taylor}, we produce a truncated Taylor estimator such that it obtains unbiased estimates of $f_i^p$, which is applied in our $L_p$ sampler and our $F_p$ estimation algorithms.
Lastly, in \secref{sec:sampler:construct}, we show the algorithm of the sampler.

\subsection{Identification of Candidate Maxima}
\seclab{sec:pl:recovery}
In this section, we assume the input to be a scaled vector $g\in\mathbb{R}^n$ distributed across $k$ servers. 
Specifically, we assume that for the index $i^*=\argmax_{i\in[n]} g_i$, we have $g_{i^*}^p\ge\frac{1}{\polylog(n)}\cdot\|g\|_p^p$ and our goal is to identify a small subset $\SetPL\subseteq[n]$ of indices that contain $i^*$, while using $k^{p-1}\cdot\polylog(n)$ total communication. 
The main intuition is that since $g_{i^*}$ is sufficiently large, then the extremal cases are either there exists a server $j\in[k]$ so that $g_{i^*}(j)$ is large or there exist many servers $j\in[k]$ with a non-trivial contribution to $g_{i^*}$. 
In the first case, we will sample the server $j$ containing the large value of $g_{i^*}(j)$ and in the second case, we will sample a sufficient number of servers with a non-trivial contribution to $g_{i^*}$ and then we can use these servers to estimate $g_{i^*}$ using concentration. 
More generally, anywhere in the spectrum between the two extremes can occur, so we formalize this intuition through the notion of level sets. 
In particular, for each fixed index $i\in[n]$, we partition the servers $[k]$ based on their contribution $g_i(j)$ with respect to the total value $g_i$. 
\begin{definition} 
\deflab{def:level:set}
For each index $i \in [n]$, server $j \in [k]$, we define the rescaled vector $g(j)$ to be $g_i(j) = \frac{f_i(j)}{\be_i^{1/p}}$. 
We define $U := \sum_{i \in [n]} \sum_{j \in [k]} g_i(j)^p$ and $G := \|g\|_p^p = \sum_{i\in[n]} g_i^p$.
We partition the servers into groups $\Lambda_b$, where $b \ge 0$ is an integer, so that 
\begin{align*}  \Lambda_b=\left\{j\in[k]\,:\,\frac{g_i(j)}{U^{1/p}}\in\left(\frac{1}{2^{b}},\frac{1}{2^{b-1}}\right]\right\}.\end{align*}  
Note that for each coordinate $i\in[n]$, we can decompose
\begin{align*}  g_i=\sum_{j\in[k]}\frac{f_i(j)}{e_i^{1/p}}=\sum_a\sum_{j\in \Lambda_a}\frac{f_i(j)}{e_i^{1/p}}.\end{align*}  
Then we define the contribution $C_i(b)$ of the group $\Lambda_b$ toward the frequency of an item $i\in[n]$ is
\begin{align*}  C_i(b)=\sum_{j\in\Lambda_b}g_i(j),\end{align*}  
so that $g_i=\sum_b C_i(b)$.
\end{definition}

We provide the algorithm construction in \algref{alg:est:server} and \algref{alg:est:coor}, which estimates the contribution of each level set in \defref{def:level:set}.

\begin{algorithm}[!htb]
\caption{Algorithm to estimate each scaled coordinate: Server}
\alglab{alg:est:server}
\begin{algorithmic}[1]
\Require{Input scaled vector $g(j)\in\mathbb{R}^n$ that arrives at server $j$ as a data stream}
\Ensure{Estimates of the frequencies of all coordinates}
\State{$R\gets \O{\log n}$}
\State{Let $C$ be the constant in the max-stability inequality} \Comment{See \factref{fact:max:heavy}}
\State{Let $\gamma = C_{\ref{lem:approx:cts}} C^2 k^{p-1}\poly\left(\frac{\log n}{\eps}\right)$ be the scaling factor for the sampling probability, which is a power of $2$.}
\If{the coordinator starts a new round}
\State{Receive an $(1+\frac{1}{100\cdot 2^p})$-approximation $\widehat{U}$ to $U=\sum_{j\in[k]}\sum_{i\in[n]}(g_i(j))^p$ from the coordinator}
\State{Generate $R$ uniform random variables $u_i^{(r)}(j) \in [0,1]$ for each $g_i(j)$}
\EndIf
\For{each time an item $i$ arrives at server $j$}
\For{$b\in [\O{\log n}]$, $r\in[R]$}
\If{$\frac{(g_i(j))^p}{\widehat{U}}\ge\frac{\eps^2}{C_{\ref{lem:approx:cts}} \cdot C^2 k^{p-1}\polylog(n)}$}
\State{Send $g_i(j)$ to the central server}
\State{Update its value once it increases by a $\frac{\eps}{100}$ fraction}
\ElsIf{$\frac{(g_i(j))^p}{\widehat{U}}\cdot \gamma\in\left[\frac{1}{2^{bp}},\frac{1}{2^{(b-1)p}}\right)$}
\State{Send $g_i(j)$ to the central coordinator if $u_i^{(r)}(j) < \frac{1}{2^{bp}}$}
\If{$g_i(j)$ has been sent and it increases by a $\frac{\eps}{100\log n}$-fraction}
\State{Inform the coordinator to delete the sample}
\State{Sample $g_i(j)$ with a newly-generated uniform variable $u_i^{(r)}(j)$}
\EndIf
\EndIf
\EndFor
\EndFor
\end{algorithmic}
\end{algorithm}

\begin{algorithm}[!htb]
\caption{Algorithm to estimate the each scaled coordinate: Coordinator}
\alglab{alg:est:coor}
\begin{algorithmic}[1]
\Require{Input vectors $g(1),\ldots,g(k)\in\mathbb{R}^n$ that arrive as a data stream, accuracy parameter $\eps\in(0,1)$}
\Ensure{Estimations of the frequencies of all coordinates}
\State{$R \gets \O{\log n}$}
\State{Obtain an $(1+\frac{1}{100\cdot 2^p})$-approximation $\widehat{U}$ to $U=\sum_{j\in[k]}\sum_{i\in[n]}(g_i(j))^p$}
\State{Pick $\xi \in [\frac{1}{2},1)$ uniformly at random}
\If{$\widehat{U}$ doubles}
\State{Delete all samples and send $\widehat{U}$ to all servers}
\EndIf
\State{$\widehat{C^{(r)}_i(b)}\gets 0$ for all $r,i,b$}
\For{$i\in[n]$}
\If{$\frac{(g_i(j))^p}{\widehat{U}}\ge\frac{\eps^3}{C_{\ref{lem:approx:cts}} \cdot C^2 k^{p-1}\polylog(n)}$}
\State{$\widehat{C_i(0)} \gets \widehat{C_i(0)} + \widehat{g_i(j)}$}
\EndIf
\For{$b\in [\O{\log n}]$, $r\in[R]$}
\If{$\frac{(g^{*,(r)}_i(j))^p}{\widehat{U}}\cdot \gamma\in\left[\frac{\xi^p}{2^{bp}},\frac{\xi^p}{2^{(b-1)p}}\right)$} 
\Comment{$g_i^{*,(r)}(j)$ is the value of $g_i(j)$ when it is sampled in the $r$-th repetition}
\If{less than $\frac{C_{\ref{lem:approx:cts}} \cdot C^2 \cdot \log^5n}{2^p \eps^2}$ samples are in level $b$ in the $r$-th repetition}
\State{Skip the level and set $\widehat{C_i(b)}\gets 0$} \label{line:skip}
\Else
\State{$\widehat{C^{(r)}_i(b)}\gets\widehat{C^{(r)}_i(b)}+2^{bp}\cdot\widehat{g^{(r)}_i(j)}$}
\EndIf
\State{$\widehat{C_i(b)}\gets\median_{r\in[R]}\widehat{C^{(r)}_i(b)}$}
\EndIf
\EndFor
\State{$\widehat{g_i}\gets\frac{1}{R}\sum_{r \in [R]}\sum_{b\in\mathbb{Z}}\widehat{C_i^{(r)}(b)}$}
\EndFor
\State{\Return $\widehat{g_1},\ldots,\widehat{g_n}$}
\end{algorithmic}
\end{algorithm}

Now, we analyze the correctness and communication complexity of our algorithm. 
We first show that the maximizer of the scaled vector is roughly $\frac{1}{\log^2 n}$-heavy, which is given by the maximum stability of exponential random variables.

\begin{lemma}
\lemlab{lem:G:max_stability}
Let $G = \sum_{i \in [n]} g_i^p$, we have
\[ \PPr{ \frac{G^{1/p}}{C \log^2n} \le \max_{i \in [n]} g_i } \ge 1 - \frac{1}{\poly(n)}.\]
\end{lemma}
\begin{proof}
Recall the max-stability property of exponential random variables (see \factref{fact:max:heavy}). Notice we define $G = \sum_{i \in [n]} g_i^p = \sum_{i \in [n]} \frac{f_i^p}{e_i}$, thus we have
\[ \PPr{ \frac{G}{C \log^2n} \le \max_{i \in [n]} g_i^p} \ge 1 - \frac{1}{\poly(n)}.\]
\end{proof}

We next show that we indeed sample each coordinate $g_i(j)$ with the correct probability. 
\begin{lemma}
\lemlab{lem:sample:prob}
In \algref{alg:est:server}, for any fixed time $t$ in the stream, all coordinates $g_i(j)$ that satisfy $\frac{(g_i(j))^p}{\widehat{U}}\cdot \gamma\in\left[\frac{1}{2^{bp}},\frac{1}{2^{(b-1)p}}\right)$ at $t$ are sampled with probability $\frac{1}{2^{bp}}$.
\end{lemma}
\begin{proof}
Consider a fixed time $t$, let $t'$ be the time we generate the last uniform variable $u_i^{(r)}(j)$.
Let $g_{i,t}(j)$ denote the value of $g_i(j)$ at time $t$.
The sampling probability of $g_i(j)$ at $t$ equals to $\PPr{u_i^{(r)}(j) < \frac{1}{2^{b'p}}} = \frac{1}{2^{b'p}}$, where $b'$ is the minimum positive integer such that $\frac{1}{2^{b'p}} \le \max_{\tau \in[t',t]} \frac{(g_{i,\tau}(j))^p}{\widehat{U}}\cdot \gamma$.
Notice that our stream is insertion-only, then we have
\[\max_{\tau \in[t',t]} \frac{(g_{i,\tau}(j))^p}{\widehat{U}}\cdot \gamma = \frac{(g_{i,t}(j))^p}{\widehat{U}}\cdot \gamma.\]
Therefore, for $\frac{(g_{i,t}(j))^p}{\widehat{U}}\cdot \gamma\in\left[\frac{1}{2^{bp}},\frac{1}{2^{(b-1)p}}\right)$, we have $b'=b$, which gives the correct sampling probability.
\end{proof}
After sampling a number of coordinates using \algref{alg:est:server}, the next step is for \algref{alg:est:coor} to use the samples to provide meaningful estimates to $g_i$. 
Namely, if $g_i$ is large, then the estimate for $g_i$ should be accurate. 
Conversely, if $g_i$ is small, then the estimate for $g_i$ does not need to be accurate, but it should not be large. 
The next statement reflects this property of our estimate using a level set argument. 
We remark that although the proof is involved, the techniques are rather standard, as similar level arguments have been commonly used in literature for sublinear algorithms, e.g., \cite{IndykW05,WoodruffZ12,BlasiokBCKY17,WoodruffZ21a,JayaramWZ24}. 

We note that our \lemref{lem:est:cts:p:small} can also be obtained using an extension of the $L_p$-heavy-hitter algorithm from recent work \cite{HuangXZW25} for $p\ge2$. 
For $1<p<2$, an $L_2$-heavy-hitter algorithm with $\O{k}$ communication suffices to give a constant-factor approximation, which is sufficient in the construction of our perfect $L_p$-sampler. 
However, our result improves the dependence on $k$ for $\eps$-heavy-hitters.
In addition, a modified version of our algorithm with refined analysis requires only $\O{k+sk^{p-1}}$ communication to estimate $s$ scaled vectors, rather than $\O{sk}$ using the $L_2$-heavy-hitter algorithm, e.g., as a subroutine to acquire $s$ $L_p$-samples. 
This is by replacing our $\widehat{U}$ by a $2$-approximation of the $L_1$-of-$L_p$ sum of the original vector: $\sum_{j\in[k]}\sum_{i\in[n]}(f_i(j))^p$.
Thus, we keep the result here for completeness.

\begin{lemma}
\lemlab{lem:est:cts:p:small}
\label{lem:approx:cts} 
For the choice of constant $C_{\ref{lem:approx:cts}}$, the output of \algref{alg:est:coor} satisfies 
\begin{enumerate}
    \item Let $i = \argmax_{\bi \in [n]} g_\bi$, the estimate $\widehat{g_i}$ satisfies $|\widehat{g_i} - g_i| \le \eps \cdot g_i$.
    \item For all other $i' \neq i$, the estimate $\widehat{g_{i'}}$ satisfies $\widehat{g_{i'}} \le g_{i'}+ \eps g_i$.
\end{enumerate}
That is, taking $\eps = \frac{1}{4}$, \algref{alg:est:coor} finds a set $\SetPL$ of coordinates that includes the maximizer, and only contains the items that are at least $\frac{1}{4}$-fraction of the maximizer.
\end{lemma}
\begin{proof}
We fix a time $t$ in the distributed data stream, and we consider a fix instance $r$.

\paragraph{Level $0$.}
In the case where $\frac{(g_i(j))^p}{\widehat{U}}\ge\frac{\eps^3}{C_{\ref{lem:approx:cts}} \cdot C^2 k^{p-1}\polylog(n)}$, $g_i(j)$ is sent to the coordinator and is never deleted in a round. 
Since the server updates its value whenever it increases by a $\frac{\eps}{100}s$-fraction, $\widehat{g_i(j)}$ is a $(1+\frac{\eps}{100})$-approximation to $g_i(j)$.
Thus, we have
\[\left|\widehat{C_i^{(r)}(0)} - C_i^{(r)}(0)\right| \le \eps \cdot C_i^{(r)}(0).\]

\paragraph{Idealized setting.}
Now, we analyze our sampling and re-construction procedure
We first consider an idealized setting where each of the estimates $\widehat{g_i^{(r)}(j)}$ are categorized into the correct level set.
Let $\nu\in[1,2]$ be a parameter such that $\frac{\nu G^{1/p}}{\xi U^{1/p}} = 2^m$, where $m$ is some integer. 
Notice that $G \ge U$, so we have $\frac{1}{2^m} = \O{1}$.
For each integer $b \in \mathbb{Z}$, we define the level set 
\begin{align*}  \Gamma_b=\left\{j\in[k]\,:\,\frac{g_i(j)}{G^{1/p}}\in\left(\frac{\nu}{2^{b}},\frac{\nu}{2^{b-1}}\right]\right\}.\end{align*} Recall that we define the reconstruction level sets in \algref{alg:est:coor} as follows,
\begin{align*}  \Lambda_a=\left\{j\in[k]\,:\,\frac{g_i(j)}{U^{1/p}}\in\left(\frac{\xi}{2^{a}},\frac{\xi}{2^{a-1}}\right]\right\}.\end{align*}  
Then, we have a one-to-one mapping between each pair of these level sets: 
\[ \Gamma_{b} = \Lambda_{b-m}. \]
Next, we consider casework on whether $\frac{(g_i(j))^p}{\widehat{U}}\cdot \gamma \ge \frac{C_{\ref{lem:approx:cts}} \cdot C^2 \cdot \polylog(n)}{2^p \eps^2 k}$ or whether $\frac{(g_i(j))^p}{\widehat{U}}\cdot \gamma < \frac{C_{\ref{lem:approx:cts}} \cdot C^2 \cdot \polylog(n)}{2^p \eps^2 k}$.
In short, we show that in our algorithm, if a level has contribution less than a roughly $\frac{\eps}{\log^2 n}$ fraction, then it is in the first case and we truncate those levels since they are negligible.
Otherwise, the level is the second case $2$ and the variance is small.

\paragraph{Idealized setting, large $b$.}
First, we consider a level $b$ where $g_i(j) \in \Gamma_b$ satisfies $\frac{(g_i(j))^p}{\widehat{U}}\cdot \gamma < \frac{C_{\ref{lem:approx:cts}} \cdot C^2 \cdot \log^5n}{2^p \eps^2 k}$.
Since there are $k$ servers in total, the number of samples at this level is at most $\frac{C_{\ref{lem:approx:cts}} \cdot C^2 \cdot \log^5n}{2^p \eps^2}$ with high probability.
Then, by our algorithm construction, we set $\widehat{C_i(b)} = 0$, i.e., the variance is $0$.

\paragraph{Idealized setting, small $b$.}
Second, we consider the case where $\frac{(g_i(j))^p}{\widehat{U}}\cdot \gamma \ge \frac{C_{\ref{lem:approx:cts}} \cdot C^2 \cdot \log^5n}{2^p \eps^2 k}$, that is, the sampling probability $p_b \ge \frac{C_{\ref{lem:approx:cts}} \cdot C^2 \cdot \log^5n}{\eps^2 k}$.
Consider the level set $\Gamma_b = \Lambda_{b-m}$, where $\frac{(g_i(j))^p}{U} = \Theta(1) \cdot \frac{1}{2^{bp-mp}}$.
In an idealized setting, since $\widehat{g_i(j)}$ is within $(1+\frac{\eps}{100})$-fraction of the actual value $g_i(j)$, we have
\[ \Ex{\widehat{C_i^{(r)}(b)}}= \frac{1}{p_b}\cdot\sum_{j\in\Gamma_b}p_b\cdot\widehat{g_i^{(r)}(j)}= (1\pm \eps) \cdot C_i(b).\] 
In addition, we have
\begin{align*}
\Var{\widehat{C_i^{(r)}(b)}}&\le\frac{1}{p_b^2}\cdot\sum_{j\in\Gamma_b}p_b\cdot \O{1} \cdot\left(g_i(j)\right)^2\\
&\le\frac{1}{p_b}\cdot|\Gamma_b|\cdot\O{1}\cdot\left(\frac{G^{1/p}}{2^b}\right)^2,
\end{align*} 
Since we sample all items in $\Gamma_b$ with probability $p_b = \Theta(1) \cdot \frac{(g_i(j))^p}{U}\cdot\gamma$ with $\gamma = \frac{C_{\ref{lem:approx:cts}} \cdot C^2k^{p-1} \polylog(n)}{\eps^2}$, we have
\[\Var{\widehat{C_i^{(r)}(b)}} \le  2^{bp - mp} \cdot\frac{\eps^2}{C_{\ref{lem:approx:cts}} \cdot C^2k^{p-1}\log^5n} \cdot |\Gamma_b|\cdot\O{1}\cdot\left(\frac{G^{1/p}}{2^b}\right)^2.\]
Notice that $\frac{1}{2^m} = \O{1}$ as mentioned earlier, so we have
\[\Var{\widehat{C_i^{(r)}(b)}} \le 2^{bp} \cdot\frac{\eps^2}{C_{\ref{lem:approx:cts}} \cdot C^2k^{p-1}\log^5n} \cdot |\Gamma_b|\cdot\O{1}\cdot\left(\frac{G^{1/p}}{2^b}\right)^2.\]
Since there are at most $\O{2^b}$ items in set $\Gamma_b$, so $|\Gamma_b| = \min \{\O{2^b}, k \}$. For $\O{2^b} < k$, we have
\begin{align*}
\Var{\widehat{C_i^{(r)}(b)}}&\le2^{bp} \cdot\frac{\eps^2}{C_{\ref{lem:approx:cts}}\cdot C^2 k^{p-1}\log^5n} \cdot 2^b\cdot\O{1}\cdot\left(\frac{G^{1/p}}{2^b}\right)^2\\
&=2^{b(p-1)} \cdot\frac{\eps^2}{C_{\ref{lem:approx:cts}}\cdot C^2 k^{p-1}\log^5n}\cdot\O{1}\cdot G^{2/p}.
\end{align*}
Since $p>1$, we have $2^{b(p-1)} < k^{p-1}$, so that
\begin{align*}
\Var{\widehat{C_i^{(r)}(b)}}&\le\O{1}\cdot k^{p-1} \cdot\frac{\eps^2}{C_{\ref{lem:approx:cts}} \cdot C^2 k^{p-1}\log^5n}\cdot G^{2/p}\\
&=\O{\eps^2}\cdot \frac{G^{2/p}}{C_{\ref{lem:approx:cts}} \cdot C^2\log^5n}.
\end{align*}
On the other hand, if $\O{2^b} \ge k$, since in this case, we have $\frac{1}{p_b} \le \frac{\eps^2 k}{C_{\ref{lem:approx:cts}} \cdot C^2 \cdot \log^5n}$, so that 
\begin{align*}
\Var{\widehat{C_i^{(r)}(b)}}& \le \frac{1}{p_b}\cdot|\Gamma_b|\cdot\O{1}\cdot\left(\frac{G^{1/p}}{2^b}\right)^2\\
&\le \frac{\eps^2 \cdot k}{ C_{\ref{lem:approx:cts}} \cdot C^2 \log^5n} \cdot k\cdot\O{1}\cdot \frac{G^{2/p}}{k^2} \\
&= \frac{\eps^2\cdot G^{2/p}}{C_{\ref{lem:approx:cts}} \cdot C^2 \log^5n}.
\end{align*}
Therefore, in all cases, we have
\[\Var{\widehat{C_i^{(r)}(b)}}\le\O{\eps^2}\cdot \frac{G^{2/p}}{C_{\ref{lem:approx:cts}} \cdot C^2\log^5n}.\]
Applying \lemref{lem:G:max_stability}, with high probability we have,
\[\Var{\widehat{C_i^{(r)}(b)}}\le\O{\frac{\eps^2}{C_{\ref{lem:approx:cts}}\log n}}\cdot \left(\max_{i \in [n]} g_i\right)^2.\]
Then, by Chebyshev's inequality, we have
\[\PPr{\left|\widehat{C_i^{(r)}(b)} -C_i(b)\right| \le \O{\frac{\eps}{C_{\ref{lem:approx:cts}}\log n}\cdot \max_{i \in [n]} g_i}} \ge \frac{2}{3}.\]
Then, by taking the median of $R = \O{\log(n)}$ instances, we have by standard Chernoff bounds
\[\PPr{\left|\widehat{C_i(b)} -C_i(b)\right| \le \O{\frac{\eps}{C_{\ref{lem:approx:cts}}\log n}\cdot \max_{i \in [n]} g_i}} \ge \frac{1}{\poly(n)}.\]
Then, by union bounding over the $\O{\log n}$ level sets $\Gamma_b$, we have that with high probability,
\[\left|\widehat{C_i(b)} -C_i(b)\right| \le \O{\frac{\eps}{C_{\ref{lem:approx:cts}}\log n}\cdot \max_{i \in [n]} g_i},\]
for all levels $b$ in this case.
Thus, summing up the estimate error for each $C_i(b)$ (both large $b$ and small $b$), with high probability we have
\[\widehat{g_i} \le g_i + \O{\frac{\eps}{C_{\ref{lem:approx:cts}}}}\cdot g_i.\]

\paragraph{Idealized setting, bounding the bias}
The above analysis shows that, in the idealized setting, the estimate $\widehat{g_i}$ satisfies $\widehat{g_i} \le g_i + \O{\frac{\eps}{C_{\ref{lem:approx:cts}}}}\cdot g_i$ for all $i \in [n]$ with high probability.
Next, let $i$ be the maximum index, we show that the truncation of $g_i(j)$ samples with small values at most adds up a $\eps$-additive estimate error.

First, we notice that the contribution of level $b$ with $C_i(b) = \frac{\eps \cdot G^{1/p}}{\log^3n}$ is negligible, since there are at most $\O{\log n}$ levels, and the sum of such levels is at most $\frac{\eps \cdot G^{1/p}}{\log^2n}$. 
So, we only consider the case where $C_i(b) \ge \frac{\eps \cdot G^{1/p}}{\log^3n}$.
Suppose that $2^{b+1} > \frac{k \log^3 n}{\eps}$, then we have
\[C_i(b)=\sum_{j \in \Gamma_b} g_i(j) \le k \cdot \frac{1}{2^{b+1}} \le \frac{\eps \cdot G^{1/p}}{\log^3n},\]
which is negligible.
Now, for $2^{b+1} \le \frac{k \log^3 n}{\eps}$, our sampling probability is at least
\begin{align*}\frac{(g_i(j))^p}{\widehat{U}}\cdot \gamma &\ge \frac{(g_i(j))^p}{G}\cdot C_{\ref{lem:approx:cts}} C^2 k^{p-1}\polylog(n) \\
& \ge \O{1} \cdot \frac{ C_{\ref{lem:approx:cts}} C^2k^{p-1}\poly\left(\frac{\log n}{\eps}\right)}{2^{bp}}.
\end{align*}
Then, since we have
\[2^{bp} = 2^{b(p-1)} \cdot 2^b \le 2^b k^{p-1} \poly\left(\frac{\log n}{\eps}\right),\]
the sampling probability is at least 
\[ \O{\frac{1}{2^b}} \cdot  C_{\ref{lem:approx:cts}} C^2\poly\left(\frac{\log n}{\eps}\right).\]
Moreover, for a level $b$ to be non-negligible, it at least contains $\O{\frac{\eps 2^b}{\log ^3 n}}$ elements.
Then, with high probability, we have at least $\frac{C_{\ref{lem:approx:cts}} \cdot C^2 \cdot \log^5n}{2^p \eps^2}$ samples for sufficiently large exponents in $\poly\left(\frac{\log n}{\eps}\right)$, so that the algorithm does not truncate those levels.
Lastly, if some level set contains less than $\frac{C_{\ref{lem:approx:cts}} \cdot C^2 \cdot \log^5n}{2^p \eps^2}$ elements, e.g., $\O{\frac{\eps 2^b}{\log ^3 n}} \le \frac{C_{\ref{lem:approx:cts}} \cdot C^2 \cdot \log^5n}{2^p \eps^2}$, then each of them satisfies $g_i(j) > \O{\frac{G^{1/p}\eps^3}{C_{\ref{lem:approx:cts}} \cdot C^2\log^8 n}}$, so our algorithm sends these elements to the coordinator.

Therefore, for all non-negligible level sets, the additive estimate error is at most $\O{\frac{\eps}{C_{\ref{lem:approx:cts}} \log n}}\cdot g_i$ based on previous discussions, and so $|\widehat{g_i} -g_i| \le \O{\frac{\eps}{C_{\ref{lem:approx:cts}}}} \cdot g_i$ for the maximum index $i$. 

\paragraph{Randomized boundaries.}
Unfortunately, we are not in an idealized setting, and thus there exist level sets $b$ and servers $j$ such that $\frac{g_i(j)}{G^{1/p}} \in \Gamma_b$, but $\frac{\widehat{g_i(j)}}{G^{1/p}} \notin \Gamma_b$.
In this case, we say that $g_i(j)$ is misclassified.
Due to our algorithm construction, we have $|\widehat{g_i(j)} - g_i(j)| \le \frac{\eps}{100} \cdot g_i(j)$ for all sampled items $g_i(j)$.
Since $\xi \in\left[\frac{1}{2}, 1\right)$ is chosen uniformly at random, then the probability any coordinate is misclassified is $\mathcal{O}\left(\frac{\varepsilon}{100}\right)$.
Observe that when an item is misclassified, its contribution is rescaled by the incorrect probability, but within a factor of two. 
Thus in expectation, the error due to the misclassification is at most $\mathcal{O}\left(\frac{\varepsilon}{100}\right) \cdot g_i$. Therefore, for those levels with $\widehat{c_i(b)} \neq 0$, by triangle inequality, we have that
\[\operatorname{Pr}\left[\left|\widehat{C_i^{(r)}(b)}-C_i(b)\right| \leq \O{\frac{\eps}{C_{\ref{lem:approx:cts}}}} \cdot g_i\right] \geq \frac{2}{3}.\]
Using a similar analysis as previous sections, we have that
\[\widehat{g_i} \le g_i + \O{\frac{\eps}{C_{\ref{lem:approx:cts}}}}\cdot g_\bi,\]
for all $i \in [n]$, and 
\[|\widehat{g_\bi} - g_\bi| \le \O{\frac{\eps}{C_{\ref{lem:approx:cts}}}}\cdot g_\bi,\]
for the maximum index $\bi$, which proves the desired statement by choosing a suitable $C_{\ref{lem:approx:cts}}$.
\end{proof}

We next upper bound the total communication of \algref{alg:est:coor}. 
To that end, we partition the analysis into rounds, where each round corresponds to an interval during which the quantity $U=\sum_{j\in[k]}\sum_{i\in[n]}(g_i(j))^p$ increases by at most a factor of two and then we upper bound the total number of rounds. 

\begin{lemma}
Define a round to be the time steps between which the quantity $U=\sum_{j\in[k]}\sum_{i\in[n]}(g_i(j))^p$ doubles. 
Then, there are at most $\polylog(n)$ rounds in expectation.
\end{lemma}
\begin{proof}
We define $V = \sum_{j\in[k]}\sum_{i \in [n]} f_i(j)^p$.
Recall that $U = \sum_{j\in[k]}\sum_{i \in [n]} \frac{f_i(j)^p}{\be_i}$.
Conditioned on $\be_i < \polylog(n)$, which happens with high probability, we have $V \le U \cdot \polylog(n)$.
Moreover, condition on $\frac{1}{\poly(n)} < \be_i < \poly(n)$ for all $i \in [n]$, by \lemref{lem:truncated:inverse:ex}, we have
\[\Ex{U} = \Ex{\sum_{j\in[k]}{i \in [n]} \frac{f_i(j)^p}{\be_i}} = V \cdot \O{\log n}.\]
Then, consider a duration within the duration that $V$ doubles, $U$ at most increases by a $\polylog(n)$-fraction in expectation: let $V_1$, $V_2$ be the value at the start and the end of this duration, $\Ex{U_2} = V_2 \cdot \O{\log n} = 2 V_1\cdot \O{\log n} \le U_1 \cdot \polylog(n)$.
This means that there are $\poly(\log\log(n))$ number of rounds in expectation within this duration.
Assuming that there is $m =\poly(n)$ updates in the data stream, there can be at most $\O{\log n}$ times that $V$ doubles.
\end{proof}

Then, it suffices to upper bound the total communication within each round, stated as follows.

\begin{lemma}
\lemlab{lem:comm:pl}
With high probability, \algref{alg:est:coor} uses $(k+k^{p-1}) \cdot \poly\left(\frac{\log n}{\eps}\right)$ bits of communication in expectation.
\end{lemma}
\begin{proof}
Consider a fixed round.
Recall that $U = \sum_{j\in[k]} \sum_{i \in [n]} g_i(j)^p$, then by \lemref{lem:sample:prob}, we sample $\gamma = k^{p-1} \cdot \poly\left(\frac{\log n}{\eps}\right)$ elements with high probability.
For each element, we delete it and re-sample it whenever it increases by a $\frac{\eps}{100 \log n}$-fraction, which each takes at most $ \poly\left(\frac{\log n}{\eps}\right)$ bits of communication.
In addition, we use $k \log n$ bits of communication in the beginning of each round to inform the new value of $U$, giving a total communication cost of $k^{p-1} \cdot \poly\left(\frac{\log n}{\eps}\right)$ bits.
\end{proof}

Combining \lemref{lem:est:cts:p:small} for $\eps=\Theta(1)$ and \lemref{lem:comm:pl}, we have the following guarantees for \algref{alg:est:server} and \algref{alg:est:coor}, both showing that the set $\SetPL$ contains the maximizer and upper bounding the total communication. 
\begin{theorem}
\thmlab{thm:pl_recover}
Given a vector $f = \sum_{j\in[k]} f(j) \in \mathbb{R}^n$ defined by a distributed data stream $S$, $n$ i.i.d. exponential random variable $\be_1,\ldots,\be_n$, and a scaled vector $g_i = \frac{f_i}{\be_i^{1/p}}$, there exists an algorithm that finds a set $\SetPL$ of coordinates that includes the maximizer of $g$ and only contains the items that are at least $\frac{1}{4}$-fraction of the maximizer at all times.
The algorithm succeeds with probability $1-\frac{1}{\poly(n)}$ and uses $\left(k+k^{p-1}\right)\cdot\polylog(n)$ total bits of communication in expectation. 
\end{theorem}

\subsection{Truncated Taylor Estimator}
\seclab{sec:truncated:taylor}
This section introduces another core technical tool in the $L_p$ sampler as well as our robust $F_p$-estimation algorithm in following sections.
We aim to retrieve unbiased estimates of $f_i^p$ for $p>0$. 
However, we only have access to unbiased estimates of $f_i$, so there is no direct method to estimate $f_i^p$ when $p$ is not an integer. 
To handle this challenge, we use the Taylor series of $f_i^p$. 
Specifically, let $s_i$ be a $1.01$-approximation to $f_i$. 
Then 
\[f_i^p = \sum_{q=0}^\infty \binom{p}{q}s_i^{p-q}(f_i-s_i)^q.\]
Since the tail of the Taylor series decays exponentially, we can truncate at $Q=\O{\log n}$ terms and only incur $\O{n^{-c}}$ additive error. 
Then we for each $q\in[Q]$, we can estimate $(f_i-s_i)^q$ using the product of $q$ independent unbiased estimate of $f_i$.
We give the algorithm in full in \algref{alg:taylor}. 

\begin{algorithm}[!htb]
\caption{Truncated Taylor estimator}
\alglab{alg:taylor}
\begin{algorithmic}[1]
\Require{$Q = \O{\log n}$ independent unbiased estimates $\{\widehat{f_i^{(l)}}\}_{l \in [Q]}$ of $f_i$ satisfying $\Var{f_i^{(l)}} = \left(\frac{f_i}{100p}\right)^2$, $p>0$}
\Ensure{Unbiased estimate of $f_i^p$}
\State{Let $s_i$ be a $(1+\frac{1}{100p})$-approximation of $f_i$}
\For{$q \in [Q] \cup \{0\}$ }
\State{$\widehat{(f_i-s_i)^q} \gets \prod_{l \in [q]} \left(\widehat{f_i^{(l)}} - s_i\right)$} 
\State{$H \gets H+\binom{p}{q}s_i^{p-q}\widehat{(f_i-s_i)^q}$} \Comment{$H$ is the truncated Taylor estimator for $f_i^{p}$}
\EndFor
\State{\Return $H$}
\end{algorithmic}
\end{algorithm}

We first show that we can truncate the Taylor series of $f_i^p$ after $Q=\O{\log n}$ terms and only incur additive error $\frac{1}{\poly(n)}$ from the tail of the Taylor series. 
\begin{lemma} \lemlab{lem:truncated:taylor:fi:p}
Let $Q=\O{\log n}$, $p>0$, and $s_i \in\left[\frac{99f_i}{100},\frac{101f_i}{100}\right]$. 
Then 
\[\left\lvert f_i^p-\sum_{q=0}^Q \binom{p}{q}s_i^{p-q}(f_i-s_i)^q\right\rvert\le n^{-c}.\]
\end{lemma}
\begin{proof}
Using the Taylor expansion we have
\[ f_i^p = \sum_{q=0}^\infty \binom{p}{q}s_i^{p-q}(f_i-s_i)^q.\]
Note that $|f_i - s_i| \le \frac{f_i}{100}$, so for $q > p$ we have
\[\left| \binom{p}{q}s_i^{p-q}(f_i-s_i)^q \right| 
\le \O{e^p \cdot f_i^{p-q} \left(\frac{f_i}{100}\right)^q}=\O{f_i^p \cdot \left(\frac{1}{100}\right)^q},\]
where the first step is by $|\binom{p}{q}| \le e^{p}$.
Hence, we have that for $Q = \O{\log n}$,
\[\sum_{q=Q+1}^\infty \binom{p}{q}s_i^{p-q}(f_i-s_i)^q \le \O{1} \cdot f_i^p \cdot \sum_{q=Q+1}^\infty \left(\frac{1}{100}\right)^q \le n^{-c},\]
since we assume the stream length to be $\poly(n)$.
\end{proof}

Now, we state an auxiliary lemma that will be used to upper bound the variance of the truncated Taylor estimator.
\begin{lemma}
\lemlab{lem:aux:lp:est}
Let $a$ be some constant and let $\widehat{a}$ be an estimate satisfying $\Ex{\widehat{a}} = a$ and $\Var{\widehat{a}} = \left(\frac{a}{C}\right)^2$, where $C$ is some large constant. 
Let $b$ be another constant satisfying $b = (1 \pm \frac{1}{C})\cdot a$. 
Then we have $\Ex{\widehat{a}-b} = a-b$ and $\Ex{(\widehat{a}-b)^2} \le 2 \cdot \left(\frac{a}{C}\right)^2$.
\end{lemma}
\begin{proof}
By the linearity of expectation, we have $\Ex{\widehat{a}-b} = a-b$. 
Now, we have
\begin{align*}\Ex{(\widehat{a}-b)^2} &= \Var{\widehat{a}-b} + \left(\Ex{\widehat{a}-b}\right)^2 \\
&= \Var{\widehat{a}} + \left(a-b\right)^2 \end{align*}
Then, by our assumption of $\Var{\widehat{a}} = \left(\frac{a}{C}\right)^2$ and $b = (1 \pm \frac{1}{C})\cdot a$, we have
\[\Ex{(\widehat{a}-b)^2} \le 2 \cdot \left(\frac{a}{C}\right)^2\]
\end{proof}

The following lemma shows the correctness of our truncated Taylor estimator.
\begin{lemma}
\lemlab{lem:taylor:fi}
Let $s_i$ be a $(1+\frac{1}{100p})$-approximation of $f_i$. For $Q = \O{\log n}$, let $\{\widehat{f_i^{(l)}}\}_{l \in [Q]}$ be independent estimates of $f_i$ satisfying $\Var{f_i^{(l)}} = \left(\frac{f_i}{100p}\right)^2$. 
Then, \algref{alg:taylor} returns the following truncated Taylor estimator
\[ \widehat{f_i^p} := \sum_{q=0}^Q\left(\binom{p}{q} s_i^{p-q} \prod_{l \in [q]} \left(\widehat{f_i^{(l)}}-s_i\right)\right)\]
and it satisfies $\Ex{\widehat{f_i^p}} =  f_i^p \pm n^{-c}$ and $\Ex{\left(\widehat{f_i^p}\right)^2} = \O{f_i^{2p}\cdot \log^2 n}$.
\end{lemma}
\begin{proof}
For $l \in [Q]$, the guarantee for each independent $\widehat{f_i^{(l)}}$ estimate is $\Ex{\widehat{f_i^{(l)}}} = f_i$ and $\Var{\widehat{f_i^{(l)}}} = \left(\frac{f_i}{100p}\right)^2$. Thus, applying \lemref{lem:aux:lp:est} with $a = f_i$, $b=s_i$, and $C = \frac{1}{100p}$, we have $\Ex{\widehat{f_i^{(l)}} - s_i} = f_i - s_i$ and $\Ex{\left(\widehat{f_i^{(l)} - s_i}\right)^2} = \O{\left(\frac{f_i}{50}\right)^2}$. Then, we have
\[\Ex{\prod_{l \in [q]} \left(\widehat{f_i^{(l)}}-s_i\right)} = \prod_{l \in [q]} \Ex{\widehat{f_i^{(l)}}-s_i} = (f_i - s_i)^q.\]
Moreover, we have
\[\Ex{\left(\prod_{l \in [q]} \left(\widehat{f_i^{(l)}}-s_i\right)\right)^2} = \prod_{l \in [q]} \Ex{\left(\widehat{f_i^{(l)}}-s_i\right)^2} = \O{\left(\frac{f_i}{50p}\right)^{2q}}.\]
Now, by \lemref{lem:truncated:taylor:fi:p} and $Q = \O{\log n}$, we have
\[\Ex{\widehat{f_i^p}} = \sum_{q=0}^Q \binom{p}{q} s_i^{p-q} \cdot \Ex{\prod_{l \in [q]} \left(\widehat{f_i^{(l)}}-s_i\right)} = \sum_{q=0}^Q \binom{p}{q} s_i^{p-q}(u_i - s_i)^q = f_i^p \pm n^{-c}.\]
Moreover, by the AM-GM inequality, we have
\begin{align*}\Ex{\left(\widehat{f_i^p}\right)^2} &= \Ex{\left(\sum_{q=0}^Q\left(\binom{p}{q} s_i^{p-q} \prod_{l \in [q]} \left(\widehat{f_i^{(l)}}-s_i\right)\right)\right)^2} \\ &= \O{Q} \cdot \sum_{q=0}^Q\binom{p}{q}^2 s_i^{2p-2q} \Ex{\prod_{l \in [q]} \left(\widehat{f_i^{(l)}}-s_i\right)^2}.\end{align*}
Therefore, 
\begin{align*}
\Ex{\left(\widehat{f_i^p}\right)^2} &= \O{Q} \cdot \sum_{q=0}^Q\binom{p}{q}^2 s_i^{2p-2q} \left(\frac{f_i}{50p}\right)^{2q} \\
&= \O{Q} \cdot \sum_{q=0}^Q s_i^{2p-2q} \left(\frac{p \cdot f_i}{50p}\right)^{2q}, \end{align*}
where the second step follows by $\binom{p}{q}<p^q$. Now, since $s_i > 0.99 f_i$
\[\Ex{\left(\widehat{f_i^p}\right)^2} = \O{Q} \cdot f_i^{2p} \sum_{q=0}^Q \left(0.99 f_i\right)^{-2q} \left(\frac{ f_i}{50}\right)^{2q} = \O{Q^2} \cdot f_i^{2p}.\]
By $Q = \O{\log n}$, we have our desired result.
\end{proof}

\subsection{Sampler Construction}
\seclab{sec:sampler:construct}
In this section, we construct the $L_p$ perfect sampler for the distributed monitoring setting. 
We first apply the subroutine introduced in \secref{sec:pl:recovery} to find a set $\SetPL$ that contains the maximum index with high probability. 
Then, we obtain a $\left(1+\frac{1}{100p}\right)$-approximation for each $g_i$ in $\SetPL$. 
Due to the error in our estimate for $g_i$, we may incorrectly identify the wrong index as the maximum index. 
To avoid this happening, we would like to do a statistical test on whether $G_{D(1)} > 2G_{D(2)}$ to fail the invalid samplers. 
However, it is not clear that the failure probability of the statistical test is independent of which index obtains the maximum. 
For instance, let the original vector be $f=(100n, 1, \ldots,1) \in \mathbb{R}^n$. 
Then we would expect that the first index is still the maximum in the scaled vector. 
And if we condition on the second index achieving the max, we would expect the max and the second max to not have a large gap, and we fail the test with a large probability. 
In this case, the failure probability differs by an additive constant when conditioning on the first or the second index achieving the max, which leads to a wrong sampling distribution.
Thus, we apply the method of duplication introduced by \cite{JayaramW18}, that is, we duplicate each item $n^c$ times in the original vector and scale them by different exponential variables. 
Formally, for a fixed $c>0$, let $F$ denote a duplicated vector, so that $F_{i,l} := f_i$ for all $i \in [n], l \in [n^c]$, but $F$ can be flattened to be a vector with dimension $n^{c+1}$. 
Notice that $\|F\|_p^p = n^c \cdot \|f\|_p^p$. 
We use $G$ to denote the scaled vector, so that $\frac{F_i}{e_i^{1/p}}$ for an independent exponential random variable $e_i$, for all $i\in[n^{c+1}]$. 

Unfortunately, as we mentioned in \secref{sec:lp:sampler}, this still does not suffice, because the distribution of the output of optimal approximate counting procedures is affected by the distribution of the coordinate across the servers. 
Namely, the estimate for $g_i$ may depend on the distribution of $g_i$ across each server $j$ and thus the failure probability for the statistical test is still not independent of the index $i$. 
To solve this problem, we induce a two-layer exponential scaling $h_i(j) = \frac{g_i(j)}{\be(j)}$ and apply the geometric mean estimator, which is unbiased and has bounded variance:
\[\widehat{g_i} = \frac{\sqrt[3]{h_i(j_{*,1})\cdot h_i(j_{*,2})\cdot h_i(j_{*,3})}}{C_{\ref{lem:geo:exp:var}}}.\]
Here, $j_{*,r}$ is the maximum index of $h^{(r)}_i(j)$ across all servers $j$ for the $r$-th instance. 
Now we can run the approximate counting algorithm introduced by \cite{HuangYZ12} to track each $h_i(j)$, where the communication is restricted to the coordinator and a single server $j$.
The procedure gives a $(1+\eps)$-estimate to the count of $h_i(j)$ using roughly $\O{\frac{1}{\eps}}$ bits of communication.
We use \textsc{Counting} to refer to this procedure in the following sections. 
Due to the max-stability of exponential variables, $\max_j h_i(j) = \frac{g_i}{\be'_i}$, where $\be'_i$ is an independent exponential variable. 
Thus, the estimate of the max only depends on the value of $\frac{g_i}{\be'_i}$, but not on how $g_i$ is distributed across the $k$ servers.
However, this approach requires us to identify the actual maximum index, e.g., when the first max and the second max is $(1+\alpha)$-fractional close to each, we need a separate counting algorithm with accuracy $\frac{\alpha}{100}$ to correctly distinguish between them.
Due to the anti-concentration of exponential variables, the first max and the second max of the scaled vector is $(1+\eps)$-fractional close with probability at most $\O{\eps}$, allowing us to bound the communication cost. 
We show our distributed monitoring $L_p$ sampler in \algref{alg:lp:sampler:cent}.

\begin{algorithm}[!htb]
\caption{Distributed monitoring $L_p$ sampler}
\alglab{alg:lp:sampler:cent}
\begin{algorithmic}[1]
\Require{Input vector $f = \sum_{j \in [k]} f(j) \in \mathbb{R}^n$ from a distributed data stream}
\Ensure{Perfect $L_p$ sample from $f$}
\State{All servers generate $n^{c+1}$ i.i.d. exponential variables $\{ \be_{i,l} \}_{i \in [n], l \in [n^c]}$ using public randomness. Let $R=\polylog(n)$ be sufficiently large. Each server $j$ generates $3R$ i.i.d. exponential variables $\{ \be^{(r)}(j) \}_{j \in [k], r\in [3R]}$ }
\State{Let $F(j), G(j), H^{(r)}(j) \in \mathbb{R}^{n^{c+1}}$ be the duplicated vectors: $F_{i,l}(j) = f_i(j)$, $G_{i,l}(j) = \frac{f_i(j)}{\be_{i,l}^{1/p}}$, and $H_{i,l}(j) = \frac{f_i(j)}{\be_{i,l}^{1/p} \be(j)}$ for all $i \in [n], l \in [n^c], j \in [k], r \in [3R]$}
\State{Maintain a $1.01$-approximation to $\|G\|_p^p$ (c.f. \cite{WoodruffZ12}). Start a new round with fresh randomness when $\|G\|_p^p$ doubles}
\State{For each update, obtain a set $\SetPL$ containing the maximum index} \Comment{See \thmref{thm:pl_recover}}
\If{a new element $G_{i,l}$ is added to $\SetPL$} \Comment{Do not delete}
\State{$\widetilde{H^{(r)}_{i,l}(j)} \gets H^{(r)}_{i,l}(j)$ for all $r \in [3R]$ and $j \in [k]$} \Comment{Query all servers}
\EndIf
\For{$(i,l) \in \SetPL$}
\For{$r \in [3R]$}
\State{Let $\widetilde{H^{(r)}_{i,l}(D(1))}$ and $\widetilde{H^{(r)}_{i,l}(D(2))}$ denote the first max and the second max of vector $[\widetilde{H^{(r)}_{i,l}(1)}, \ldots,\widetilde{H^{(r)}_{i,l}(k)}]$} 
\State{Let $\SetPL_i$ be the set of servers with $\widetilde{H^{(r)}_{i,l}(j)} \ge \frac{\widetilde{H^{(r)}_{i,l}(D(1))}}{200}$} \Comment{$\SetPL_i$ contains the servers with heavy values with respect to index $i$}
\State{$\widetilde{H^{(r)}_{i,l}(j)} \gets \textsc{Counting}(H^{(r)}_{i,l}(j),\frac{1}{100p})$ for all $j \in [k] \backslash \SetPL_i $}
\If{$\frac{\widetilde{H^{(r)}_{i,l}(D(1))}}{\widetilde{H^{(r)}_{i,l}(D(2))}} = 1+\alpha$} 
\State{$\widetilde{H^{(r)}_{i,l}(j)} \gets \textsc{Counting}(H^{(r)}_{i,l}(j),\frac{\alpha}{100p})$ for all $j \in \SetPL_i$} \Comment{\textsc{Counting} is ran on a single server}
\State{Inform all servers in $\SetPL_i$ to change the accuracy if $\alpha$ changes by a $2$-fraction}
\EndIf
\State{$j_{*,r} \gets \argmax_{j \in [k]} \widetilde{H^{(r)}_{i,l}(j)}$} \Comment{Obtain the maximum index across servers} \label{line:max}
\State{$\widehat{H^{(r)}_{i,l,q}(j)} \gets \textsc{Counting}(H^{(r)}_{i,l}(j),\frac{1}{100p})$ for all $j \in [k]$ and $q \in [Q]$, where $Q = \O{\log n}$} \Comment{Run a separate counting algorithm to estimate $G_{i,l}$} 
\State{$\widehat{\left(H^{(r)}_{i,l}(j)\right)^{1/3}} \gets \textsc{TaylorEst}(\{\widehat{H^{(r)}_{i,l,q}(j)\}_{q \in [Q]}})$} \Comment{See \lemref{lem:taylor:fi}}
\EndFor
\State{$\widehat{G_{i,l}} \gets \frac{1}{R}\sum_{r=1}^R \frac{\widehat{\left(H^{(3r-2)}_{i,l}(j_{*,3r-2})\right)^{1/3}} \widehat{\left(H^{(3r-1)}_{i,l}(j_{*,3r-1})\right)^{1/3}} \widehat{\left(H^{(3r)}_{i,l}(j_{*,3r})\right)^{1/3}}}{C_{\ref{lem:geo:exp:var}}}$} \label{line:geo:mean}
\EndFor
\State{$\mu \gets {\cal U}(0.99,1.01)$}
\State{Let $\widehat{G_{D(1)}}$ and $\widehat{G_{D(1)}}$ be the max and the second max of all tue estimates $\widehat{G_{i,l}}$}
\State{\Return $i^* = \argmax \widehat{G_{i,l}}$ if $\widehat{G_{D(1)}} > 2 \mu \widehat{G_{D(2)}}$; otherwise $\mathrm{FAIL}$}
\end{algorithmic}
\end{algorithm}

\paragraph{Correctness.} 
We first show that our algorithm has the correct output distribution. 
By \thmref{thm:pl_recover}, we know that $\SetPL$ contains the maximum index at all times with high probability. 
Therefore, we condition on these events in the following analysis.

We begin by bounding the failure probability of our proposed $L_p$ sampler. 
Here, we remark that our algorithm only shows an instance of the $L_p$ sampler. We show that it reports some index from the correct $L_p$ distribution with constant probability in the following proof. 
Hence, we run $\O{\log n}$ such instances simultaneously and take the output of the first successful sampler. 
By standard concentration inequalities, it reports some index at all times with high probability, while only losing $\O{\log n}$ factors in communication.

\FloatBarrier

First, we bound the error of our estimate $\widehat{G_{i,l}}$.
\begin{lemma}
\lemlab{lem:est:pl}
For all $(i,l) \in \SetPL$, $\widehat{G_{i,l}}$ is a $1.01$-approximation to its true value $G_{i,l}$.
\end{lemma}
\begin{proof}
Notice that $H_{i,l}^{(r)}(j) = \frac{F_{i,l}(j)}{\be_{i,l}^{1/p} \be^{(r)}(j)} = \frac{G_{i,l}(j)}{ \be^{(r)}(j)}$ for all $r \in [3R]$.
By the max-stability of the scaled vector (c.f. \corref{cor:exp_inverse}), we have $\max_{j \in [k]} H_{i,l}^{(r)}(j) = \frac{F_{i,l}}{\be_{i,l}^{1/p} \be^{*,(r)}}$, where $\be^{*,(r)}$ is an exponential variable independent of $\{\be^{(r)}_{i,l}(j)\}_{j \in [k]}$ and $F_{i,l} = \sum_{j \in [k]} F_{i,l}(j)$.
Now, let $j_{*,r} = \max_{j \in [k]} H_{i,l}^{(r)}(j)$ and we define the geometric mean estimator as
\[G_{i,l}^{(r)} = \frac{\sqrt[3]{H^{(3r-2)}_{i,l}(j_{*,3r-2})H^{(3r-1)}_{i,l}(j_{*,3r-1})H^{(3r)}_{i,l}(j_{*,3r})}}{C_{\ref{lem:geo:exp:var}}}.\]
Then, we have
\[G_{i,l}^{(r)} = \frac{G_{i,l}}{C_{\ref{lem:geo:exp:var}} \cdot\sqrt[3]{ \be^{*,(3r-2)}\be^{*,(3r-1)}\be^{*,(3r)}}}.\]
Applying \lemref{lem:geo:exp:var}, we have $\Ex{G_{i,l}^{(r)}} = G_{i,l}$ and $\Ex{\left(G_{i,l}^{(r)}\right)^2} = \O{G_{i,l}^2}$.

By our construction in \algref{alg:lp:sampler:cent}, the accuracy of the first counting algorithm is $\frac{\alpha}{100p}$ if the ratio between the max and the second max is $1+\alpha$.
So, we detect the actual max from the estimated scaled vector, i.e., $j_{*,r}$ in line~\ref{line:max} is the maximum index of vector $[H_{i,l}^{(r)}(1),\ldots,H_{i,l}^{(r)}(k)]$ for all $r \in [3R]$.
Additionally, by \lemref{lem:taylor:fi}, the truncated Taylor estimator for each $\left(H^{(r)}_{i,l}(j)\right)^{1/3}$ satisfies $\Ex{\widehat{\left(H^{(r)}_{i,l}(j)\right)^{1/3}}} = \left(H^{(r)}_{i,l}(j)\right)^{1/3} \pm n^{-c}$ and $\Ex{\left(\widehat{\left(H^{(r)}_{i,l}(j)\right)^{1/3}}\right)^2} = \O{\left(H^{(r)}_{i,l}(j)\right)^{2/3} \cdot \log^2 n}$. We define our estimate of $G^{(r)}_{i,l}$ as
\[\widehat{G^{(r)}_{i,l}} = \frac{\widehat{\left(H^{(3r-2)}_{i,l}(j_{*,3r-2})\right)^{1/3}} \widehat{\left(H^{(3r-2)}_{i,l}(j_{*,3r-2})\right)^{1/3}} \widehat{\left(H^{(3r-2)}_{i,l}(j_{*,3r-2})\right)^{1/3}}}{C_{\ref{lem:geo:exp:var}}}.\]
Then, we have $\Ex{\widehat{G_{i,l}^{(r)}}} = G_{i,l} \pm n^{-c}$ and $\Ex{\widehat{\left(G_{i,l}^{(r)}\right)}^2} = \O{G_{i,l}^2 \cdot \log^2 n}$. Thus, averaging $R=\polylog(n)$ estimates $\widehat{G_{i,l}^{(r)}}$ gives as a $1.01$-approximation to $g_{i,l}$.
\end{proof}
We now show that with high probability, our perfect $L_p$ sampler successfully produces a sample at each time, i.e., some instance does not fail the statistical test at each time.  
\begin{lemma}
\lemlab{lem:success}
With probability $1-\frac{1}{\poly(n)}$, our $L_p$ sampler succeeds at all times.
\end{lemma}
\begin{proof}
Consider a fixed time $t$. 
We accept the output from the $L_p$ sampler if its estimates satisfy $\widehat{G_{D(1)}} > 2 \mu \widehat{G_{D(2)}}$. 
From \lemref{lem:est:pl} we have that $\widehat{G_{D(1)}}$ and $\widehat{G_{D(2)}}$ are both $1.01$-approximations. 
Also, we have $\mu \in [0.99,1.01]$.
Therefore, we accept the sample if $G_{D(1)} > 3 G_{D(2)}$. 
Thus, it suffices to bound the probability that $G_{D(1)} > 3 G_{D(2)}$ holds at time $t$.
By \corref{cor:exp_inverse}, we have $G_{D(1)} = \|F\|_p / E_1^{1/p}$ and $G_{D(2)} = (E_1 / \|F\|_p^p + E_2 / \|F_{-D(1)}\|_p^p)^{-1/p}$, where $E_1$ and $E_2$ are independent exponential variables. 
Now, since we duplicate each term by $n^c$ times, we have
\[G_{D(2)} = \left(\frac{E_1 }{ \|F\|_p^p} + \frac{E_2}{\|F\|_p^p \cdot(1\pm n^c)}\right)^{-1/p} = \frac{\|F\|_p}{(E_1 + E_2\cdot(1\pm n^c))^{1/p}}.\]
Therefore, due to the probability density function of exponential variables, 
\[\PPr{G_{D(1)} > 3 G_{D(2)}} > \Omega(1).\]
Then, since we run $\O{\log n}$ parallel samplers with different exponential variables, with probability $1-n^{-c}$, at least one of them satisfies $G_{D(1)} > 3 G_{D(2)}$. 
By a union bound across all times, we have our desired result.
\end{proof}

Next, we show that conditioning on the identity of an index that achieves the max, the failure probability of a sampler only shifts by $\frac{1}{\poly(n)}$ after duplication. 
This is important because it shows that the failure condition cannot overly depend on the maximizer itself, which would otherwise destroy the perfect sampler. 

First, we recall the following statement by \cite{JayaramW18}, which states that each entry of the duplicated scaled vector can be decomposed as the sum of an independent term and a dependent term with magnitude $n^{-c}$.
\begin{lemma}
\cite{JayaramW18}
\lemlab{lem:anti:rank:indep}
Let $D=(D(1)\cdots,D(n^{c+1}))$ be the anti-rank vector of $G$. 
Let $N$ be size of the support of $F$. 
For every $i<N-n^{9 c / 10}$ we have 
\[G_{D(i)}= U_{D(i)}^{1 / p}\left(1 \pm \O{\frac{1}{p} n^{-(c+1) / 10}}\right),\]
where $U_{D(i)}=\left(\sum_{\tau=1}^a \frac{E_\tau}{\mathbb{E}\left[\sum_{j=\tau}^N\left|F_{D(j)}\right|^p\right]}\right)^{-1}$ is totally determined by $a$ and the hidden exponentials $E_i$, and thus, independent of the anti-rank vector $D$.    
\end{lemma}
Next, we introduce the following result, which states that if we decompose a variable $Z$ into $X$ and $Y$, where $X$ is independent of some event $E$ and $Y$ has small magnitude, then the conditional probability of $Z$ falling within an interval $I$ is close to the unconditional probability of $Z$ falling into $I$. 
\begin{lemma}
\cite{JayaramW18}
\lemlab{lem:sum_independent}
Let $X, Y \in \mathbb{R}^d$ be random variables where $Z=X+Y$. Suppose $X$ is independent of some event $E$, and let $M>0$ be such that for every $i \in[d]$ and every $a<b$ we have $\operatorname{Pr}[a \leq$ $\left.X_i \leq b\right] \leq M(b-a)$. Suppose further that $|Y|_{\infty} \leq \eps$. Then if $I=I_1 \times I_2 \times \cdots \times I_d \subset \mathbb{R}^n$, where each $I_j=\left[a_j, b_j\right] \subset \mathbb{R},-\infty \leq a_j<b_j \leq \infty$ is a (possibly unbounded) interval, then
$$\operatorname{Pr}[Z \in I \mid E]=\operatorname{Pr}[Z \in I]+\O{\eps d M}.$$
\end{lemma}
Utilizing the above lemma, we decompose the estimates of $G_{D(1)}$ and $G_{D(2)}$ into an independent term and a dependent disturbance with small magnitude, so that we eliminate the dependence of the failure probability on the anti-ranks.
\begin{lemma}
\lemlab{lem:lp:sampler:indep}
Let FAIL denote the event that we reject the sample given by \algref{alg:lp:sampler:cent}. We have $\PPr{\neg \mathrm{FAIL} ~| D(1)} = \PPr{\neg \mathrm{FAIL}} \pm n^{-c}$ at all times, where $D = (D(1), \cdots, D(n))$ is the anti-rank vector.
\end{lemma}
\begin{proof}
First, note that the subroutine that maintains $\SetPL$ succeeds with high probability by \thmref{thm:pl_recover}. This still holds after conditioning on which index achieves the max. To see this, let $\calE$ be the event that it fails, we have
\[\PPr{\calE} = \sum_{i \in [n^{c+1}]} \PPr{\calE ~| D(1) = i} \cdot \PPr{D(1)=i} = \sum_{i \in [n^{c+1}]} \PPr{\calE ~|D(1) = i} \cdot \frac{|F_i|^p}{\|F\|_p^p}.\]
We assume that $F_i \ge 1$ and $\|F\|_p = \poly(n)$, then we have for all $i$,
\[ \PPr{\calE ~| D(1) = i} \cdot \frac{1}{n^{\alpha}} \le \PPr{\calE} = n^{-\beta}.\]
Therefore, we have $\PPr{\calE ~| D(1) = i} \le n^{-\beta+\alpha}$ which is $\frac{1}{\poly(n)}$ by setting a large enough $\beta$.
We union bound across all times in the stream, and hence we show that the dependency only causes a $\frac{1}{\poly(n)}$ perturbation on the failure probability.
Next, we condition on the event $\calE$, which means that $\SetPL$ contains the maximum index at all times. 
Based on the same argument, we can condition on all (distributed) approximate counters succeed in the algorithm.
Then, it suffices to bound the influence of the dependency on the estimate of $G_{D(1)}$ and $G_{D(2)}$.

Now, we briefly describe the distributed counting algorithm introduced by \cite{HuangYZ12}, which gives a $(1+\eps)$-approximation to the total count $n$ of the number of updates in a distributed data stream (or the count of a particular item). 
Since we only run the counting algorithm on a single server, we consider a communication protocol between a server and a coordinator, and our task is to maintain an estimate at the coordinator of the count of the server. 
The server sends the first $\Theta\left(\frac{1}{\eps}\right)$ items, then it sets a probability $p = \Theta\left(\frac{1}{\eps \cdot n}\right)$. 
When the server receives an item, it sends its current local count to the coordinator with probability $p$. 
Let $\bar{n}$ be the last update received by the coordinator, the estimate is $\bar{n}-1+1/p$.
The algorithm starts a new round with updated sampling probability whenever the estimate of $n$ doubles.  

Recall that $H_{i,l}^{(r)}(j) = \frac{F_{i,l}(j)}{\be_{i,l}^{1/p} \be^{(r)}(j)}$ for all $r \in [3R]$.
By the max-stability of the scaled vector (c.f. \corref{cor:exp_inverse}), we have $\max_{j \in [k]} H_{i,l}^{(r)}(j) = \frac{F_{i,l}}{\be_{i,l}^{1/p} \be^{*,(r)}}$, where $\be^{*,(r)}$ is an exponential variable independent of $\{\be^{(r)}(j)\}_{j \in [k]}$, and $F_{i,l} = \sum_{j \in [k]} F_{i,l}(j)$.
This implies that the distribution of $\max_{j \in [k]} H_{i,l}^{(r)}(j)$ is independent of how $F_{i,l}$ is distributed across the server $j$.

Due to our algorithm construction, $j_{*,r}$ in line~\ref{line:max} of \algref{alg:lp:sampler:cent} is the maximum index of vector $[H_{i,l}^{(r)}(1),\ldots,H_{i,l}^{(r)}(k)]$ for all $r \in [3R]$.
Additionally, we run a separate counting algorithm on each server $j$ to obtain a $1.01$-approximation $\widehat{H^{(r)}_{i,l}(j)}$, which is independent of the counting algorithm to detect the max $j_{*,r}$.
We use the notation $\widehat{H^{(r)}_{i,l}(j_{*,r})}$ in the following estimation steps.
Notice that in the counting algorithm, the sampling probability is defined solely by the count value.
Therefore, the estimate of $\widehat{H^{(r)}_{i,l}(j_{*,r})}$ only depends on the value of $\max_{j \in [k]} H_{i,l}^{(r)}(j) $, which is $ \frac{F_{i,l}}{\be_{i,l}^{1/p} \be^{*,(r)}}$.

Let $D(i)$ be the anti-ranks of vector $G$. Using \lemref{lem:anti:rank:indep}, we have the following decomposition
\[G_{D(i)}= U_{D(i)}^{1 / p}\left(1 \pm \O{n^{-(c+1) / 10}}\right),\]
where $U_{D(i)}$ does not depend on the anti-ranks.
Moreover, since each $\be^{*,(r)}$ is an independent exponential variable that does not depend on the anti-ranks. Thus, we can write $\be^{*,(r)}$ inside $U_{D(i)}^{1 / p}$:
\[\frac{F_{D(i)}}{\be_{D(i)}^{1/p} \be^{*,(r)}} = U_{D(i)}^{1 / p}\left(1 \pm \O{ n^{-(c+1) / 10}}\right).\]
Since the sampling probability is defined as $\O{\frac{\be_{i,l}^{1/p} \be^{*,(r)}}{F_{i,l}}}$ for the counting algorithm, we can condition on that the $\O{ n^{-(c+1) / 10}}$ disturbance caused by the dependency does not affect the time that the server update its count to the coordinator. It follows that the estimate of the counting algorithm has a similar decomposition (here we abuse the notation of $U_{D(i)}$),
\[\widehat{H^{(r)}_{i,l}(j_{*,r})} =  \widehat{U_{D(i)}} \cdot \left(1 \pm \O{ n^{-(c+1) / 10}}\right).\]
Then, the truncated Taylor estimator in \lemref{lem:taylor:fi} is decomposed into
\[\widehat{\left(H^{(r)}_{i,l}(j_{*,r})\right)^{1/3}} =  \widehat{U_{D(i)}} \cdot \left(1 \pm \O{ \log n \cdot n^{-(c+1) / 10}}\right).\]
Consequently, the geometric mean estimator in Line~\ref{line:geo:mean} of \algref{alg:lp:sampler:cent} is decomposed into
\[\widehat{G_{D(i)}} =  \widehat{U_{D(i)}} \cdot \left(1 \pm \polylog n \cdot n^{-(c+1) / 10}\right).\]
We abort a sampler if $\widehat{G_{D(1)}} < 2 \mu \widehat{G_{D(2)}}$. So, we decompose $\widehat{G_{D(1)}} - 2 \mu \widehat{G_{D(2)}}$ into $\widehat{U_{D(1)}} - 2 \mu \widehat{U_{D(2)}} + \widehat{V_{D(1)}} - 2 \mu \widehat{V_{D(2)}}$ where $\widehat{U_{D(1)}} - 2 \mu \widehat{U_{D(2)}}$ is independent of the anti-ranks and $ \widehat{V_{D(1)}} - 2 \mu \widehat{V_{D(2)}}$ is some random variable upper bounded by $| \widehat{V_{D(1)}} - 2 \mu \widehat{V_{D(2)}}| \le \O{n^{-(c+1)/20}}$. Notice that for any interval $I \in \mathbb{R}$, we have
\[\PPr{\widehat{U_{D(1)}} - 2 \mu \widehat{U_{D(2)}} \in I} = \PPr{\mu \in I'/  2 \widehat{U_{D(2)}}} = \O{|I| /  2 \widehat{U_{D(2)}}}\]
where $I'$ is the shifted interval of $I$ which is independent of $\mu_1$. 
Therefore, applying \lemref{lem:sum_independent} with $X = \widehat{U_{D(1)}} - 2 \mu \widehat{U_{D(2)}}$, $Y = \widehat{V_{D(1)}} - 2 \mu \widehat{V_{D(2)}}$, and $E = D(1)$ we have
\[\PPr{\widehat{G_{D(1)}} - 2 \mu \widehat{G_{D(2)}} \le0~|~ D(1)} = \PPr{\widehat{G_{D(1)}} - 2 \mu \widehat{G_{D(2)}} } \pm \frac{1}{\poly(n)},\]
which proves our desired result.
\end{proof}
Finally, we show that our $L_p$ sampler has the correct sampling distribution.
\begin{theorem}
\thmlab{thm:correct:lp:sampler} 
\algref{alg:lp:sampler:cent} returns an index $\bi \in [n]$ at all times, such that for any $i \in [n]$, we have
\[\PPr{\bi = i} = \frac{f_i^p}{\|f\|_p^p} \pm \frac{1}{\poly(n)}.\]
Moreover, if $i$ is sampled, it outputs an estimation $\widehat{G_{i}}$ such that $\Ex{\widehat{G_{i}}} = G_i$ and $\Ex{\left(\widehat{G_{i}}\right)^2} = \O{G_i^2}$.
\end{theorem}
\begin{proof}
\algref{alg:lp:sampler:cent} reports the index if it satisfies that $\widehat{G_{D(1)}} > 2 \mu \widehat{G_{D(2)}}$. Since our algorithm gives a $1.01$-approximation for each item in $\SetPL$ at all times with probability $1 - n^{-c}$ (c.f. \lemref{lem:est:pl}), $G_{D(1)} $ is strictly larger than $ G_{D(2)}$. Hence, we recover the true maximum index of the scaled vector with probability $1 - n^{-c}$. 
Let $q$ be the success probability of \algref{alg:lp:sampler:cent}. Then, by \lemref{lem:lp:sampler:indep}, the probability that \algref{alg:lp:sampler:cent} outputs $\bi$ is
\begin{align*}
\sum_{j \in\left[n^{c+1}\right]} \operatorname{Pr}\left[\neg \mathrm{FAIL} \mid \bi_j=\arg \max _{i', j'}\left\{\left|G_{i'_{j'}}\right|\right\}\right] \operatorname{Pr}\left[i_j=\arg \max _{i', j'}\left\{\left|G_{i_{j'}}\right|\right\}\right] &=\sum_{j \in\left[n^{c}\right]} \frac{\left|f_i\right|^p}{\|F\|_p^p}\left(q \pm \frac{1}{\poly(n)}\right) \\ 
&=\frac{\left|f_i\right|^p}{\|f\|_p^p}\left(q \pm \frac{1}{\poly(n)}\right)
\end{align*}
Note that the additive $\frac{1}{\poly(n)}$ term due to dependencies on the maximizer can be union bounded across all samplers and all time. 
Moreover, by \lemref{lem:success}, some index is reported at all times with probability $1-n^{-c}$, since we run $\O{\log n}$ instances of \algref{alg:lp:sampler:cent}. 
Therefore, the probability of our $L_p$ sampler outputting $i$ is
\[\frac{f_i^p}{\|f\|_p^p} + \frac{1}{\poly(n)}.\]
Additionally, our guarantee for the estimation of $\widehat{G_i}$ follows from the analysis of \lemref{lem:est:pl}.
\end{proof}

\paragraph{Communication cost.} 
We analyze the communication cost of our algorithm. 
Let $\calE$ denote the event that $\frac{1}{\poly(n)} < \be < \poly(n)$ and let $\calB$ denote the event that $\be < \polylog(n)$.
From the pdf of exponential variables, we have $\PPr{\calE \cap \calB} = 1 - \frac{1}{\poly(n)}$. Therefore, we condition on $\calE$ and $\calB$ for all exponentials in this section, which still has high probability by a union bound. We use restriction $\Ex{X} = \Ex{X ~|~ \calE \cap \calB}$ for all expectation computations in this section.

We define a round to be the times between which $\|G\|_p^p$ doubles. 
Within a round, for a parameter $X$, we use $X'$ to denote this term at the end of this round, and we use $X$ to denote this term at the beginning of this round. 
For simplicity of notation, we use $F_i$ with $i \in [n^{c+1}]$ to denote a coordinate in the duplicated vector $F$, instead of $F_{i,l}$ in previous sections.
The next lemma upper bounds the expected number of rounds in the data stream.

\begin{lemma}
\lemlab{lem:round:lp:sampler}
Let a round be the times between which $\|G\|_p^p$ doubles, i.e., between $2^i$ and $2^{i+1}$ for a fixed integer $i$. 
Then there are at most $\polylog(n)$ rounds in expectation.
\end{lemma}
\begin{proof}
Conditioned on $\calE$ and $\calB$, we have
\[\Ex{\|G\|_p^p} = \Ex{\sum_{i \in [n^{c+1}]} \frac{(F_i)^p }{ \be_i}} = \O{\log n} \cdot \|F\|_p^p = \O{\log n} \cdot \|F\|_p^p,\]
where the second step follows from \lemref{lem:truncated:inverse:ex}.
Since we conditioned on $\calB$, we have $(F_i)^p / \be_i \ge (F_i)^p / \polylog(n)$ for all $i \in [n^{c+1}]$, which implies that $\|F\|_p^p \le \|G\|_p^p \cdot \polylog(n)$.
Then, consider a duration within the duration that $\|F\|_p^p$ doubles, $\|G\|_p^p$ at most increases by a $\polylog(n)$-fraction in expectation:
\[\Ex{\|G'\|_p^p} = \|F'\|_p^p \cdot \O{\log n} = 2 \|F\|_p^p\cdot \O{\log n} \le \|G\|_p^p \cdot \polylog(n).\]
This means that there are $\poly(\log\log(n))$ number of rounds in expectation within this duration.
Assuming that there is $m =\poly(n)$ updates in the data stream, there can be at most $\O{\log n}$ times that $\|F\|_p^p$ doubles, and so there are at most $\polylog(n)$ rounds in expectation.
\end{proof}

Given the above result, it suffices to upper bound the communication cost within one round.
To upper bound the communication in one round, we first upper bound the size of $\SetPL$ in one round. 
In our algorithm, once an index is added to $\SetPL$, we never delete it. 
Then it suffices to upper bound the number of additions. 
The key observation is that all elements in $\SetPL$ are roughly $\frac{1}{\log^2 n}$-heavy with respect to the scaled vector $G$ at the beginning. 
Therefore, it is still a heavy-hitter with respect to $G$ at the end, up to some $\polylog(n)$ factors. 
Thus, there are at most $\polylog(n)$ such elements.
We formalize this argument in the following lemma.
\begin{lemma}
\lemlab{lem:num:set:pl}
With high probability, there are at most $\polylog(n)$ elements added to $\SetPL$ during a round.
\end{lemma}
\begin{proof}
Let $G_{D(1)}$ be the max of $G$. 
Then by \factref{fact:max:heavy}, we have $G^p_{D(1)} \ge \frac{\|G\|_p^p }{ \polylog n}$ with high probability, so we condition on this event for the remainder of the proof. 
Note that our data stream is insertion-only. 
Thus by the guarantee of \textsc{PLRecovery} (c.f. \thmref{thm:pl_recover}), $\SetPL$ only contains those elements that are at least $\frac{1}{4}$-fraction of the maximum coordinate $G_{D(1)}$. 
At the end of this round, with high probability, we have
\[\|G'\|_p^p = 2\|G\|_p^p \le \frac{G^p_{D(1)}}{\polylog(n)}.\]
Therefore, there are at most $\polylog(n)$ elements in $\SetPL$ at the end of the round that are at least $\frac{1}{4}$-fraction of $G_{D(1)}$, and so there are at most $\polylog(n)$ elements in $\SetPL$.
\end{proof}

Next, we upper bound the communication used to estimate $G_i$. 
Recall that we scaled the local count $G_i(j)$ by an independent exponential and defined $H_i(j) = \frac{G_i(j)}{\be(j)}$. 
Due to the dependence on the anti-ranks, we need to identify the actual maximum candidate among the $k$ servers, which means that if the ratio between the max and the second max is $1+\alpha$, we need $\O{\frac{k}{\alpha}}$ bits of communication to run the counting algorithm with accuracy $\frac{\alpha}{100p}$ at each server. 
In addition, if the ratio is reduced from $1+\alpha$ to $1+\frac{\alpha}{2}$, we need to inform each server to run a counting algorithm with a higher accuracy. 
So, we need to upper bound the expected times of changing to a new anti-concentration ratio.

To solve this problem, we define an auxiliary variable $H_i = \sum_{j\in[k]}H_i(j) = \sum_{j\in[k]} \frac{G_i(j)}{ \be(j)}$.
First, we show that $H_i$ at most increases by $G_i \cdot \polylog(n)$ in expectation during the whole round.
Next, we show that whenever the ratio is reduced from $1+\alpha$ to $1+\frac{\alpha}{2}$, $H_i$ increases by at least $\alpha \cdot \frac{G_i} { \polylog(n)}$ with high probability.
This implies that we expected to witness at most $\frac{1}{\alpha} \cdot \polylog(n)$ times of such changes in a round, which does not give our desired communication bound.
Fortunately, due to the anti-concentration property of the exponential variables, the ratio between the max and the second max is $1+\alpha$ with probability at most $\O{\alpha}$.
Thus, we can amortize the communication to change accuracy by this property and prove an expected $k \cdot \polylog(n)$ communication in total.
We formalize the above argument in the following lemmas.
In the following statement, we show that each item in $\SetPL$ only increases by a $\polylog(n)$ fraction in each round.

\begin{lemma}
\lemlab{lem:gi:increase:ub}
With high probability, we have $G_i' \le G_i \cdot \polylog(n)$ for all $i \in \SetPL$.
\end{lemma}
\begin{proof}
For a coordinate $i \in \SetPL$, from the analysis in \lemref{lem:num:set:pl}, we know that $G_i^p \ge \frac{\|G\|_p^p }{\polylog(n)}$ with high probability.
From \lemref{lem:num:set:pl}, we have $|\SetPL| = \polylog(n)$, so by a union bound, we have $(G_i)^p \le G_i^p \cdot \polylog(n)$ for all $i \in \SetPL$ with high probability.
Conditioned on these events, with high probability we have
\[(G_i')^p \le \|G'\|_p^p \le 2 \cdot \|G\|_p^p \le G_i^p \cdot \polylog(n).\]
Note that in the above proof, we assume that $G_i$ is added to $\SetPL$ at the beginning of the round.
This assumption is without loss of generality, since otherwise, we suppose that $i$ is added to $\SetPL$ at time $t$. 
Then we can consider $t$ as the start time of this round for coordinate $i$, e.g., let $G_i = G_i^{(t)}$ and $H_i = H_i^{(t)}$, so the analysis remains the same since our data stream is insertion-only.
In addition, when $i$ is not added to $\SetPL$, we only run a \textsc{PLRecovery} on $G_i$ to give a rough estimate, so we do not incur additional communication to identify the max of $H_i(j)$.
\end{proof}
Next, we upper bound how much $H_i$ changes over the course of a round. 
\begin{lemma}
\lemlab{lem:hi:increase:round}
Define $H_i$ to be $\sum_{j\in[k]} H_i(j)$. 
Then with high probability, we have $\Ex{H'_i } \le G_i \cdot \polylog(n)$.
\end{lemma}
\begin{proof}
From the definition of $H_i$, we have
\begin{align*}
\Ex{H'_i} &= \Ex{\sum_{i \in k} \frac{G'_i(j) }{ \be_i(j)}} = \O{\log n} \cdot \sum_{i \in k} G'_i(j) = \O{G'_i \cdot \log n},
\end{align*}
where the second step follows from \lemref{lem:truncated:inverse:ex} and the third step follows from \lemref{lem:gi:increase:ub}. 
Then, since $G_i' \le G_i \cdot \polylog(n)$, c.f. \lemref{lem:gi:increase:ub}, we have
\[ \Ex{H'_i} = \O{G'_i \cdot \log n} \le G_i \cdot \polylog(n).\]
\end{proof}

Recall that in our algorithm, for each heavy coordinate $i \in \SetPL$, we select a list of servers that obtains a heavy value, denoted as $\SetPL_i$, and we only keep track of these servers with higher accuracy.
The next lemma shows that this set contains $\polylog(n)$ servers in expectation, so that we only inform those servers whenever we change the accuracy.
\begin{lemma}
\lemlab{lem:pl:i}
With high probability, there are at most $\polylog(n)$ elements added to $\SetPL_i$ in expectation during a round for all $i\in \SetPL$.
\end{lemma}
\begin{proof}
By \lemref{lem:hi:increase:round}, we have $\Ex{H'_i } \le G_i \cdot \polylog(n)$.
Recall that in the analysis of \lemref{lem:hi:increase:change}, we have shown $H_i(D(1)) \ge \frac{G_i }{ \polylog(n)}$ with high probability.
Note that we add $j$ to $\SetPL_i$ only if $H^{(t)}_i(j) \ge \frac{H_i(D(1))}{300}$ at some time $t$.
Since our stream is insertion-only, we have $H'_i(j) \ge \frac{H_i(D(1))}{300}  \ge \frac{G_i }{\polylog(n)}$. 
Then, there are at most $\polylog(n)$ items in $\SetPL_i$ in expectation.
\end{proof}

Next, we lower bound the increment of $G_i$ by that of $H_i$.
\begin{lemma}
\lemlab{lem:gi:increase:lb}
We define $H_i$ as $\sum_{j\in[k]} H_i(j)$. With probability at least $1 - \frac{\eta}{\polylog(n)}$, for all $i \in \SetPL$, we have that if $H_i$ increases by $\Delta$, then $G_i$ at least increases by $\frac{\eta\Delta}{k \polylog(n)}$.
\end{lemma}
\begin{proof}
Consider a fixed $i$, due to the distribution of exponential variables, we have for each $\be_i(j)$, it satisfies $\be_i(j) \ge \frac{\eta}{k \polylog(n)}$ with probability at least $1-\frac{\eta}{k \polylog(n)}$. 
By a union bound across the $k$ servers and the $\polylog(n)$ items in $\SetPL$, we have with  probability at least $1 - \frac{\eta}{\polylog(n)}$, $\be_i(j) \ge \frac{\eta}{k \polylog(n)}$ for all $i \in \SetPL$ and $j \in [k]$.
Let $\bar{H_i}$ be the value of $H_i$ at some later time, we have
\[\bar{H_i}-H_i = \sum_{j\in[k]} \frac{\Bar{G_i(j)}-G_i(j)}{\be_i(j)} < \frac{k \polylog(n)}{\eta} \cdot (\bar{G_i}-G_i).\]
where the second step follows by our conditioning on all $\be_i(j)$ larger than $\frac{1}{k \polylog(n)}$.
Thus, we have our desired result.
\end{proof}

Now we show that $H_i$ increases by at least $\frac{\alpha \cdot G_i }{ \polylog(n)}$ when we change the accuracy.
\begin{lemma}
\lemlab{lem:hi:increase:change}
Let $H_i(D(1))$ and $H_i(D(2))$ denote the first max and the second max of $H_i$. 
Suppose that the ratio $\frac{H_i(D(1))}{H_i(D(2))}$ changes from $1+\alpha$ to $1+2\alpha$ or $1+\frac{\alpha}{2}$, then $H_i$ increases by at least $\frac{\alpha \cdot G_i }{\polylog(n)}$ in this duration.
\end{lemma}
\begin{proof}
If the ratio is reduced from $1+\alpha$ to $1+\alpha/2$, then since our data stream is insertion-only, at least one element $H_i(j)$ increases by $\frac{\alpha}{4} \cdot H_i(D(1))$.
Thus, their sum $H_i$ increases by at least $\frac{\alpha}{4} \cdot H_i(D(1))$.
Recall that $H_i(j) = \frac{G_i(j)}{\be(j)}$. Then, due to the max-stability property of the exponential variables (c.f. \factref{fac:exp_order_statistics}), we have $H_i(D(1)) = \frac{\sum_{j\in[k]} G_i(j)}{\be} = \frac{G_i}{\be}$, where $\be$ is an independent exponential variable.
With our conditioning on the event $\calB$, we have $H_i(D(1)) \ge \frac{G_i }{\polylog(n)}$.
Therefore, $H_i$ increases by at least $\alpha \cdot \frac{G_i }{ \polylog(n)}$.
Using the same argument, if the ratio is increased from $1+\alpha$ to $1+2\alpha$, $H_i$ increases by at least $\alpha \cdot \frac{G_i }{ \polylog(n)}$ as well.
\end{proof}
Next, we upper bound the communication used to reset the accuracy for the counting algorithm.
\begin{lemma}
\lemlab{lem:identify:max}
With probability at least $1 - \frac{\eta}{\polylog(n)}$, we use $\frac{k}{\eta} \cdot \polylog(n)$ bits of communication to broadcast the new accuracy of the counting algorithm.
\end{lemma}
\begin{proof}
In this proof, we condition on the success of \lemref{lem:gi:increase:ub} and \lemref{lem:hi:increase:change}, which holds with high probability, by union bounding across the $\polylog(n)$ items in $\SetPL$, c.f. \lemref{lem:num:set:pl}. 
We also condition on the success of \lemref{lem:gi:increase:lb}, which holds with probability at least $1-\frac{\eta}{\polylog(n)}$. 
From \lemref{lem:pl:i}, we know that there are at most $\polylog(n)$ elements in $\SetPL_i$. 
Thus, we only pay $\polylog(n)$ bits of communication to inform all servers in $\SetPL_i$ each time we change the accuracy. 
Then it suffices to upper bound the number of times we switch to a new accuracy for a single instance $i \in \SetPL$.

Now, we consider a fixed instance $i \in \SetPL$. 
Let $D_l$ denote the number of times we switch to a new accuracy in cases where, at the moment we last set the current accuracy, the anti-concentration ratio $\frac{H_i(D(1))}{H_i(D(2))}$ lies in the interval $\left[1+\frac{1}{2^{l+1}},1+\frac{1}{2^l}\right)$. 
Notice that by the anti-concentration property of exponentials, c.f. \lemref{lem:anti}, the ratio is at least $\frac{1}{\poly(n)}$ at all times with high probability. 
Moreover, when the ratio is larger than $2$, we do not need to change the accuracy since a constant accuracy suffices. 
Note that the total number of times we change the accuracy is $D = \sum_{l \in [\O{\log n}]} D_l$. 
Thus, it remains to fix an index $l \in [\O{\log n}]$ and upper bound $D_l$.

We partition a round into blocks where $G_i$ increases by $\frac{\eta G_i}{10 k \cdot 2^l \polylog(n)}$.
Now we consider a block such that the anti-concentration ratio $\frac{H_i(D(1))}{H_i(D(2))}$ lies in the interval $\left[1+\frac{1}{2^{l+1}},1+\frac{1}{2^l}\right)$ at some time during this block.
Let this ratio be $1+\alpha$. 
By \lemref{lem:hi:increase:change}, $H_i$ increases by at least $\frac{G_i}{2^{l+1}\polylog(n)}$ if the ratio changes to $1+2\alpha$ or $1+\frac{\alpha}{2}$.
Then by \lemref{lem:gi:increase:lb}, we have that $G_i$ increases by at least $\frac{\eta G_i}{k \cdot 2^{l+1} \polylog(n)}$ in the duration when $H_i$ increases by at least $\frac{G_i}{2^{l+1}\polylog(n)}$.
However, $G_i$ increases only by $\frac{\eta G_i}{10 k \cdot 2^l \polylog(n)}$ at the end of this block.
Thus, if the ratio is $1+\alpha$ where $\alpha \in \left[1+\frac{1}{2^{l+1}},1+\frac{1}{2^l}\right)$ at time $t$ of some block and we broadcast a new accuracy at $t$, the next switch of accuracy does not happen in the same block.
Therefore, $D_l$ is upper bounded by the number of blocks such that the ratio ever falls to $\left[1+\frac{1}{2^{l+1}},1+\frac{1}{2^l}\right)$.

Recall that by \lemref{lem:gi:increase:ub}, we have that $G_i$ increases by at most $G_i \cdot \polylog(n)$ in a round.
Then, since $G_i$ increases by $\frac{\eta G_i}{10 k \cdot 2^l \polylog(n)}$ in each block, there are at most $\frac{k}{\eta} \cdot 2^l \polylog(n)$ blocks in a round.
Moreover, if the ratio lies in the interval $\left[1+\frac{1}{2^{l+1}},1+\frac{1}{2^l}\right)$ at time $t$ in some block, then it is at most $1+\frac{1}{2^{l-1}}$ at the beginning of this block. 
Otherwise, $H_i$ increases by at least $\frac{G_i}{2^{l-1} \polylog(n)}$, which implies that $G_i$ increases by at least $\frac{\eta G_i}{k 2^{l-1} \polylog(n)}$, so we are in another block at $t$.

Now, we consider the list of start times $T = \{t_1, \ldots t_B\}$ of each block. 
Consider a fixed time $t_b \in T$, due to the anti-concentration property of exponentials (c.f. \lemref{lem:anti}), the ratio $\frac{H_i(D(1))}{H_i(D(2))}$ is smaller than $1+\frac{1}{2^{l-1}}$ with probability at most $\O{\frac{1}{2^{l-1}}}$. 
Here, we remark that since we define the start times of each block by the increment of $G_i$, which is independent of the scaled variables $\be(j)$, the anti-concentration property of $\be(j)$ does not depend on $t_b$. 
Then, let $X_b$ be an indicator variable that is $1$ if the ratio is smaller than $1+\frac{1}{2^{l-1}}$ at $t_b$, and $0$ otherwise, and let $X = \sum_{b\in [B]} X_b$. 
Then, $X$ gives an upper bound on the number of blocks such that the ratio ever falls to $\left[1+\frac{1}{2^{l+1}},1+\frac{1}{2^l}\right)$. 
Since $\Ex{X_b} \le \O{\frac{1}{2^{l-1}}}$, we have
\[\Ex{X} = \sum_{b\in [B]} \Ex{X_b} \le B \cdot \O{\frac{1}{2^{l-1}}}.\]
Recall that the total number of blocks is $B \le \frac{k}{\eta}\cdot 2^l \polylog(n)$, then we have
\[\Ex{X}= \frac{k}{\eta}\cdot 2^l \polylog(n) \cdot \O{\frac{1}{2^{l-1}}} = \frac{k}{\eta} \cdot \polylog(n).\]
Therefore, by Markov's inequality $D_l \le X \le \frac{k}{\eta} \cdot \polylog(n)$ with probability at least $1-\frac{\eta}{\polylog(n)}$. 
Summing up all $\O{\log n}$ levels, $D = \sum_{l \in [\O{\log n}]} D_l \le \frac{k}{\eta} \cdot \polylog(n)$, showing our final statement. 
We note that we can union bound the failure probability across the $\polylog(n)$ elements in $\SetPL$. 
\end{proof}
Given the above bounds, we can upper bound the total communication used by the approximate counting algorithms and the corresponding procedures used to maintain the correct accuracies for the approximate counting algorithms. 
\begin{lemma}
\lemlab{lem:counting:comm}
We use at most $k \cdot \polylog(n)$ bits of communication to run the counting algorithm in expectation.
\end{lemma}
\begin{proof}
Consider coordinate $i$ in $\SetPL$.
By our algorithm construction, when the anti-concentration ratio $\frac{H_i(D(1))}{H_i(D(2))}$ lies in the interval $\left[1+\frac{1}{2^{l+1}},1+\frac{1}{2^l}\right)$, we track all servers $j \in \SetPL_j$ using a counting algorithm with accuracy $\frac{1}{100p \cdot 2^{l+1}}$.
In the counting algorithm with accuracy $\alpha$, each update is sent with probability $\O{\frac{1}{\alpha n}}$, where $n$ is the total count, which is $H_i(j)$ in our case.
Recall that from the analysis in \lemref{lem:identify:max}, we have the ratio is smaller than $1+\frac{1}{2^l}$ for at most $\O{\frac{1}{2^l}}$-fraction of blocks during a round in expectation.
Then, the expected number of updates sent by the counting algorithm is smaller than
\[\sum_{l \in [\O{\log n}]} \O{\frac{1}{2^l}} \cdot H_i(j) \cdot \O{\frac{100p \cdot 2^l}{H_i(j)}} = \O{\log n}. \]
By \lemref{lem:pl:i}, there is at most $\polylog(n)$ servers in $\SetPL_i$, giving us $\polylog(n)$ in total.

For $j \notin \SetPL_i$, it suffices to run a counting algorithm with accuracy $1.01$ on those servers, which uses $k \polylog(n)$ bits of communication in total.
In addition, we run $\polylog(n)$ separate instances of the counting algorithm for each $i \in \SetPL$ to estimate $G_i$, which uses an extra $k \cdot \polylog(n)$ bits of communication.
\end{proof}
Combining the results above, we can upper bound the communication cost of our $L_p$ sampler.
\begin{lemma}
\lemlab{lem:lp:sampler:comm}
\algref{alg:lp:sampler:cent} uses $(k+k^{p-1})\cdot \polylog(n)$ bits of communication in expectation.
\end{lemma}
\begin{proof}
The result follows from the combination of \thmref{thm:pl_recover}, \lemref{lem:identify:max} and \lemref{lem:counting:comm}.

To sum up, \thmref{thm:pl_recover} shows that we use $(k+k^{p-1}) \cdot \polylog(n)$ bits of communication in expectation to track the set $\SetPL$.
\lemref{lem:identify:max} shows that with probability at least $1 - \frac{\eta}{\polylog(n)}$, we use $\frac{k}{\eta} \cdot \polylog(n)$ bits of communication to broadcast the new accuracy of the counting algorithm.
\lemref{lem:counting:comm} shows that we use $k \cdot \polylog(n)$ bits of communication in expectation to run the counting algorithm at each server.
Let $X$ be the random variable denoting the communication for \lemref{lem:identify:max}.
Conditioned on $X< \poly(n)$, which happens with high probability, we have
\begin{align*}\Ex{X} &\le \sum_{b \in [\log n]} \O{2^b} \cdot \PPr{X \in \left[2^b, 2^{b+1}\right]} \\
& \le \sum_{b \in [\log n]} \O{2^b} \cdot \PPr{X \ge 2^b} \\
& \le \sum_{b \in [\log n]} \O{2^b} \cdot \frac{k \polylog(n)}{2^b} = k \polylog(n).
\end{align*}
Thus, we need $(k +k^{p-1}) \cdot \polylog(n)$ bits of communication in expectation in total.
\end{proof}

We conclude the result of our continuous perfect $L_p$ sampler, given the correctness of \thmref{thm:correct:lp:sampler} and \lemref{lem:lp:sampler:comm}.
\begin{theorem}
\thmlab{thm:lp:sampler}
Let vector $f \in \mathbb{R}^{n}$ be induced by a distributed stream $S$, let $k$ denote the total number of servers, and let $p\ge 1$. 
Then, there is an algorithm \textsc{LpSampler} that gives a perfect $L_p$ sampler of $f$, using at most $(k+k^{p-1})\cdot \polylog(n)$ bits of communication in expectation.
\end{theorem}

\subsection{Lower Bounds for Perfect Sampling}
In this section, we show that any perfect $L_p$ sampler requires $\Omega(k^{p-1})$ bits of communication. 
We first recall the multi-party set disjointness, where $k$ servers hold subsets $S_1,\ldots,S_k\subseteq[n]$. 
In the YES case, the subsets are all pairwise disjoint, so that there is no coordinate shared among any two servers. 
In the NO case, there exists a single coordinate that is shared across all $k$ servers, i.e., there exists $i\in[n]$ such that $i\in S_1\cap\ldots\cap S_k$. 
It is known that the multi-party set disjointness communication problem requires $\Omega\left(\frac{n}{k}\right)$ communication:
\begin{theorem}
\thmlab{thm:multi:set:lb}
\cite{Jayram09}
Any protocol that solves the multi-party set disjointness problem with probability at least $\frac{3}{4}$ requires $\Omega\left(\frac{n}{k}\right)$ bits of communication. 
\end{theorem}
By taking a constant number of perfect $L_p$ samples, we can identify any item shared across $k$ servers, for a universe of size $n=k^p$. 
This translates to a lower bound of $\Omega(k^{p-1})$ for perfect $L_p$ sampling. 
\thmlpsamplerlb*
\begin{proof}
Given sets $S_1,\ldots,S_k\subseteq[k^{p}]$, let $f(1),\ldots,f(k)\in\mathbb{R}^n$ denote the indicator vectors so that for all $i\in[n]$ and $j\in[k]$, we have $f_i(j)=1$ if and only if $i\in S_j$. 
Let $f=f(1)+\ldots+f(k)$. 
Now, if there exists an item $i\in[k^{p}]$ such that $i\in S_1,\ldots,S_k$, then $f_i=k$. 
Moreover, the other $k^p-1$ coordinates have magnitude at most $1$. 
Thus in this case, any perfect $L_p$ sampler on $f$ selects $i$ with probability at least $\frac{k^p}{k^p+k^p-1}\ge\frac{1}{3}$. 
Hence, with $4$ independent instantiations of perfect $L_p$ samplers, we will sample $i$ with probability at least $1-\left(1-\frac{1}{3}\right)^{4}=\frac{65}{81}\ge\frac{3}{4}$. 
We can then query all $k$ servers for the values of the sampled indices, which takes an additional $\O{k}$ communication, and check whether any of the sampled indices appears across all $k$ servers. 
Note that this also handles the case where all servers have disjoint pairwise support. 
Thus, we can solve multi-party set disjointness across $k$ servers on a universe of size $k^p$ with probability at least $\frac{3}{4}$ using a constant number of perfect $L_p$ servers. 
By \thmref{thm:multi:set:lb}, it therefore follows that any perfect $L_p$ sampler requires $\Omega(k^{p-1})$ bits of communication. 
\end{proof}

\section{Framework for Adversarially Robust Distributed Monitoring Algorithms}
In this section, we describe two frameworks that achieve adversarially robust distributed monitoring protocols. 
In \secref{sec:sketch:switch}, we first describe a simple framework that transforms non-robust distributed monitoring protocols into robust distributed monitoring protocols, at the cost of sub-optimal communication.  
In \secref{sec:diff:est}, we then describe a more sophisticated framework that achieves near-optimal communication for adversarially robust distributed monitoring, but requires more algorithmic ingredients. 

\subsection{Warm-up: Framework by Sketch Switching}
\seclab{sec:sketch:switch}
We start by describing the sketch switching framework. 
In this approach, we initialize several independent sketches to track the value of $F$. 
To prevent the adversary from affecting the output of the algorithm, the estimate of the $i$-th sketch is not revealed until the value of $F$ becomes at least $(1+\eps)^i$. 
Then, the algorithm repeatedly outputs this value until the $(i+1)$-th subroutine has value at least $(1+\eps)^{i+1}$. 
Hence, the algorithm will always give a $(1+\eps)$-approximation to $F$, and the adversary knows nothing about the internal randomness of a subroutine until it is revealed.
The framework appears in full in \algref{alg:frame:ss}. 

To formally analyze our algorithm, we first recall the definition of strong tracker:
\begin{definition}[Strong tracker, Definition 1.11 in \cite{WoodruffZ21b}]
Given a distributed data stream $S$, a fixed time $t_0$, a splitting time $t_1$, a frequency vector $u$ induced by updates from time $t_0$ to $t_1$, a $\eps$-strong tracker gives a $(1+\eps)$-approximation of $F(u)$ at all times with probability at least $1-\frac{1}{\poly(n)}$. We use ${\cal F}(t_0,t_1,\eps)$ to denote a $\eps$-strong tracker.
\end{definition}
In this section, we shall generally set $t_0$ to $0$, so that the strong tracker will simply estimate the value of $F$ on the frequency vector defined by the entire stream at each time $t_1$. 

Now, we define the flip number of a function $F$, which controls the time it changes proportionally in the data stream.
\begin{definition}
[Flip number, Definition 3.1 of \cite{BenEliezerJWY20}]
Let $\eps \ge 0$. 
Let $x = (x_0,\ldots,x_m)$ for some $m \in \mathbb{N}$ be a sequence of real numbers. 
We define the $(\eps,m)$-flip number $\lambda_{\eps,m}(x)$ of $x$ to be the maximum $k \in \mathbb{N}$ such that there exist $0 \le i_1 < \cdots < i_k \le m$ so that $x_{i_j-1} \notin x_i \cdot (1\pm \eps)$ for all $j = 2,3, \ldots k$. 

Fix a function $F: \mathbb{R}^n \to \mathbb{R}$ and a class $\mathcal{C} \subseteq([n] \times \mathbb{Z})^m$ of distributed stream updates, which has the form of $\left(\left(u_1, i_1\right), \ldots,\left(u_m, i_m\right)\right)$. Here, $u_j$ is the local server that receives the item, $i_j$ is the index in the induced vector. We define the $(\eps, m)$-flip number $\lambda_{\eps, m}(F)$ of $F$ over $\mathcal{C}$ to be the maximum, over all sequences , of the $\eps$-flip number of the sequence $y=\left(y_0, y_1, \ldots, y_m\right)$ defined by $y_j=F\left(1,j\right)$ for all $j \in [m]$, where $(1,j)$ denotes the frequency vector induced by the distributed stream updates $\left(u_1, i_1\right), \ldots,\left(u_j, i_j\right)$ (and $f^{(0)}$ is the $n$-dimensional zeros vector).
\end{definition}

This lemma shows that we can approximate a vector by another vector with a small number of unique coordinates, which is applied to define our algorithmic output.
\begin{lemma}[Lemma 3.2 of \cite{BenEliezerJWY20}]
\lemlab{lem:flip}
Let $\eps \in (0,1)$. 
Suppose that $u = (u_0,\ldots,u_m)$, $v = (v_0,\ldots,v_m)$, and $w = (w_0,\ldots,w_m)$, such that (1) $v_i = u_i \cdot (1\pm\frac{\eps}{8})$ for all $i \in [m]$, (2) $w_0 = v_0$, and if $w_{i-1} = v_i \cdot (1\pm \eps/2)$ then $w_i = v_{i-1}$, otherwise $w_i = v_i$. Then $w_i = u_i \cdot (1 \pm \eps)$ for all $i \in [m]$, and moreover, $\lambda_0(w) \le \lambda_{\eps/8,m}(u)$. In particular, if $u_j = F(1,j)$ where $(1,j)$ denotes the frequency vector induced by the distributed stream updates $\left(u_1, i_1\right), \ldots,\left(u_j, i_j\right)$ for all $j \in [m]$, then $\lambda_0(w) \le \lambda_{\eps/8,m}(u)$.
\end{lemma}

\begin{algorithm}[!htb]
\caption{Framework for Robust Distributed Streaming Algorithms via Sketch Switching}
\alglab{alg:frame:ss}
\label{alg:framework:ss}
\begin{algorithmic}[1]
\Require{Updates $(u_1,i_1), \cdots, (u_m,i_m) \in \mathbb{R}^{[n] \times [k]}$ for a distributed data stream, $i_t$ denotes the index of the insertion, accuracy $\eps \in (0,1)$, $\eps$-strong tracker ${\cal F}$ of $F$}
\Ensure{Adversarially robust $(1+\eps)$-approximation of $F$}
\State{$\lambda \gets \lambda_{\eps/8,m}(F)$, start $\lambda$ independent instances ${\cal F}_1, \cdots, {\cal F}_\lambda$ of $\frac{\eps}{8}$-strong tracker of $F$}
\State{$\rho \gets 1$, $Z \gets 0$}
\For{each update $(u_t,i_t)$}
\State{Update ${\cal F}_1, \cdots, {\cal F}_\lambda$ with $(u_t,i_t)$}
\State{$y \gets$ current output of ${\cal F}_\rho$}
\If{$Z \notin y \cdot (1 \pm \frac{\eps}{2})$}
\State{$Z \gets y$}
\State{$\rho \gets \rho + 1$} \Comment{Switch the sketch}
\EndIf
\State{{\bf Return} $Z$}
\EndFor
\end{algorithmic}
\end{algorithm}
We now describe the full guarantees of the sketch switching framework. 
\begin{theorem}[Sketch switching framework, analogous to Theorem 3.5 in \cite{BenEliezerJWY20}]
Let $n$ be the cardinality of the universe, let $k$ be the number of the servers, let $m$ be the stream length where $\log m = \O{\log n}$. Suppose there exists a $\eps$-strong tracker ${\cal F}$ of $F$ that uses $H(n, k, \eps)$ bits of communication, where the function $H$ depends on the problem $F$. Then, there exists an adversarially robust streaming algorithm that outputs a $(1+\eps)$-approximation to $F$ at all times with probability $1-\frac{1}{\poly(n)}$. Let $\lambda = \lambda_{\eps/8,m}(F)$. The total communication cost of the algorithm is $\O{\lambda \cdot H(n,k,\eps)}$  bits.
\end{theorem}
\begin{proof}
Without loss of generality, we can assume the adversary against a fixed randomized \algref{alg:frame:ss} is deterministic. This is because given a randomized adversary and algorithm, if the adversary succeeds with probability greater than $\delta$ in fooling the algorithm, there must exist a protocol with deterministic inputs of the adversary which fools $\mathcal{A}$ with probability greater than $\delta$ over the randomness of \algref{alg:frame:ss}. Also, conditioned on a fixing of the randomness for both the algorithm and adversary, the entire stream and behavior of both parties is fixed.
Thus, let $Z_t$ is the output of the streaming algorithm at time $t$, given $Z_1, Z_2, \ldots, Z_k$ and the stream updates $\left(u_1, i_1\right), \ldots,\left(u_k, i_k\right)$ so far, the next stream update $\left(u_{k+1}, i_{k+1}\right)$ is deterministically fixed. 

Now, consider the first instance ${\cal F}_1$. Let $t_1$ be the time that $Z \notin y \cdot (1\pm \frac{\eps}{2})$, where $y$ is the estimation of ${\cal F}_1$. Then it satisfies that $Z\notin y \cdot (1\pm \frac{\eps}{2})$ when $t \in [1,t_1-1]$. Moreover, since ${\cal F}_1$ is a $\frac{\eps}{8}$-strong tracker, it gives a $(1+\frac{\eps}{8})$ to $F$ at time $t_1$ when we reveal ${\cal F}_1$ and set the output $Z$ to be ${\cal F}_1(1,t)$. Observe that when we change the output $Z$ to ${\cal F}_1(1,t)$, we deactivate ${\cal F}_1$ and activate ${\cal F}_2$ to trace the following distributed stream. Since ${\cal F}_2$ is not revealed before, its randomness remains hidden from the adversary. Then, by an inductive argument we can show that at time $t_{l-1}$ when the $l-1$-th sketch is revealed and the $l$-th sketch is activate. We have that $Z \notin y \cdot (1\pm \frac{\eps}{2})$ when $t \in [t_{l-1},t_l]$ and ${\cal F}_l$ gives a $(1+\frac{\eps}{8})$ to $F$ at time $t_l$. Notice that we can union bound the failure probability of all sketches. Therefore, applying \lemref{lem:flip} with $u = (F(1,1),F(1,2),\ldots,F(1,m))$, $v = (F(1,1), {\cal F}_1(1,2), \ldots, {\cal F}_1(1,t_1), {\cal F}_2(1,t_1+1), \ldots, {\cal F}_1(1,t_2),\ldots)$, and $w$ being the output of \algref{alg:frame:ss}, we have that the algorithm returns a $(1+\eps)$-approximation to the function value $F$ at all times with high probability. Also, the flip number of $w$ satisfies $\lambda_0(w) \le \lambda_{\eps/8,m}(F)$. Therefore, it is easy to see it suffices to implement $\lambda_{\eps/8,m}(F)$ sketches in total so that our communication cost is $\O{\lambda_{\eps/8,m}(F)\cdot H(n,k,\eps)}$ in bits.
\end{proof}

\subsection{Framework by Difference Estimator}
\seclab{sec:diff:est}
In this section, we present our framework for adversarially robust distributed monitoring through the usage of difference estimators, which yields optimal results.
We first recall the following definition of a difference estimator introduced by \cite{WoodruffZ21b}:

\begin{definition}[Difference estimator]
Suppose that we are given a distributed data stream $S$, a splitting time $t_0$, a frequency vector $u$ induced by updates from time $0$ to $t_0$, a frequency vector $v$ induced by updates from time $t_0$ to $t$, an accuracy parameter $\eps \in (0,1)$, and a ratio parameter $\gamma \in(0,1]$. 
Suppose $F(u+v)-F(u) \leq \gamma \cdot F(u)$ and $F(v) \leq \gamma F(u)$. 
Then a $(\gamma, \eps)$ difference estimator for a function $F$ outputs an additive $\eps \cdot F(u)$ approximation to $F(u+v)-F(u)$ with probability at least $1-\frac{1}{\poly(n)}$. 
\end{definition}
We use continuous difference estimator to refer to a distributed functional monitoring protocol where the frequency vector $v$ evolves over time and we use static difference estimator to refer to a protocol for the one-shot setting where $v$ is fixed. 
We do not show the existence of difference estimators for various problems of interest; for these, we defer to \secref{sec:robust:algs}. 
We also introduce the following concept of a distributed algorithm that flags each time the value of $F$ roughly doubles:
\begin{definition}[Double detector]
Given a distributed data stream $S$, a starting time $t_0$, a splitting time $t_1$, an evolving frequency vector $u$ induced by updates from time $t_0$ to $t_1$, an $\eps$-double detector flags at some time $t^*$ where $F(u)$ satisfies $0.99 \cdot 2^a \leq F \leq 1.01 \cdot 2^a$ for $a \in \mathbb{N}$. 
Moreover, it outputs a $(1+\varepsilon)$-approximation of $F(u)$ at $t^*$ with failure probability $1-\frac{1}{\poly(n)}$. 
\end{definition}

Our framework progresses in a number of rounds, switching to a new round whenever the double detector discovers that $F$ has increased by roughly a factor of $2$. 
In this paper, we assume that $F$ at most doubles $\O{\log n}$ times through out the entire stream with length $m=\poly(n)$.
We use the counter $a$ to count the number of rounds. 

Consider a fixed round, let $u$ be the vector at its initialization and let $u+v$ be the underlying frequency vector for the full data stream at some later time. 
The key intuition of the difference estimator framework is that, when $F(u+v)-F(u)$ increases by roughly $2^j \eps \cdot F(u)$, it suffices to keep a $(1+\frac{1}{2^j})$-approximation to the increment, i.e., an additive $\eps \cdot F(u)$ error in total. 
Let $v_{\le j} = \sum_{l=1}^{j} v_l$ where $v_1 \ldots v_\beta$ arrive in subsequential order. Then, we can decompose
\[F(u+v) = F(u) + \sum_{j=1}^{\beta} \Bigl( F(u+v_{\le j}) - F(u+v_{\le j-1}) \Bigr),\]
where we adopt the convention that $u+v_{\le 0} = u$ and $u+v_{\le\beta} = u+v$. 
Moreover, suppose that $F(u+v_{\le j}) - F(u+v_{\le j-1}) \le 2^j \eps F(u)$ for all $j$. 
Then, we can estimate $F(u+ v_{\le j}) - F(u+ v_{\le j-1})$ by a difference estimator with accuracy $\frac{1}{2^j}$. 

In our algorithmic implementation, we use the counter $b$ to count the number of blocks in one round, where a block is defined as an interval in which $F$ increases by roughly $\eps \cdot F(u)$. 
Suppose that $t_a$ is the most recent time that the double detector $\mathcal{M}$ reported. 
We use the counter $Z$ to stitch together the sketches of the difference estimator. 
In particular, let the binary expression of $b + 1= \sum_{j=1}^r 2^{z_j}$. 
Then, to estimate $F(1,t) - F(1,t_a)$ with block size $b$,
the algorithm stitches together estimate $Z_{a,j}$ for $F(1,t_{a,j}) - F(1,t_{a,j-1})$ with the corresponding granularity. Specifically, for the first $r - 1$ difference $F(1,t_{a,j}) - F(1,t_{a,j-1})$, we use a static estimators as these difference estimators have already been frozen and for the last part we shall use a continuous difference estimator with granularity $0$.
Whenever the value of $b$ increases, we update the start time and end time of the corresponding difference estimators.

When $b + 1$ is even, we slightly modify the binary expression to ensure that the ongoing sketch is at granularity $0$. 
This is beneficial because we only need a constant approximation at granularity $0$, whose communication cost avoids $\frac{1}{\eps}$ factors. 
In particular, the resulting estimate to the true value of the function has error roughly $\O{\eps}\cdot F(1,t_a)$. 
Thus when the estimated value increases by roughly $\eps \cdot F(1,t_a)$, we increase the value of $b$, freeze the current sketch, and set up the start time of the new sketches. 

However, at this point, we may need a difference estimator with higher granularity to avoid the error accumulating. 
Notice that the centralized coordinator sends the stopping time each to the servers whenever a sketch is frozen, so the subsequent vector is fixed and can be retrieved from each server. 
Thus, we can apply an $(\eps,\gamma)$-static difference estimator.
The framework appears in full in \algref{alg:frame}.

\begin{algorithm}[!htb]
\caption{Framework for Robust Distributed Algorithms via Difference Estimator}
\alglab{alg:frame}
\label{alg:framework}
\begin{algorithmic}[1]
\Require{Updates $(u_1,i_1), \cdots, (u_m,i_m) \in \mathbb{R}^{[n] \times [k]}$ arriving in a distributed data stream, accuracy $\eps \in (0,1)$, $(\eps, \eps)$-continuous estimator ${\cal B}_C$,  $(\gamma, \eps)$-static estimator ${\cal B}_S$, $\eps$-double detector ${\cal M}$}
\Ensure{Adversarially robust $(1+\eps)$-approximation of $F$}
\State \label{line:init_count}
$\beta \gets\lceil \log \frac{8}{\eps} \rceil$, $\eps' \gets \frac{\eps}{64\beta}$, $a \gets 0$, $\widehat{F} \gets 0 $
\State For $j \in [\beta]$,  $\gamma_j \gets 2^{j} \eps$ \Comment{Accuracy of each difference estimator}
\For{each update $(u_t, i_t)$}
    \State \label{line:prefix_count_f0} Update ${\cal M}_{a+1}(1,t,\frac{\eps}{100})$ with $(u_t, i_t)$  \Comment{Double detector}
    \If{${\cal M}_{a+1}(1,t,\frac{\eps}{100})$ returns $\mathbf{FLAG}$} \Comment{Switch to a new round when $F$ duplicates}
    \State{$\widehat{F} \gets {\cal M}_{a+1}(1,t,\frac{\eps}{100})$} \Comment{Run the double detector to give a $(1+\eps)$-approximation}
    \State $a \gets a+1$, $b \gets 0$, $Z_a \gets \widehat{F}$, $t_{a,j} \gets t$ for $j \in [\beta]$ 
    \EndIf
    \If{$b+1$ is odd}
    \State $r \gets \operatorname{numbits}(b+1), z_i \leftarrow \operatorname{lsb}(b+1, r+1-i)$ for all $i \in [r]$ \Comment{$z_1 > \cdots > z_r$ are the nonzero bits in the binary representation of $b$}
    \Else
    \State{$r \gets \operatorname{numbits}(b)+1, z_i \leftarrow \operatorname{lsb}(b, r+1-i)$ for all $i \in [r-1]$, $z_r \gets 0$ } \Comment{Make sure the ongoing sketch is a difference estimator at granularity $0$, i.e., $\gamma = \eps$}
    \EndIf
    \State $\widehat{F} \gets Z_a$
    \For{$0 \leq j \leq r-1$} \Comment{Sum up previous frozen sketches to estimate $F$}
    \State \label{line:threshold_round_f0} $\widehat{F} \gets \widehat{F} + Z_{a,z_j}$ 
    \EndFor
    \State $j \gets  b+1 $
    \State $\widehat{F} \gets \widehat{F} + {\cal B}_{C_j} (t_{a, z_r}, t, \eps', \eps' )$ \label{line:est_f} \Comment{Use continuous difference estimator for the ongoing sketch}
    \If{$\widehat{F} > (1 + \frac{(b+1)\eps}{8}) \cdot Z_a$} \label{line:threshold_block_f0} \Comment{Switch to the next block in the same round}
    \State $b \gets b+1$, $r \gets \operatorname{numbits}(b), z_r \gets \operatorname{lsb}(b,1)$, $j \gets \lfloor \frac{b}{2^{z_r}} \rfloor$
    \State \label{line:block_count_f0}$Z_{a, r} \gets {\cal B}_{S_j} (t_{a, z_r}, t, \gamma_{z_r}, \eps')$ \Comment{Use static difference estimator for the frozen sketch}
    \State $t_{a,j} \gets t$ for $j \in [z_r]$ \Comment{Update difference estimator times}
    \EndIf
    \State \Return{$(1+ \frac{b \eps}{8}) \cdot Z_a$}
    \EndFor
\end{algorithmic}
\end{algorithm}

Now, we analyze the correctness of our framework.
The proof is analogous to that of \cite{WoodruffZ21b}.
For simplicity, in the proof, we use ``difference estimator'' to refer to both the continuous version and the static version since they have the same accuracy.
In our algorithm, we use $\calB(t_0,t,\gamma,\eps)$ to denote a difference estimator, where $t_0$ is the split time of the prefix vector $u$, and $t$ is the current time of the data stream.
The following lemma shows that during the activated duration of a $(\gamma,\eps)$-difference estimator, $F$ at most increases by $\frac{\gamma}{2}F_{\mathrm{start}}$ additively, where $F_{\mathrm{start}}$ is the value of $F$ at the start of the round, so that the difference estimator is always valid.

\begin{lemma}[Success of difference estimators]
\lemlab{lem:success:diff}
Suppose all subroutines ${\cal B}$ and ${\cal M}$ succeed. Consider a fixed round $a$, let $t_a$ denote its start time, let $F(1,t_a)$ be the function value at time $t_a$. 
For each integer $r \ge 0$, let $t_{a,r}$ and $w_{a,r}$ be the start time and end time of difference estimator $Z_{a,r}$. 
We have $ F(1,w_{a,r}) - F(1,t_{a,r}) \le \frac{2^r\eps}{2}\cdot F(1,t_a)$, and all difference estimators give a valid estimation. 
\end{lemma}
\begin{proof}
Fix a block $b$. Consider the binary representation of $b$:
\begin{align*}
    b = \sum_{i=1}^r 2^{z_i}, ~ \mathrm{where} ~ z_1 > \cdots > z_r > 0, r = \mathrm{numbits}(b),
\end{align*}
such that the difference estimator ${\cal B}_{a,{z_r}}(t_{a,{z_r}}, w_{a,{z_r}},\gamma_{z_r}, \eps')$ is running in block $b$ to estimate $F(1,w_{a,z_r}) - F(1,t_{a,z_r})$. Now, suppose by way of contradiction we have $F(1,w_{a,z_r}) - F(1,t_{a,z_r}) > \frac{2^r\eps}{2}\cdot F(1,t_a)$. Then, at some time $t < w_{a,z_r}$, it satisfies that 
\[F(1,t) - F(1,t_{a,z_r}) < \frac{2^r\eps}{2} F(1,t_a)~\mathrm{and} ~F(1,t+1) - F(1,t_{a,z_r}) \ge \frac{2^r\eps}{2} F(1,t_a).\]
Since we only receive one update at time $t$, so we have $F(1,t) - F(1,t_{a,z_r}) > \frac{2^r\eps}{4} F(1,t_a)$. Then, because we suppose all procedures ${\cal B}$ succeed and the actual difference has not passed $\frac{2^r\eps}{2} F(1,t_a)$, by definition our difference estimator ${\cal B}_{a,{z_r}}(t_{a,{z_r}}, w_{a,{z_r}},\gamma_{z_r}, \eps')$ gives an estimation with additive error $\eps' \cdot  F(1,t_a)$. So, we have
\[ Z_{a,z_r} \ge F(1,w_{a,z_r}) - F(1,t_{a,z_r}) - \eps' \cdot F(1,t_a) \ge \frac{2^r\eps}{4} F(1,t_a).\]
Notice that our criterion for switching to a new difference estimator is the estimated value increases by $\frac{(b-b')\eps}{8} \cdot Z_a = \frac{2^{z_r} \eps}{8}\cdot Z_a$. Since $Z_a$ is a $(1+\frac{\eps}{100})$-approximation to $F(1,t_a)$, $Z_{a,z_r} > \frac{2^r\eps}{5} Z_a$ already passes the threshold. Thus, ${\cal B}_{a,z_r}$ should have been switched to a new difference estimator, which is a contradiction. Therefore, $ F(1,w_{a,r}) - F(1,t_{a,r}) \le \frac{2^r\eps}{2}\cdot F(1,t_a)$, and all difference estimators give a valid estimation.
\end{proof}

The following lemma shows that for each difference estimator with granularity $r$, the value $F$ increases by roughly a $2^r \eps$-fraction after it completes, which further upper bounds the number of difference estimators that are used in each round.
This statement will eventually be used to bound the estimation error and the communication cost of the framework. 

\begin{lemma}[Upper bound on the number of difference estimators]
\lemlab{lem:geo:bounds}
Suppose all subroutines ${\cal B}$ and ${\cal M}$ succeed. Consider a fixed round $a$, let $t_a$ denote its start time, let $F(1,t_a)$ be the function value at time $t_a$. 
For each integer $r \ge 0$, let $t_{a,r}$ and $w_{a,r}$ be the start time and end time of difference estimator $Z_{a,r}$. 
We have $ F(1,w_{a,r}) - F(1,t_{a,r}) \ge \frac{2^{r} \eps}{4} \cdot F(1,t_a)$.
\end{lemma}
\begin{proof}
As before, we fix a block $b$. Consider the binary representation of $b$:
\begin{align*}
    b = \sum_{i=1}^r 2^{z_i}, ~ \mathrm{where} ~ z_1 > \cdots > z_r > 0, r = \mathrm{numbits}(b),
\end{align*}
such that the difference estimator ${\cal B}_{a,{z_r}}(t_{a,{z_r}}, w_{a,{z_r}},\gamma_{z_r}, \eps')$ is running in block $b$ to estimate $F(1,w_{a,z_r}) - F(1,t_{a,z_r})$. By \lemref{lem:success:diff}, our difference estimator ${\cal B}_{a,{z_r}}(t_{a,{z_r}}, w_{a,{z_r}},\gamma_{z_r}, \eps')$ outputs an estimation with additive error $\eps' \cdot F(1,t_i)$. Let $b' = \sum_{i=1}^{r-1} 2^{z_i}$ such that ${\cal B}_{a,{z_r}}(t_{a,{z_r}}, w_{a,{z_r}},\gamma_{z_r}, \eps')$ starts at the end of block $b'$. Note that our criterion for switching to a new difference estimator is the estimated value increases by $\frac{(b-b')\eps}{8} \cdot Z_a = \frac{2^{z_r} \eps}{8}\cdot Z_a$. Then, since ${\cal B}_{a,{z_r}}(t_{a,{z_r}}, w_{a,{z_r}},\gamma_{z_r}, \eps')$ is frozen at time $w_{a,z_r}$, we have
\[\frac{2^{z_r}\eps}{8} \cdot Z_a \le Z_{a,z_r} \le F(1,w_{a,z_r}) - F(1,t_{a,z_r}) + \eps' \cdot F(1,t_i).\]
Therefore, since $Z_a$ is a $(1+\frac{\eps}{100})$-approximation to $F(1,t_a)$, we have 
\[ \frac{2^{z_r} \eps}{4} \cdot F(1,t_a) \le F(1,w_{a,z_r}) - F(1,t_{a,z_r}).\]
\end{proof}

The following lemma shows the correctness of our framework on non-adaptive streams.
Thus, it remains to argue that the inner randomness of each difference estimator is hidden from the adversary.
\begin{lemma}[Correctness on non-adaptive streams]
\lemlab{lem:correct:frame}
Suppose all subroutines ${\cal B}$ and ${\cal M}$ succeed. 
Then \algref{alg:frame} gives an $(1+\eps)$-approximation of the function $F$ at all times. 
\end{lemma}
\begin{proof}
We fix a time $t$ in round $a$ and block $b$. 
Consider the binary representation of $b$:
\begin{align*}
    b = \sum_{i=1}^r 2^{z_i}, ~ \mathrm{where} ~ z_1 > \cdots > z_r > 0, r = \mathrm{numbits}(b).
\end{align*}
Let $t_{a,i}$ denote the time when the counter $b$ is first set to $i$. Let $t_{a,0}$ denote the start time of round $a$. Let $\widehat{F}$ denote the current estimation.
Notice that our output $(1+\frac{b\eps}{8}) \cdot Z_a$ is within one block-size from $\bar{F}$. Therefore, we have
\begin{align}
\label{eq:err_output_frame}
\left|\left(1+\frac{b\eps}{8}\right) \cdot Z_a - F\right| \leq \left|\left(1+\frac{b\eps}{8}\right) \cdot Z_a - \widehat{F}\right| + | \widehat{F} - F | \leq \frac{\eps}{8} \cdot Z_a  +  | \widehat{F} - F | 
\end{align}
Now, we analyze the second term.
Recall that $\widehat{F}$ is the sum of $Z_a$ and each difference estimator, so by triangle inequality we have
\begin{align*}
    | \widehat{F} - F | \leq |Z_a - F(1,t_{a,0})| + \sum_{i=1}^r \Big| \Big( \widehat{F(1,t_{a,r})-F(1,t_{a,r-1})} \Big) - \Big(F(1,t_{a,r})-F(1,t_{a,r-1}) \Big) \Big|.
\end{align*}
For the first term, our double detector ${\cal M}_a(1,t_{a,0},\frac{\eps}{100})$ outputs a $(1+\frac{\eps}{100})$-approximation, hence the error is within $\frac{\eps}{100} \cdot F(1,t_{a,0})$. For the second term, by \lemref{lem:success:diff} our difference estimator  ${\cal B}_{a,r}(1,t_{a,r},t,\gamma_{r}, \eps')$ gives an estimation with additive error $\eps' \cdot F(1,t_{a,0})$. Then, we have
\begin{align*}
    | \widehat{F} - F | \leq \frac{\eps}{100} \cdot F(1,t_{a,0}) + r\eps' \cdot F(1,t_{a,0}).
\end{align*}
Now, by \lemref{lem:geo:bounds}, we know
\[\sum_{i=1}^r F(1,t_{a,r})-F(1,t_{a,r-1}) \ge \sum_{i=1}^r \frac{2^r \eps}{4}\cdot F(1,t_{a,0}).\]
By our double detector, $F$ at most increases by a factor of $3$ in one round. So, we have 
\[\sum_{i=1}^r \frac{2^r \eps}{4}\cdot F(1,t_{a,0}) \le 2F(1,t_{a,0}), \]
which implies $r \le \log \frac{8}{\eps}$. Therefore, by our choice of $\eps'$, we have $r\eps' \le \frac{\eps}{8}$, so we have $|\widehat{F}-F| \le \frac{\eps}{4}\cdot F(1,t_{a,0})$. 
This implies that
\[\left|\left(1+\frac{b\eps}{8}\right) \cdot Z_a - F\right| \le \eps \cdot F(1,t_{a,0}) \le \eps F.\]

\end{proof}

Now, we give the formal statement of our framework.
\begin{restatable}{framework}{framediff}
\framelab{frame:diff}
Let $n$ be the cardinality of the universe, let $k$ be the number of the servers, let $m$ be the stream length where $\log m = \O{\log n}$. 
Suppose the flip number of $F$ satisfies $\lambda_{1,m}(F) = \O{\log n}$
Suppose there exists a $(\eps, \eps)$-continuous difference estimator ${\cal B}$ of $F$ that takes communication cost
\begin{align*}
    \frac{1}{\eps} \cdot G_1(n, k, \eps) + G_2(n,k,\eps),
\end{align*}
a $(\gamma, \eps)$-static difference estimator ${\cal B}$ of $F$ that takes communication cost
\begin{align*}
    \frac{\gamma}{\eps^2} \cdot G_1(n, k, \eps) + G_2(n,k,\eps),
\end{align*}
and a $\eps$-double detector ${\cal M}$ of $F$ that takes communication cost
\begin{align*}
     \frac{1}{\eps^2} \cdot H_1(n, k, \eps) + H_2(n,k,\eps),
\end{align*}
where functions $G_1, G_2, H_1, H_2$ depend on the problem $F$. 
Then, there exists an adversarially robust streaming algorithm that outputs a $(1+\eps)$-approximation to $F$ at all times with probability $1-n^{-c}$. 
The total communication cost of the algorithm is 
\begin{align*}
    & ~ \O{\frac{1}{\eps^2} \log n \cdot H_1\left(n, k, \frac{\eps}{100}\right) + \log n \cdot  H_2\left(n,k,\frac{\eps}{100}\right)} + \\
& ~  \O{\frac{1}{\eps^2} \log n \log^3 \frac{1}{\eps} \cdot G_1(n, k, \eps') + \frac{1}{\eps} \log n \log \frac{1}{\eps} \cdot G_2(n, k, \eps')} + \O{\frac{k}{\eps}\log^2 n \log \frac{1}{\eps}}.
\end{align*}
where $\eps' = \O{\frac{\eps}{\log \frac{1}{\eps}}}$.
\end{restatable}
\begin{proof}
The correctness of the estimation is shown by \lemref{lem:correct:frame}. 
Observe that each time either counter $a$ or $b$ increase by $1$, \algref{alg:frame} switch to new subroutines ${\cal M}$ or ${\cal B}$ that have not been previously revealed to the adversary. Thus, the adversarial input is independent of the new subroutines, which gives us adversarially robustness.

Now, we bound the communication cost. Notice that we have $\O{\log n}$ rounds since $F$ at most doubles $\O{\log n}$ times in the stream. In each round, we first run a  $\frac{\eps}{100}$-double detector ${\cal M}$ of $F$ that takes communication cost
\begin{align*}
     \frac{1}{\eps^2} \cdot H_1\left(n, k, \frac{\eps}{100}\right) + H_2\left(n,k,\frac{\eps}{100}\right)
\end{align*}
At the end time $t$ of each round, the binary expression of counter $b$ is $b = \sum_{i=1}^r 2^{z_i}$ where $ r = \O{\log \frac{1}{\eps}}$ by the analysis of \lemref{lem:correct:frame}.
Now, let us suppose $b = 2^{r+1}$, which always gives an upper bound. 
Due to the upper bound on the number of different estimators implied by \lemref{lem:geo:bounds}, for blocks at granularity $0$, we run $\O{\frac{1}{\eps}}$ separate $(\eps', \eps')$-continuous difference estimators totally in one round; moreover, for blocks at granularity $z_i > 0$, we run $\O{\frac{1}{2^{z_i} \eps}}$ separate $(\gamma_{z_i}, \eps')$-static difference estimators in total for the round, where $\gamma_{z_i} = 2^{z_i} \eps$. 
Thus, for our choice of $\eps' = \O{\frac{\eps}{\log \frac{1}{\eps}}}$, the communication cost is
\begin{align*}
& ~    \O{\frac{1}{ \eps}} \Big( \frac{\eps}{\eps'^2} \cdot G_1(n, k, \eps', \delta') + G_2(n,k,\eps',\delta') \Big) \\
 & ~ +  \sum_{i=1}^{\O{\log \frac{1}{\eps}}} \O{\frac{1}{2^{z_i} \eps}} \Big( \frac{2^{z_i} \eps}{\eps'^2} \cdot G_1(n, k, \eps', \delta') + G_2(n,k,\eps',\delta') \Big) \\
= & ~ \O{ \frac{1}{\eps^2} \log^3 \frac{1}{\eps} \cdot G_1(n, k, \eps', \delta') + \frac{1}{\eps} \log \frac{1}{\eps} \cdot G_2(n, k, \eps', \delta')}.
\end{align*}
Notice that the central coordinator needs to inform all servers the frozen time for each difference estimator at the bottom level (e.g. difference estimators with constant accuracy). In each round, we run $\O{\frac{1}{\eps}}$ such different estimators. Therefore, this gives us additional $\O{\frac{k}{\eps}\log^2 n \log \frac{1}{\eps}}$ bits of communications. Combining them together gives us the result.
\end{proof}

\section{Adversarially Robust Distributed Streaming Algorithms}
\seclab{sec:robust:algs}
In this section, we show several applications of our difference estimator framework presented in \secref{sec:diff:est}. 
The framework requires both the construction of a double detector and a difference estimator. 
We propose these subroutines for the counting problem ($F_1$ estimation problem) in \secref{sec:robust:f1}, for the distinct elements estimation problem ($F_0$ estimation problem) in \secref{sec:robust:f0}, for the $F_2$ estimation problem in \secref{sec:robust:f2}, for the $F_p$ estimation problem for $p\ge2$ in \secref{sec:robust:fp:large}, and for the $L_2$-heavy-hitters problem in \secref{sec:l2:hh}. 

We remark that we construct $(\gamma,\eps)$-continuous difference estimator for all granularities $\gamma$ in some of the above problems, which generalize static estimators since they can be initiated when a sub-stream begins and finalized when it freezes. 
In particular, we achieve $(\gamma,\eps)$-continuous difference estimators for $F_1$ estimation in \secref{sec:robust:f1}, $F_0$-estimation in \secref{sec:robust:f0}, and $L_2$-heavy-hitters problem in \secref{sec:l2:hh}.
In our framework in \secref{sec:diff:est}, however, we use static estimators because constructing continuous $F_p$-difference estimators for $\gamma > \O{\eps}$ is significantly harder.

\subsection{Adversarially Robust \texorpdfstring{$F_1$}{F1} Estimation}
\seclab{sec:robust:f1}
In this section, we construct an adversarially robust algorithm for the counting problem, i.e., $F_1$ approximation. 
We begin by describing the counting algorithm in the classical non-robust distributed setting introduced by \cite{HuangYZ12}.
Let $p$ be a parameter to be determined later.
In their algorithm, each server $j$ sends the local count $n(j)$ to the coordinator with probability $p$ whenever it receives a new update.
The coordinator maintains an estimator $\widehat{n(j)}$ for each local count $n(j)$.
Here, $\widehat{n(j)}$ is defined as $\widetilde{n(j)} -1+\frac{1}{p}$ where $\widetilde{n(j)}$ is the last updated value from the server $j$, it is $0$ if $\widetilde{n(j)}$ does not exist.
The coordinator outputs $\widehat{n}=\sum_{j=1}^k \widehat{n(j)}$ as the estimate of the total count.
The next lemma states the guarantee of this estimator.
\begin{lemma}[\cite{HuangYZ12}] 
\lemlab{lem:counting}
For each $j$, the estimator satisfies $\Ex{\widehat{n(j)}}=n(j)$ and $\Var{\widehat{n(j)}} =\frac{1}{p^2}$.
\end{lemma}
Taking $p = \frac{\sqrt{k}}{\eps n}$ gives $\Var{\widehat{n(j)}} =\frac{\eps^2 n^2}{k}$, hence the total variance is $\eps^2 n^2$ summing across all the $k$ estimates.
This algorithm is not naturally robust because its correctness requires that the error of each estimate $\widehat{n(j)}$ be independent.
If an estimate $\widehat{n(j)}$ is exposed to the adversary, the inputs afterward are not independent of $\widehat{n(j)}$, which implies that the errors of other estimates are not independent of $\widehat{n(j)}$.
Then, the variance of the output $\widehat{n}=\sum_{j=1}^k \widehat{n(j)}$ may include covariance terms, which does not guarantee correctness.

We show that our difference estimator framework solves this problem while preserving the optimal communication bound $\frac{\sqrt{k}}{\eps}+k$.
Let $u$ be the prefix distributed stream, and let $v$ be the subsequent evolving distributed stream, we abuse the notation of $u$ and $v$ to also denote their total count.
Recall that we decompose the evolving stream $v$ into a sequence of sub-streams $v_1, \ldots, v_\beta$ with geometrically decreasing count: $u+v = u + \sum_{\ell=1}^{\beta} v_{\ell}$, where each $v_{\ell}$ is roughly $2^\ell \eps u$ and $\beta = \O{\log \frac{1}{\eps}}$.
First, we observe that if we run the algorithm from \cite{HuangYZ12} on each sub-stream $v_\ell$ with sampling probability $p = \frac{\sqrt{k}}{\eps u}$, then the variance of the estimator $\widehat{v_\ell(j)}$ on each server is $\frac{\eps^2 u^2}{k}$ by \lemref{lem:counting}, which suffices to give an estimate of the difference $v_\ell$ with additive error $\eps u$.
Notice that for each difference estimator, its output is hidden from the adversary before the end of its block, when the coordinator reveals the estimate to the adversary.
So, the estimates of the local differences $\widehat{v_\ell(j)}$ are independent of each other, implying the robust result.
However, in our standard difference estimator framework introduced in \secref{sec:diff:est}, the algorithm requires that each server knows the start time of each $v_\ell$, which would use $\frac{k}{\eps}$ bits of communication since there are roughly $\frac{1}{\eps}$ difference estimators in one round, and it is prohibitively large for our communication bound.

Instead, for our $F_1$ estimation algorithm, it is unnecessary to update the start time of each difference estimators to all the $k$ servers.
Instead, we apply a lazy update scheme to reduce communication. 
Recall that in the difference estimator framework, we divide the stream into $\O{\log n}$ rounds where the count doubles in each round.
Then, at the beginning of the round, the coordinator sends the count of the prefix stream $u$, and each server update the sampling probability as $p = \frac{\sqrt{k}}{\eps u}$.
Within a round, the coordinator sends the start time of $v_\ell$ to a server $j$ only if $j$ initiates the communication. 
That is, if $j$ sends an update to the coordinator in the duration of $v_\ell$, then the coordinator informs $j$ of the start time of $v_\ell$, and $j$ returns the actual value of $v_\ell(j)$.
This only loses constant factors to the communication cost.
In addition, this does not change the distribution of the estimator of $v_\ell(j)$, since each server define the sampling probability with the correct count of $u$, and thus ensures the correctness under the adversarially robust setting.
We next present the formal theorem statement and analysis of the above descriptions.

\begin{theorem}[Adversarially robust distributed streaming counting]
Let $n$ be the total count of items in a distributed stream $S$, let $k$ be the number of servers, and let $\eps \in (0,1)$ be the accuracy parameter satisfying $\frac{1}{\eps^2} \ge k$. 
Then, there exists an adversarially robust algorithm that outputs a $(1+\eps)$-approximation to $n$ at all times with probability $1-\frac{1}{\poly(n)}$. 
This algorithm uses $\O{(k + \frac{\sqrt{k}}{\eps}) \log^2n \log \frac{1}{\eps}}$ bits of communication.
\end{theorem}
\begin{proof}
First, we remark that \lemref{lem:counting} implies the existence of a $\eps$-double detector that communicates $\O{(k+\frac{\sqrt{k}}{\eps}) \log n}$ bits (c.f. Theorem 2.1 in \cite{HuangYZ12}). 

Next, we consider a fixed difference estimator block. 
Our lazy update sketch guarantees that in each round, each server samples an update to the coordinator with probability $p = \frac{\sqrt{k}}{\eps u}$, and the coordinator receives the true value of $v_j$ if it receives a sample from $j$. 
Then, the distribution of the communication transcript between the coordinator and the servers is the same as running \cite{HuangYZ12} on $v$ with $p = \frac{\sqrt{k}}{\eps u}$.
Hence, inside each separate difference estimator block, the error of each estimate $\widehat{v_j}$ is independent, as the estimates are hidden from the adversary.
Thus, the variance of each estimator $\widehat{v_j}$ is $\frac{\eps^2 u^2}{k}$, and the variance of the estimator $\widehat{v}$ is $\eps^2 u^2$, which follows from the analyze framework in \lemref{lem:counting} and \cite{HuangYZ12}.
Since $\widehat{v}$ is unbiased, by standard concentration inequalities, averaging $\O{\log n}$ instances gives an estimate of the difference $v$ with additive error $\eps u$ with high probability.
Therefore, our algorithm provides a $(\gamma, \eps)$-continuous difference estimator, then by \frameref{frame:diff}, there is an robust algorithm that gives a $(1+\eps)$-approximation of the total count.

Now, we compute the communication cost of a $(\gamma,\eps)$-difference estimator. 
With our assumption on $v \le \gamma u$, we sample each update with probability $\frac{\sqrt{k}}{\eps u} \le \frac{\gamma \sqrt{k}}{\eps v}$.
Since there are $v$ updates in total and we run $\O{\log n}$ instances, the communication cost of our $(\gamma,\eps)$-continuous difference estimator is $\O{\frac{\gamma \sqrt{k}}{\eps } \log n}$.
Notice that all $\frac{1}{\eps} $ difference estimators in the same round are implemented with the same sampling probability $\frac{\sqrt{k}}{\eps u}$, where $u$ is the count at the beginning of the round. So, we only need to inform each server of the sampling probability once a round, which saves the $\O{\frac{k}{\eps}\log n}$ factor in our difference estimator framework in \frameref{frame:diff}. Therefore, following the calculation in \frameref{frame:diff}, the total communication of the adversarially robust algorithm is $\O{(k+\frac{\sqrt{k}}{\eps} ) \log^2n \log \frac{1}{\eps}}$ bits.
\end{proof}

\subsection{Adversarially Robust \texorpdfstring{$F_0$}{F0} Estimation}
\seclab{sec:robust:f0}
This section provides an adversarially robust algorithm for $F_0$ estimation. 
We first describe the subsampling scheme to count the distinct elements in \cite{CormodeMY11}. 
The algorithm maintains roughly $\frac{1}{\eps^2}$ samples from the distinct elements. 
When there are more than $\frac{1}{\eps^2}$ samples, it discards each item with probability $\frac{1}{2}$. 
We iteratively run this procedure until the stream ends so that the number of samples does not pass the threshold. 
Then at the end of the stream, we retrieve the estimate by rescaling the number of surviving samples by $2^t$, where $t$ is the number of times we subsample. 
To implement the subsampling scheme, each server hashes the incoming elements to $[n]$ by a unique hashing function $f:[n]\to [n]$, and it computes the largest coordinate $q$ in its binary expression. 
The centralized coordinator keeps a counter $t$ to count the rounds of the subsample procedure. 
Whenever the algorithm increases a subsampling layer, $t$ increases by $1$ and we discard all the elements with $q < t$.
We state the formal theorem for the non-robust distributed algorithm to estimate the $F_0$ norm in a data stream as follows.

\begin{theorem}[\cite{CormodeMY11}] \thmlab{thm:non_robust_f0}
Let $n$ be the number of items in the universe, let $k$ be the total number of servers, and let $\eps \in (0,1)$ be the accuracy parameter. 
Then, there exists an algorithm that, with high probability, outputs a $(1+\eps)$-approximation to the $F_0$ estimation problem at all times $t$, using $\O{\frac{k}{\eps^2} \log^2 n + \frac{1}{\eps^2} \log n \log \frac{1}{\eps}}$ bits of communication. 
\end{theorem}

We extend the above algorithm to obtain a difference estimator, which is analogous to the construction introduced by \cite{WoodruffZ21b} in the centralized streaming setting. 
The protocol first runs a subsampling procedure on the static vector $u$ to some level $t$ with $\Theta(\frac{\gamma}{\eps^2})$ survivors. 
Then, it runs the same subsampling procedure on $v - u$ and only counts the additional elements surviving at level $t$. 
Recall that our assumption is $F_0(v) - F_0(u) = \gamma F_0(u)$, then in expectation there are $\Theta(\frac{\gamma^2}{\eps^2})$ distinct elements in $v$ but not in $u$ in the $t$-th subsampling layer. 
Therefore, rescaling by $2^t$ gives us a $(1+\frac{\gamma}{\eps})$-approximation to $F_0(v) - F_0(u)$, which implies an additive error $\eps \cdot F_0(u)$ as we desired. 
We now present the formal analysis of this algorithm.
We first show the correctness of the protocol in the following statement. The proof is analogous to the analysis in Lemma 6.2 in \cite{WoodruffZ21b}.
\begin{lemma}
\lemlab{lem:f0_difference_est_err}
Suppose that $F_0(u + v) - F_0(u) \le \gamma F_0(u)$, then the difference estimator gives an estimation to $F_0(u + v) - F_0(u)$ within additive error $\eps F_0(u)$ with probability $1-\frac{1}{\poly(n)}$.
\end{lemma}
\begin{proof}
Consider the deepest layer $t$ where we subsample each item with probability $\frac{1}{2^t}$, we have $ \frac{F_0(u)}{2^t} = C\cdot \frac{\gamma}{\eps^2}$.
By our assumption on $F_0(v) - F_0(u) \le \gamma F_0(u)$, in expectation we have $\frac{F_0(v) - F_0(u)}{2^t} = C \cdot \frac{\gamma^2}{\eps^2}$.
Hence, let $X$ denote the survivors from $F_0(v) - F_0(u)$ at $t$, we have 
\begin{align*}
    & \Ex{2^t \cdot X} = F_0(v) - F_0(u) \\
    & \mathrm{Var} (2^t \cdot X) = 2^t \cdot (F_0(v) - F_0(u)) \leq \frac{\eps^2 \cdot F_0(u)}{C \gamma} \cdot (F_0(v) - F_0(u)) \leq \frac{\eps^2}{C} \cdot (F_0(u))^2.
\end{align*}
Therefore, for a sufficiently large $C$, by Chebyshev's inequality the additive error of our output is $\eps F_0(u)$ with probability at least $0.5$. Averaging $\O{\log n}$ repetitions achieves high probability.
\end{proof}
Next, we upper bound the communication cost of this protocol. 
\begin{lemma}
\lemlab{lem:f0_difference_est_comm}
The communication cost of the difference estimator is $\O{\frac{k \gamma}{\eps^2} \log^2 n}$ bits.
\end{lemma}
\begin{proof}
It takes $\O{k \log n}$ communications to send the hashing function to each site. In each subsample layer, we sample $\Theta (\frac{\gamma}{\eps^2})$ elements from the stream. Since we have $\O{\log n}$ such layers and $k$ servers, the communication cost is $\O{\frac{\gamma}{\eps^2} \cdot k \log n + k \log n}$ bits in total. 
To reach a high probability, we run $\O{\log n}$ repetitions, which gives an additional logarithmic factor.
\end{proof}
Combining \lemref{lem:f0_difference_est_err} and \lemref{lem:f0_difference_est_comm}, we have the following statement of our difference estimator. 
\begin{theorem}
\thmlab{thm:f0:de}
Let $n$ be the number of items in the universe, let $k$ be the number of servers, and let $\eps \in (0,1)$ be the accuracy parameter.
Then, there exists a distributed $(\gamma, \eps)$-continuous difference estimator for the distinct element problem that uses $\O{\frac{k \gamma}{\eps^2} \log^2 n}$ bits of communication. 
\end{theorem}

Given the non-robust distributed $F_0$ tracker and our difference estimator, we show the following theorem for adversarially robust $F_0$-estimation.
\begin{theorem}
Let $S$ be a distributed stream, let $n$ be the number of items in the universe, let $k$ be the number of servers, and let $\eps \in (0,1)$ be the accuracy parameter. 
Then, there exist an adversarially robust distributed monitoring algorithm that outputs a $(1+\eps)$-approximation to the distinct element problem in stream $S$. 
The algorithm uses $\O{\frac{k}{\eps^2} \log^3 n \log^3 \frac{1}{\eps}}$ bits of communication.
\end{theorem}
\begin{proof}
By \thmref{thm:non_robust_f0}, we have a $\eps$-double detector that communicates $\O{k \log^2 n + \frac{1}{\eps^2} \log n \log \frac{1}{\eps}}$ bits. 
By \thmref{thm:f0:de}, we have a $(\gamma, \eps)$-continuous difference estimator that communicates $\O{\frac{k \gamma}{\eps^2} \log^3 n}$ bits. 
Therefore, by \frameref{frame:diff}, there is an adversarially robust algorithm that uses $\O{\frac{k}{\eps^2} \log^3 n \log^3 \frac{1}{\eps}}$ bits of communication. 
\end{proof}

\subsection{Adversarially Robust \texorpdfstring{$F_2$}{F2} Estimation}
\seclab{sec:robust:f2}
In this section, we provide an adversarially robust algorithm for $F_2$ estimation. 
We note that the result of $F_2$ estimation is a direct corollary from the $F_p$ estimation algorithm in \secref{sec:robust:fp:large}.
Here, we present a simpler difference estimator using the inner product of $u$ and $v$.
Observe that we can decompose the difference as the sum of the inner product $\langle u,v \rangle$ and the $F_2$-moment of $v$:
\[\|u+v\|_2^2 - \|u\|_2^2 = 2\langle u,v \rangle + \|v\|_2^2.\]
Then, to estimate the inner product, we simply obtain $L_2$ samples from $u$ and compute $u_i \cdot v_i \cdot \frac{\|u\|_2^2}{u_i^2}$. 
This gives an unbiased estimate with variance bounded by $\gamma \cdot \|u\|_2^4$, and therefore, averaging $\O{\frac{\gamma}{\eps^2}}$ such estimators would give an estimate with additive error $\eps \cdot \|u\|_2^2$. 
It may seem unclear how to obtain the estimation of $\|v\|_2^2$ due to the lack of a continuous monitoring algorithm for $F_2$-estimation with optimal $\frac{1}{\eps}$ factor in the communication cost. 
However, recall that in our difference estimator framework, we only need a continuous difference estimator for an evolving $v$ at granularity $0$, where a constant-multiplicative estimate of $\|v\|_2^2$ is sufficient. 
For the static difference estimator at a higher granularity, we can simply run the one-shot $F_2$-estimation algorithm in \cite{EsfandiariKMWZ24} for $v$ with a higher accuracy.

Now, we state another result of distributed counting algorithm that gives a constant-multiplicative estimation of the count, which is frequently used in our difference estimator. 
We remark that this result is a corollary of the estimator in \lemref{lem:counting}.

\begin{theorem}[\cite{HuangYZ12}] 
\thmlab{thm:counting}
Let $n$ be the total count, let $k$ be the number of servers.
There is a procedure \textsc{Counting} such that, at any given time $t$, it gives an unbiased estimate of the current count, and its variance is bounded by $\O{\left(\frac{n}{100}\right)^2}$. 
Moreover, the estimate is a $(1+\frac{1}{100})$-approximation to the current count, with probability at least $1-\frac{1}{\poly(n)}$. 
The total communication of the algorithm is $\O{k \cdot \log^2 n}$ bits. 
\end{theorem}

We next state the following result of perfect $L_2$ sampler, which is a direct corollary of the perfect $L_p$ sampler, in \thmref{thm:lp:sampler}. 

\begin{theorem}[Perfect $L_2$ sampler, c.f. \thmref{thm:lp:sampler}]
\thmlab{thm:lp:sampler:prev}
Let vector $f \in \mathbb{R}^{n}$ be defined by a distributed data stream $S$, and let $k$ denote the number of servers. 
Then, there is an algorithm that gives a perfect $L_2$ sampler of $f$, using $k\cdot \polylog(n)$ bits of communication in expectation.
\end{theorem}

Next, we present a subrountine that produces unbiased $F_2$ estimates, which is essential to our difference estimator.
We defer its analysis to \secref{sec:robust:fp:large}.

\begin{theorem}[Unbiased $F_2$ estimation, c.f. \thmref{thm:lp:est} and \thmref{thm:fp:est:stat}]
\thmlab{thm:lp:est:prev}
Let vector $f \in \mathbb{R}^{n}$ be defined by a distributed data stream $S$, and let $k$ denote the number of servers. 
Then, there is an algorithm that gives an estimate of $\widehat{F}$ to $\|f\|_2^2$, which satisfies $\Ex{\widehat{F}} = (1+\O{\eps})\cdot \|f\|_2^2$ and $\Ex{\left(\widehat{F}\right)^2} = \O{\|f\|_2^{4} }$. 
The algorithm uses $(k+\frac{\sqrt{k}}{\eps})\cdot \polylog(n)$ bits of communication. 
Moreover, if $f$ is a static vector, the algorithm gives an estimate of $\widehat{F}$ to $\|f\|_2^2$, which satisfies $\Ex{\widehat{F}} = \|f\|_2^2$ and $\Ex{\left(\widehat{F}\right)^2} = \O{\|f\|_2^{4} }$. This algorithm uses $k \polylog(n)$ bits of communication.
\end{theorem}

The following statement shows the $F_2$-estimation algorithm with constant-factor approximation.
\begin{theorem}[Constant-multiplicative $F_2$-estimation, c.f. \thmref{thm:fp:cts:const} and \thmref{thm:fp:stat}]
\thmlab{thm:lp:cts:const:prev}
Let vector $f \in \mathbb{R}^{n}$ be defined by a distributed data stream $S$, and let $k$ denote the number of servers. 
There is an algorithm that gives a $1.01$-approximation to the $F_2$-moment of $f$ at all times in the stream using $k\cdot \polylog(n)$ bits of communication.
In addition, if $f$ is static, there is an algorithm that gives a $(1+\eps)$-approximation to the $F_2$-moment of $f$ at all times in the stream using $\frac{k^{p-1}}{\eps^2}\cdot \polylog(n)$ bits of communication.
\end{theorem}
We describe our difference estimator for robust $F_2$ moment estimation in full in \algref{alg:norm:diff:l2}. 
\begin{algorithm}[!htb]
\caption{Difference estimator for $F_2$-moment}
\alglab{alg:norm:diff:l2}
\begin{algorithmic}[1]
\Require{Accuracy $\eps\in(0,1)$, parameter $\gamma>0$, fixed vector $u$, subsequent vector $v$ that arrives in a distributed data stream}
\Ensure{Difference estimator for $\|u+v\|_2^2 - \|u\|_2^2$}
\State{$S\gets\O{\frac{\gamma}{\eps^2}}$}
\State{Let $\widehat{U}$ be an unbiased estimate of $\|u\|_2^2$ with variance $\O{\|u\|_2^{4}}$} \Comment{See \thmref{thm:lp:est:prev}}
\State{Let $\widehat{V}$ be an $\O{\frac{\eps}{\gamma}}$-approximation of $\|v\|_2^2$} \Comment{See \thmref{thm:lp:cts:const:prev}}
\For{$s\in[S]$}
\State{$L_2$ sample a coordinate $w_s$ from $u$} \Comment{See \thmref{thm:lp:sampler:prev}}
\State{$\widehat{v_{w_s}} \gets \textsc{Counting}(v_{w_s})$} \Comment{See \thmref{thm:counting}}
\State{$W_s \gets \langle \widehat{v_{w_s}} ,u_{w_s}\rangle\cdot\frac{(\widehat{U})^2}{u_{w_s}^2}$}
\EndFor
\State{\Return $2\cdot\frac{1}{S}\sum_{s\in[S]}W_s+\widehat{V}$}
\end{algorithmic}
\end{algorithm}

The following lemma shows that our estimator $W_s$ in \algref{alg:norm:diff:l2} for the inner product $\langle v,u \rangle$ is unbiased, up to a small additive term. 
\begin{lemma}
\lemlab{lem:diff:two:ex}
$\Ex{W_s}=\langle v,u\rangle+\frac{1}{\poly(n)}$
\end{lemma}
\begin{proof}
By the correctness of the $L_2$ sampler, we have that 
\[\PPr{w_s=i}=\frac{u_i^2}{\|u_i\|_2^2}+\frac{1}{\poly(n)}.\]
Therefore, since the \textsc{Counting} algorithm gives an unbiased estimation to $v_{w_s}$, we have
\[\Ex{W_s}=\sum_{i\in[n]} u_iv_i\cdot\frac{\Ex{\widehat{U}^2}}{u_i^2}\cdot\left(\frac{u_i^2}{\|u\|_2^2}+\frac{1}{\poly(n)}\right).\]
Since $\Ex{\widehat{U}^2}=\|u\|_2^2$, then 
\[\Ex{W_s}=\sum_{i\in[n]}u_iv_i+\frac{1}{\poly(n)}=\langle v,u\rangle+\frac{1}{\poly(n)},\]
as desired. 
\end{proof}
We next upper bound the variance of our estimator $W_s$.
\begin{lemma}
\lemlab{lem:diff:two:var}
$\Ex{W_s^2}\le\O{\gamma\cdot\|u\|_2^4}+\frac{1}{\poly(n)}$
\end{lemma}
\begin{proof}
By the correctness of the $L_2$ sampler, we have that 
\[\PPr{w_s=i}=\frac{u_i^2}{\|u_i\|_2^2}+\frac{1}{\poly(n)}.\]
Therefore, since the estimation $\widehat{v_{w_s}}$ by \textsc{Counting} algorithm has second moment $\O{v_{w_s}^2}$, c.f. \thmref{thm:counting}, we have
\[\Ex{W_s^2}=\O{1} \cdot \sum_{i\in[n]} (u_iv_i)^2\cdot\frac{(\widehat{U})^4}{u_i^4}\cdot\left(\frac{u_i^2}{\|u\|_2^2}+\frac{1}{\poly(n)}\right).\]
Since $\Ex{\widehat{U}^4}\le\O{\|u\|_2^{4}}$, then 
\[\Ex{W_s^2}\le\O{\|u\|_2^2 \cdot \sum_{i\in[n]}v_i^2}+\frac{1}{\poly(n)}\le\O{\gamma\cdot\|u\|_2^4}+\frac{1}{\poly(n)},\]
since $\|v\|_2^2\le\gamma\cdot\|u\|_2^2$. 
\end{proof}
We now prove the correctness of our difference estimator for $F_2$ estimation. 
\begin{lemma}
\lemlab{lem:f2:cor}
With probability $0.99$, \algref{alg:norm:diff:l2} outputs an additive $\eps\cdot\|u\|_2^2$ approximation to $\|u+v\|_2^2-\|u\|_2^2$. 
In addition, taking the median of $\O{\log n}$ repetitions would give a correct result with high probability.
\end{lemma}
\begin{proof}
We have 
\[\|u+v\|_2^2-\|u\|_2^2=2\langle u,v\rangle+\|v\|_2^2.\]
By \lemref{lem:diff:two:ex} and \lemref{lem:diff:two:var}, we have $\Ex{W_s}=\langle v,u\rangle+\frac{1}{\poly(nm)}$ and $\Ex{W_s^2}\le\O{\gamma\cdot\|u\|_2^4}+\frac{1}{\poly(nm)}$. 
Let $W=2\cdot\frac{1}{S}\sum_{s\in[S]}W_s$. 
Then we have
\[\Ex{W}=\frac{2}{S}\sum_{s\in[S]}\Ex{W_s}=2\langle v,u\rangle+\frac{1}{\poly(nm)}\]
and
\[\Ex{W^2}\le\frac{4}{S^2}\sum_{s\in[S]}\Ex{W_s^2}\le\frac{2}{S}\cdot\O{\gamma\cdot\|u\|_2^4}+\frac{1}{\poly(n)}\le\O{\eps\cdot\|u\|_2^4}+\frac{1}{\poly(n)},\]
for $S=\O{\frac{\gamma}{\eps^2}}$. 
By Chevyshev's inequality, it follows that with probability at least $0.99$,
\[|W-2\langle u,v\rangle|\le\frac{\eps}{2}\cdot\|u\|_2^2+\frac{1}{\poly(n)}.\]
The desired claim then follows from the fact that $\widehat{V}$ is an $\O{\frac{\eps}{\gamma}}$-approximation of $\|v\|_p^p$, which itself is at most $\gamma\cdot\|u\|_p^p$, i.e., the additive error is $\O{\frac{\eps}{\gamma}\|v\|_p^p} = \O{\eps \|u\|_p^p}$. 
\end{proof}
We next upper bound the communication cost of our difference estimator for $F_2$ estimation. 
\begin{theorem}
\thmlab{thm:f2:de}
Let $n$ be the number of items in the universe, $k$ denote the total number of servers, and $\eps \in (0,1)$ denote some accuracy parameter. 
Let $u$ be a fixed vector and $v$ be a subsequent vector in the stream.
There is a distributed $(\eps, \eps)$-continuous $F_2$ difference estimator which uses $\frac{1}{\eps} \cdot k\cdot\polylog(n)$ bits of communication in expectation. 
Also, there is a distributed $(\gamma, \eps)$-static $F_2$ difference estimator which uses $\frac{\gamma}{\eps^2} \cdot k \cdot \polylog(n)$ bits of communication in expectation.  
\end{theorem}
\begin{proof}
Given the correctness result in \lemref{lem:f2:cor}, it suffices to show the communication cost of \algref{alg:norm:diff:l2}. 
For each instance $s \in S$, where $S = \O{\frac{\gamma}{\eps^2}}$, we use $k \polylog(n)$ bits in expectation to run the $F_2$-sampler (see \thmref{thm:lp:sampler}), and we use $\O{k \log^2 n}$ to run the \textsc{Counting} algorithm (see \thmref{thm:counting}). 
This sums up to $\frac{\gamma}{\eps^2}\cdot k\cdot \polylog(n)$ bits of communication. In addition, we use $k\cdot\polylog(n)$ bits to estimate $\|u\|_2^2$ (see \thmref{thm:lp:est:prev}). 
Now, we consider the estimate of $\|v\|_2^2$. For the continuous version, we have $\gamma = \eps$, so we only need a constant approximation to $\|v\|_2^2$. 
Thus, we use $\frac{1}{\eps} \cdot k \cdot \polylog(n)$ bits in total. For the static version, we use $\frac{\gamma^2}{\eps^2}\cdot k\cdot \polylog(n)$ to estimate $\|v\|_2^2$, which is dominated by other terms. 
Thus, we use $\frac{\gamma}{\eps^2} \cdot k \cdot \polylog(n)$ bits in total.
\end{proof}

Given the oblivious strong $F_2$ tracker and our difference estimator, we can give the following guarantees for adversarially robust $F_2$-estimation.
\begin{theorem}
Let $S$ be a distributed stream, $n$ be the number of items in the universe, $k$ denote the total number of servers, and $\eps \in (0,1)$ denote some accuracy parameter.
There is an adversarially robust distributed streaming algorithm that gives a $(1+\eps)$-approximation to the $F_2$-moment of the vector induced by stream $S$, while using $\frac{k}{\eps^2} \polylog\left(n, \frac{1}{\eps}\right)$ bits of communication with probability $0.99$.
\end{theorem}
\begin{proof}
By \thmref{thm:lp:cts:const:prev}, we have an algorithm that flags whenever the $F_2$-moment doubles. 
When the algorithm flags, we can retrieve a $(1+\eps)$-approximation by running the static algorithm in \cite{EsfandiariKMWZ24}. 
This gives an $\eps$-double detector that communicates $\frac{k}{\eps^2}\cdot\polylog\left( n, \frac{1}{\eps}\right)$ bits.
By \thmref{thm:f2:de}, there is a $(\eps, \eps)$-continuous difference estimator that uses $\frac{1}{\eps} \cdot k \polylog(n)$ bits of communication. 
Also, there is a $(\gamma, \eps)$-static difference estimator that uses $\frac{\gamma}{\eps^2} \cdot k \cdot \polylog(n)$ bits of communication. 
Therefore, by \frameref{frame:diff}, there is an adversarially robust algorithm that communicates in $\frac{k }{\eps^2}\cdot\polylog\left( n, \frac{1}{\eps}\right)$ bits.
\end{proof}

\subsection{Adversarially Robust \texorpdfstring{$F_p$}{Fp} Estimation, \texorpdfstring{$p\ge2$}{p>2}}
\seclab{sec:robust:fp:large}
In this section, we present an adversarially robust distributed monitoring protocol for $F_p$ moment estimation, with $p\ge2$. 
We first describe how we obtain the $(\gamma,\eps)$ difference estimator for $F_p(u+v) - F_p(u)$. 
We obtain a mixed $L_p$ sample $i\in[n]$ with probability $\frac{u_i^p}{2\|u\|_p^p} + \frac{v_i^p}{2\|v\|_p^p}$. 
Then, we construct our estimator for the difference as
\[s = \frac{1}{p_i} [(u_i+v_i)^p - u_i^p].\]
Suppose that we have an unbiased estimate for $s$ with variance $\O{s^2}$. 
Then averaging $\frac{\gamma}{\eps^2}$ instances of the estimator gives us an estimate for the difference with additive error $\eps \cdot F_p(u)$. 
In the following sections, we formalize the framework for the $F_p$ difference estimator, and we provide algorithms to estimate $\frac{1}{p_i}$ and $(u_i+v_i)^p - u_i^p$ with near-optimal communication.

We again emphasize that we only need a continuous algorithm for difference estimators at granularity $0$, i.e., corresponding to $F_p(u+v)-F_p(u)\le\eps\cdot F_p(u)$. 
So, we introduce the $(\eps,\eps)$-continuous $F_p$ difference estimator and the $(\gamma, \eps)$-static $F_p$ difference estimator.

\subsubsection{Outline for the \texorpdfstring{$F_p$}{Fp} Difference Estimator}
We begin by outlining our difference estimator. 
We first consider the mixed $L_p$ sampling subroutine. 
Intuitively, it detects the important coordinates for both $u$ and $v$, and thereby has a small variance. 
Now, we state the correctness of several important subroutines, whose analysis is deferred to the following sections. 
We remark that we show two versions of the algorithm for estimating $\frac{1}{p_i}$. 
We present a result for evolving $v$ with constant-fractional bias, which is used to construct the $(\eps,\eps)$-continuous difference estimator, where a constant-multiplicative approximation suffices. 
We also present an unbiased estimator for static $v$, which is used to construct the $(\gamma,\eps)$-static difference estimator.
\begin{theorem}
\thmlab{thm:distributed_algs}
Let there be $k$ servers, each server $j\in[k]$ with an $n$-dimensional vector $f(j) = [f_1(j), \cdots,f_n(j)]$ that arrives in an insertion-only stream, and let $f = \sum_{j\in[k]} f(j)$. 
Then for $p\ge2$, we have the following algorithms
\begin{itemize}
    \item 
    There is an algorithm \textsc{LpSampler} that outputs $i \in [n]$ so that for any $j\in[n]$,
    \[\PPr{i = j} = \frac{f_j^p}{\|f\|^p_p} \pm n^{-c},\]
    at all times. 
    This protocol uses $k^{p-1} \cdot \polylog(n)$ bits of communication in expectation.
    \item 
    Fix a time $t$. 
    Let $u$ be the vector at time $t$ and let $v$ be the subsequent vector. 
    Let $p_i = \frac{u_i^p}{2\|u\|_p^p} + \frac{v_i^p}{2\|v\|_p^p}$ denote the sampling probability. 
    Fix an $i \in [n]$. 
    Then there is an algorithm \textsc{InverseProbEst} that gives an estimation $\widehat{\frac{1}{p_i}}$ to $\frac{1}{p_i}$ such that $\Ex{\widehat{\frac{1}{p_i}}} = \frac{1}{p_i} \left(1\pm\O{\frac{1.01}{C_{\ref{lem:inverse:prob:corr}}}}\right)$ and $\Var{\frac{1}{p_i}} = \frac{1}{p_i^2} \cdot \polylog(n)$ at all times, with probability $1-\frac{1}{\poly(n)}$. 
    This protocol uses $k^{p-1} \cdot \polylog(n)$ bits of communication. 
    For static $v$, the algorithm uses the same amount of communication, and the estimate satisfies $\Ex{\widehat{\frac{1}{p_i}}} = \frac{1}{p_i} \pm n^{-c}$ and $\Var{\frac{1}{p_i}} = \frac{1}{p_i^2} \cdot \polylog(n)$.
    \item 
    Fix a time $t$. 
    Let $u$ be the vector at time $t$ and let $v$ be the subsequent vector. 
    Fix an index $i \in [n]$. 
    There is an algorithm \textsc{DifferenceEst} that gives an estimation $\widehat{(u_i+v_i)^p - u_i^p}$ to $(u_i+v_i)^p - u_i^p$ such that $\Ex{\widehat{(u_i+v_i)^p - u_i^p}} = (u_i+v_i)^p - u_i^p \pm n^{-c}$ and $\Var{\widehat{(u_i+v_i)^p - u_i^p}} = \O{ (u_i+v_i)^{2p}\cdot \log^2 n }$ at all times with probability $1-\frac{1}{\poly(n)}$. 
    Moreover, if $u_i > v_i$, we have $\Var{\widehat{(u_i+v_i)^p - u_i^p}} = \O{ \left((u_i+v_i)^p - u_i^p \right)^2 \cdot \log n }$. 
    This protocol uses $k \cdot \polylog(n)$ bits of communication.
\end{itemize}
\end{theorem}
\begin{proof}
See \thmref{thm:lp:sampler} for the \textsc{LpSampler}. 
See \thmref{thm:lp:est} for the \textsc{LpEstUnbiased}.
See \lemref{lem:inverse:prob:corr}, \lemref{lem:inverse:prob:comm}, and \corref{cor:inverse:prob} for \textsc{InverseProbEst}. 
See \lemref{lem:diff:index} and \lemref{lem:diff:comm} for \textsc{DifferenceEst}.
\end{proof}
In light of the subroutines in \thmref{thm:distributed_algs}, we present our difference estimator for $F_p$ moment estimation for $p\ge2 $ in \algref{alg:norm:diff}. 
\begin{algorithm}[!htb]
\caption{Difference estimator for $F_p$ moment estimation}
\alglab{alg:norm:diff}
\begin{algorithmic}[1]
\Require{Accuracy $\eps\in(0,1)$, parameter $\gamma>0$, fixed vector $u$, subsequent vector $v$ that arrives in a distributed data stream}
\Ensure{Difference estimator for $\|u+v\|_p^p - \|u\|_p^p$}
\State{$R \gets \frac{\gamma}{\eps^2} \cdot \polylog(n)$}
\For{$r \in [R]$}
\State{$s^r \gets 0$, $ U \gets {\cal U}(0,1)$}
\If{$U<\frac{1}{2}$}
\State{$i \gets \textsc{LpSampler}(u)$} \Comment{Run the static version.}
\State{$\widehat{\frac{1}{p_i}} \gets \textsc{InverseProbEst}(u,v,i,u_i,\eps)$}
\Comment{Query the actual value of $u_i$}
\State{$\widehat{(u_i+v_i)^p - u_i^p} \gets \textsc{DifferenceEst}(u,v,i,u_i)$}
\State{$s^r \gets s^r + \widehat{\frac{1}{p_i}} \cdot \left(\widehat{(u_i+v_i)^p - u_i^p}\right)$}
\EndIf
\If{$U>\frac{1}{2}$}
\State{$j \gets \textsc{LpSampler}(v)$} 
\State{$\widehat{\frac{1}{p_j}} \gets \textsc{InverseProbEst}(u,v,j,u_j,\eps)$}
\State{$\widehat{(u_j+v_j)^p - u_j^p} \gets \textsc{DifferenceEst}(u,v,j,u_j)$}
\State{$s^r \gets s^r + \widehat{\frac{1}{p_j}} \cdot \left(\widehat{(u_j+v_j)^p - u_j^p}\right)$}
\EndIf
\EndFor
\State{\Return $\frac{1}{R} \sum_{r \in [R]} s^r$}
\end{algorithmic}
\end{algorithm}
We now upper bound the error of our difference estimator. 
\begin{lemma}
\lemlab{lem:fp:de:cor}
For $\|v\|_p^p \le \gamma \|u\|_p^p$ and $\|u+v\|_p^p - \|u\|_p^p \le \gamma \|u\|_p^p$, \algref{alg:norm:diff} gives an estimation to $\|u+v\|_p^p - \|u\|_p^p$ with additive error $\eps \cdot \|u\|_p^p$, with probability $1-\frac{1}{\poly(n)}$. 
\end{lemma}
\begin{proof}
Let $p_i$ denote the sampling probability of index $i$. 
Let $S := \|u+v\|_p^p - \|u\|_p^p = \sum_{i \in [n]} (u_i+v_i)^p - u_i^p$. 
Notice that our estimator is defined as
\[s_r = \widehat{\frac{1}{p_i}} \cdot \left(\widehat{(u_i+v_i)^p - u_i^p}\right).\]
Therefore, for the continuous version, we have
\begin{align*}\Ex{s_r} &= \sum_{i \in [n]} (p_i \pm n^{-c}) \cdot \left(\frac{1}{p_i}\right)  \cdot \left(1\pm\O{\frac{1.01}{C_{\ref{lem:inverse:prob:corr}}}}\right)  \cdot \left((u_i+v_i)^p - u_i^p \right) \\
&= \sum_{i \in [n]} p_i  \cdot \frac{1}{p_i} \cdot ((u_i+v_i)^p - u_i^p) \cdot \left(1\pm\O{\frac{1.01}{C_{\ref{lem:inverse:prob:corr}}}}\right) = S \cdot \left(1\pm \O{\frac{1.01}{C_{\ref{lem:inverse:prob:corr}}}}\right)
\end{align*}
Note that we only run continuous difference estimator at granularity $0$, where we have $\gamma = \eps$. Taking a sufficiently large constant $C_{\ref{lem:inverse:prob:corr}}$, by our assumption on $S \le \gamma \|u\|_p^p = \eps \|u\|_p^p$, we have 
\[\Ex{s_r} = S \cdot \left(1 \pm \frac{1}{2}\right) = S \pm \frac{\eps }{2} \cdot \|u\|_p^p .\]
For the static version, the same equation holds since the bias is bounded by $\O{n^{-c}}$.

Now, let $S_1 := \{ i, u_i\ge v_i \}$, $S_2 := \{ i, u_i< v_i \}$, we have
\begin{align*}
\Ex{s_r^2} &= \sum_{i \in S_1} (p_i \pm n^{-c}) \cdot \Ex{\widehat{\frac{1}{p_i}}^2} \cdot \Ex{\left(\widehat{(u_i+v_i)^p - u_i^p}\right)^2} \\
&+ \sum_{i \in S_2} (p_i \pm n^{-c}) \cdot \Ex{\widehat{\frac{1}{p_i}}^2} \cdot \Ex{\left(\widehat{(u_i+v_i)^p - u_i^p}\right)^2} \\
&:= S_1 + S_2.
\end{align*}
For the first term, we have
\begin{align*}S_1 &= \sum_{i \in S_1} p_i \cdot \frac{1}{p_i^2} \cdot ((u_i+v_i)^{p} - u_i^p)^2 \cdot \polylog(n) \\
&= \sum_{i \in S_1} \frac{1}{p_i} \cdot ((u_i+v_i)^{p} - u_i^p)^2\cdot \polylog(n).
\end{align*}
Then,
\begin{align*}
S_1 &= \sum_{i \in S_1} \frac{\|u\|_p^p\|v\|_p^p}{\|u\|_p^pv_i^p + \|v\|_p^p u_i^p} \cdot ((u_i+v_i)^{p} - u_i^p)^2 \cdot \polylog(n) \\
&\le \sum_{i \in S_1} \frac{\|u\|_p^p}{u_i^p} \cdot ((u_i+v_i)^{p} - u_i^p)^2 \cdot \polylog(n). 
\end{align*}
Moreover, since $u_i \ge v_i$ for $i \in S_1$, we have $(u_i+v_i)^p - u_i^p = \O{u_i^p}$. Then, we have
\begin{align*}
S_1 &\le \|u\|_p^p \cdot \sum_{i \in S_1}  ((u_i+v_i)^{p} - u_i^p) \cdot \polylog(n) \\
&\le \gamma \|u\|_p^{2p} \cdot \polylog(n),
\end{align*}
where the second step is by our assumption of $\|u+v\|_p^p - \|u\|_p^p \le \gamma \|u\|_p^p$. For the second term, similarly we have 
\begin{align*}
S_2 &= \sum_{i \in S_2} \frac{\|u\|_p^p\|v\|_p^p}{\|u\|_p^pv_i^p + \|v\|_p^p u_i^p} \cdot (u_i+v_i)^{2p} \cdot \polylog(n) \\
&\le \sum_{i \in S_2} \frac{\|v\|_p^p}{v_i^p} \cdot (u_i+v_i)^{2p} \cdot \polylog(n). 
\end{align*}
Now, since $u_i < v_i$ for $i \in S_2$, we have $(u_i+v_i)^p = \O{v_i^p}$. 
Thus, we have
\begin{align*}
S_2 = \sum_{i \in S_2} \frac{\|v\|_p^p}{v_i^p} \cdot v_i^{2p} \cdot \polylog(n)  \le \|v\|_p^{2p} \cdot \polylog(n) \le \gamma \cdot \|u\|_p^{2p},
\end{align*}
where the last step is by our assumption of $\|v\|_p^p \le \gamma \|u\|_p^p$. Therefore, $\Ex{s_r^2} = \gamma \|u\|_p^{2p} \cdot \polylog(n)$.
Recall that the expectation of our estimator satisfies $\Ex{s_r} = S \pm \frac{\eps}{2} \cdot \|u\|_p^p$. 
By Chebyshev's inequality, taking the mean of our $\frac{\gamma}{\eps^2} \cdot \polylog(n)$ instances of estimators $s_r$ gives us an estimation to $S$ with additive error $\eps \cdot \|u\|_p^p$. 
Last, union bounding across all instances of estimators will give us the desired failure probability.
\end{proof}
Now, we give the full guarantees of our difference estimator. 
\begin{theorem}
\thmlab{thm:fp:de}
Let $k$ denote the total number of servers, $p\ge2$, and $\eps \in (0,1)$ denote some accuracy parameter. 
Given a vector $f$ with universe size $n$, there exists a distributed $(\eps, \eps)$-continuous $F_p$ difference estimator which uses $\frac{1}{\eps} \cdot k^{p-1} \polylog(n)$ bits of communication in expectation. Also, there is a distributed $(\gamma, \eps)$-static $F_p$ difference estimator which uses $\frac{\gamma}{\eps^2} \cdot k^{p-1} \cdot \polylog(n)$ bits of communication in expectation. 
\end{theorem}
\begin{proof}
Given the correctness result shown in \lemref{lem:fp:de:cor}, it remains to bound the communication of \algref{alg:norm:diff}. For the continuous difference estimator, since we run $\frac{\gamma}{\eps^2} \cdot \polylog(n) = \frac{1}{\eps} \cdot \polylog(n)$ instances of estimators and each of them has communication cost $k^{p-1} \cdot \polylog(n)$ in expectation, the total communication cost is $\frac{1}{\eps} \cdot k^{p-1} \polylog(n)$. For the static version,  we run $\frac{\gamma}{\eps^2} \cdot \polylog(n)$ instances of estimators and each of them has communication cost $k^{p-1} \cdot \polylog(n)$, the total communication cost is $\frac{\gamma}{\eps^2} \cdot k^{p-1} \polylog(n)$.
\end{proof}

\subsubsection{Unbiased \texorpdfstring{$F_p$}{Fp} Estimator}
\seclab{sec:unbiased:fp:oneshot}
In this section, we provide the unbiased $F_p$ estimation scheme, which is an important subroutine in our difference estimator. We use the unbiased estimator given by the geometric mean of three independent instances $X = \max_{i \in [n]} \frac{f_i}{\be_i^{1/p}}$:
\[Y = \sqrt[3]{X_1^{p} \cdot X_2^{p} \cdot X_3^{p}},\]
which provides an unbiased estimate to $\|f\|_p^p$.

\paragraph{Static algorithm.}
We first consider the static version of $F_p$-estimation. Notice that our $L_p$ sampler reports the maximum index $i^*$ of the scaled vector, then we can query all servers to attain the real value of the $\frac{f_{i^*}}{\be_{i^*}^{1/p}}$. We present the algorithm as follows.

\begin{algorithm}[!htb]
\caption{Static unbiased $F_p$ estimation}
\alglab{alg:lp:unbiased:static}
\begin{algorithmic}[1]
\Require{Vector $f = \sum_{j \in [k]} f(j) \in \mathbb{R}^n$}
\Ensure{Unbiased estimation of $\|f\|_p^p$}
\For{$r \in [3]$}
\State{Scaled $f$ by i.i.d. exponential variables to get vector $g_r$ where $g_{i,r} = \frac{f_i}{e_{i,r}^{1/p}}$}
\State{$i \gets \argmax_{j \in [n]} g_{j,r}$} \Comment{Run the static version of \textsc{LpSampler} (See \thmref{thm:lp:sampler})}
\State{$H_r \gets g_{i,r}^p$}
\EndFor
\State{\Return $\frac{1}{ C_{\ref{lem:geo:exp:var}} } \cdot \sqrt[3]{H_1 \cdot H_2 \cdot H_3}$}
\end{algorithmic}
\end{algorithm}

\begin{theorem}
\thmlab{thm:fp:est:stat}
With probability $1-\frac{1}{\poly(n)}$, \algref{alg:lp:unbiased:static} outputs an estimation $\widehat{\|f\|_p^p}$ to $\|f\|_p^p$ such that $\Ex{\widehat{\|f\|_p^p}} = \|f\|_p^p$ and $\Ex{\left(\widehat{\|f\|_p^p}\right)^2} = \O{\|f\|_p^{2p}}$. It uses $k^{p-1} \polylog(n)$ bits of communication.
\end{theorem}
\begin{proof}
By \corref{cor:exp_inverse}, we have
\[\max_{j \in [n]} g_j^p = \frac{\|f\|_p^p}{E},\]
where $E$ is an independent exponential variable with rate $1$. Then, we have
\[\sqrt[3]{H_1 \cdot H_2 \cdot H_3} = \frac{\|f\|_p^p}{\sqrt[3]{E_1\cdot E_2\cdot E_3}},\]
where $E_1,E_2,E_3$ are independent. Now, by \lemref{lem:geo:exp:var}, we have
\begin{align*}\Ex{\frac{1}{ C_{\ref{lem:geo:exp:var}}} \cdot \sqrt[3]{H_1 \cdot H_2 \cdot H_3}} &= \frac{\|f\|_p^p}{ C_{\ref{lem:geo:exp:var}}} \cdot \Ex{ \frac{1}{\sqrt[3]{E_1\cdot E_2\cdot E_3}}} = \|f\|_p^p.
\end{align*}
Thus, our estimator is unbiased. Moreover, using \lemref{lem:geo:exp:var} again yields
\[\Ex{\left(\frac{1}{ C_{\ref{lem:geo:exp:var}}} \cdot \sqrt[3]{H_1 \cdot H_2 \cdot H_3}\right)^2} = \frac{\|f\|_p^{2p}}{ C_{\ref{lem:geo:exp:var}}^2} \cdot \Ex{ \left(\frac{1}{\sqrt[3]{E_1\cdot E_2\cdot E_3}}\right)^2} = \O{\|f\|_p^{2p}}.\]
Therefore, we have bounded the second moment. In addition, its communication bound directly follows from the communication bound of our $L_p$ sampler \thmref{thm:lp:sampler}.
\end{proof}

\paragraph{Continuous algorithm.}
Now, we consider an evolving vector $f$ arriving in a data stream. Applying the $L_p$ sampler in \algref{alg:lp:sampler:cent} identifies the maximum index $i$ of the scaled vector $g$ in the continuous monitoring setting, which constructs the geometric mean estimator.
Recall that we conduct a statistical test to ensure that we select the actual maximum, given the estimation error. However, the correctness of our geometric mean estimator relies on the max-stability of exponential variables: $\frac{f^p_{i}}{\be_{i}} \sim \frac{\|f\|_p^p}{\be}$. It is unclear whether this property holds when some instances are rejected due to the statistical test. To address this issue, we run the $\textsc{Counting}$ algorithm with $\O{1}$ accuracy, ensuring that we incur at most an $\O{1}$-fractional bias. This modification results in an additional communications cost of $\O{\sqrt{k}}$ bits, which is acceptable because, ultimately, we only use the continuous difference estimator at granularity $0$, where constant accuracy suffices, and there is roughly $\frac{1}{\eps}$ such continuous different estimators.

Similar to the static algorithm, we calculate the geometric mean estimator. Since we do not have access to the actual value of $g_i$ in the continuous monitoring setting, we make a slight modification to the static algorithm to ensure an unbiased estimate. Instead of taking the geometric mean of three estimators $G_i^p$ under different scaling, we retrieve three estimators $g_i^{p/3}$ and then multiply them together. Hence, our goal is to give an unbiased estimation of $G_{i}^{p/3}$ with bounded variance. This is done by applying the truncated Taylor estimator mentioned introduced in \secref{sec:truncated:taylor}:
\[\widehat{g_i^{(q)}} = \sum_{q=0}^Q \binom{p/3}{q}s_i^{p/3-q}(g_i-s_i)^q,\]
where $s_i$ is a $(1+\frac{1}{100p})$-approximation to $g_i$. 
Truncating at $\O{\log n}$ terms, the estimation only has a $\frac{1}{\poly(n)}$ additive bias.
We describe the algorithms in full in \algref{alg:lp:est:cent} and \algref{alg:lp:unbiased:cts:frac}.

\begin{algorithm}[!htb]
\caption{Distributed monitoring $F_p$ estimation: Truncated Taylor estimator}
\alglab{alg:lp:est:cent}
\begin{algorithmic}[1]
\Require{Input vector $f = \sum_{j \in [k]} f(j) \in \mathbb{R}^n$ from a distributed data stream, index $i$}
\Ensure{Unbiased estimate of $g_i^{p/3}$}
\State{The central coordinator sends the random seeds to each server to generate $n$ i.i.d. exponential random variables with rate $1$}
\State{For each update, obtain a set $\SetPL$ containing the maximum index} \Comment{See \thmref{thm:pl_recover}}
\For{$i \in \SetPL$}
\State{$\widehat{g_i} \gets \textsc{Counting}(g_i,\frac{1}{100p})$} \Comment{Give a $(1+\frac{1}{100p})$-approximation to $g_i$}
\EndFor
\State{$i \gets \argmax_{\bi} \widehat{g_{\bi}}$}
\State{$\widehat{g_i^{(q)}} \gets \textsc{Counting}(g_i,\frac{1}{100p})$ for all $q \in [\O{\log n}]$}
\State{$\widehat{g_i^{p/3}} \gets \textsc{TaylorEst}(\{\widehat{g_i^{(q)}}\}_{q\in[\O{\log n}]})$} \Comment{See \lemref{lem:taylor:fi}}
\State{\Return $\widehat{g_i^{p/3}}$}
\end{algorithmic}
\end{algorithm}

\begin{algorithm}[!htb]
\caption{Distributed monitoring $F_p$ estimation: geometric mean estimator}
\alglab{alg:lp:unbiased:cts:frac}
\begin{algorithmic}[1]
\Require{Input vector $f = \sum_{j \in [k]} f(j) \in \mathbb{R}^n$ from a distributed data stream}
\Ensure{Nearly unbiased estimate for $\|f\|_p^p$}
\For{$r \in [3]$}
\State{Run the distributed $F_p$ estimation algorithm and collect the estimate $H_r$} \Comment{\algref{alg:lp:est:cent}}
\EndFor
\State{\Return $\frac{1}{ C_{\ref{lem:geo:exp:var}}} \cdot H_1 \cdot H_2 \cdot H_3$}
\end{algorithmic}
\end{algorithm}

We now show the correctness of our algorithm.
\begin{lemma}
\lemlab{lem:lp:unbiased:cts:frac}
With probability $1-\frac{1}{\poly(n)}$, \algref{alg:lp:unbiased:cts:frac} outputs an estimation $\widehat{\|f\|_p^p}$ to $\|f\|_p^p$ such that $\Ex{\widehat{\|f\|_p^p}} = \|f\|_p^p \cdot (1 \pm \O{\eps})$ and $\Ex{\left(\widehat{\|f\|_p^p}\right)^2} = \O{\|f\|_p^{2p} \cdot \log^6 n}$.
\end{lemma}
\begin{proof}
Consider an instance $r$ and a fixed time $t$. Let $i$ be the reported max index, by \lemref{lem:taylor:fi}, our estimator $H_r$ satisfies $\Ex{H_r} =  g_i^{p/3} \pm n^{-c}$ and $\Ex{H_r^2} = \O{g_i^{p/3}\cdot \log^2 n}$, where the expectation is taken w.r.t. the \textsc{Counting} procedure. Now, since we have a $(1+\eps)$-approximation by \textsc{Counting} for each item in $\SetPL$, $g_i$ is within $(1\pm\O{\eps})$ fraction from the real max, denoted by $g_{D(1)}$. Therefore, we have $\Ex{H_r} =  g_{D(1)}^{p/3} \cdot (1 \pm \O{\eps})$ and $\Ex{H_r^2} = \O{g_{D(1)}^{p/3}\cdot \log^2 n}$.
By \corref{cor:exp_inverse}, we have
\[g_i^{p/3} = \max_{j \in [n^{c+1}]} g_j^{p/3} = \frac{\|f\|_p^{p/3}}{E^{1/3}},\]
where $E$ is an independent exponential variable with rate $1$. 
Therefore, let $g_{D(1),r}$ denote the real max in round $r$, we have
\[\Ex{\frac{1}{C_{\ref{lem:geo:exp:var}}\cdot n^c} \cdot H_1 \cdot H_2 \cdot H_3} = \frac{1}{ C_{\ref{lem:geo:exp:var}}\cdot n^c} \cdot \Ex{ g_{D(1),1}^{p/3} \cdot g_{D(1),2}^{p/3} \cdot g_{D(1),3}^{p/3}}\cdot( 1 \pm \O{\eps})\]
Now, by \lemref{lem:geo:exp:var}, we have
\begin{align*}\Ex{\frac{1}{ C_{\ref{lem:geo:exp:var}}} \cdot H_1 \cdot H_2 \cdot H_3} = \frac{\|f\|_p^p}{ C_{\ref{lem:geo:exp:var}}} \cdot \Ex{ \frac{1}{\sqrt[3]{E_1\cdot E_2\cdot E_3}}}\cdot (1 \pm \O{\eps}) =  \|f\|_p^p (1 \pm \O{\eps}).
\end{align*}
 Moreover, using \lemref{lem:geo:exp:var} again yields
\[\Ex{\left(\frac{1}{ C_{\ref{lem:geo:exp:var}}} \cdot H_1 \cdot H_2 \cdot H_3\right)^2} = \O{\log^6 n} \cdot \frac{\|f\|_p^{2p}}{ C_{\ref{lem:geo:exp:var}}^2} \cdot \Ex{ \left(\frac{1}{\sqrt[3]{E_1\cdot E_2\cdot E_3}}\right)^2} = \O{\|f\|_p^{2p} \cdot \log^6 n}.\]
Therefore, we have bounded the second moment.
\end{proof}
We now upper bound the total communication cost of our protocol. 
\begin{lemma}
\lemlab{lem:lp:unbiased:cts:comm}
With probability $1-\frac{1}{\poly(n)}$, \algref{alg:lp:unbiased:cts:frac} uses $k^{p-1}\cdot\polylog(n) + \frac{\sqrt{k}}{\eps} \cdot \polylog(n)$ bits of communication.
\end{lemma}
\begin{proof}
First, we need $k^{p-1}\cdot\polylog(n)$ bits of to run \textsc{LpSampler}. 
Second, we run \textsc{Counting} with accuracy $\eps$ to track each item in $\SetPL$, which uses $\O{\frac{\sqrt{k}}{\eps} \log n}$ bits of communication each. 
Recall that in \lemref{lem:num:set:pl}, we showed that only $\polylog(n)$ items are added to $\SetPL$ whenever $\|f\|_p^p$ doubles.
Therefore, the total communication cost is at most $k^{p-1}\cdot\polylog(n) + \frac{\sqrt{k}}{\eps} \cdot \polylog(n)$ bits. 
\end{proof}
Finally, we give the full guarantees of our continuous $F_p$ estimation protocol, given that \lemref{lem:lp:unbiased:cts:frac} and \lemref{lem:lp:unbiased:cts:comm} hold. 
\begin{theorem}
\thmlab{thm:lp:est}
Given a vector $f \in \mathbb{R}^{n}$ arriving in the distributed streaming model, let $k$ denote the total number of servers. Let $p\ge 2$. 
Then, there is a algorithm \textsc{LpEstUnbiased} that gives an estimation $\widehat{F}$ to $\|f\|_p^p$, which satisfies $\Ex{\widehat{F}} = (1\pm \O{\eps}) \cdot \|f\|_p^p$ and $\Ex{\left(\widehat{F}\right)^2} = \|f\|_p^{2p} \cdot \polylog(n)$. 
The algorithm uses $k^{p-1}\cdot\polylog(n) + \frac{\sqrt{k}}{\eps} \cdot \polylog(n)$ bits of communication. The algorithm succeeds with probability $1-\frac{1}{\poly(n)}$.
\end{theorem}

\subsubsection{Estimation of the Inverse Sampling Probability}

This section provides an estimator for the inverse of the mixed sampling probability $\frac{1}{p_i}$, which is
\[\frac{1}{p_i}=\frac{2}{u_i^p/\|u\|_p^p + v_i^p / \|v\|_p^p} = \frac{2\|u\|_p^p\|v\|_p^p}{\|u\|_p^pv_i^p+\|v\|_p^pu_i^p}.\]
So, we need to construct an algorithm to estimate $\frac{1}{\|u\|_p^pv_i^p+\|v\|_p^pu_i^p}$. 
For convenience, we denote $a = \|u\|_p^pv_i^p$ and $b = \|v\|_p^pu_i^p$. 
Our approach is to express the function $f(a,b) = \frac{1}{a+b}$ via a bivariate Taylor expansion. 
The expansion has bounded tail because the partial derivatives of $f(a,b)$ with respect to both $a$ and $b$ have the same form. 
Then, we utilize the subroutines introduced in previous sections to give a good estimation to each term in the Taylor expansion, such as $\|u\|_p^p$, $\|v\|_p^p$, and $v_i^p$. 
We remark that, since each server stores the vector $u$ upon its arrival, we can obtain the actual value of $u_i$ by querying each server for each sampled index $i$.

Again, for static different estimator, it suffices to use the static subroutines in \secref{sec:unbiased:fp:oneshot} to give (nearly)-unbiased estimates to $\|u\|_p^p$ and $\|v\|_p^p$.
We describe our algorithm in full in \algref{alg:lp:sample:prob}.

\begin{algorithm}[!htb]
\caption{Continuous estimation of the inverse of the sampling probability}
\alglab{alg:lp:sample:prob}
\begin{algorithmic}[1]
\Require{Underlying vector $u$ defined at time $t$, sampled index $i$, $u_i$}
\Ensure{Estimates of $1/p_i$, accuracy $\eps$}
\State{$B \gets \polylog(n)$, $Q \gets \O{\log n}, D \gets \O{\log n}$ } \Comment{$B$ is the number of replications of each estimator, $Q$ is the number of truncated terms of $1/p_i$, $D$ is the truncated terms of $v_i^p$}
\For{$q \in [Q]$}
\For{$b \in [B]$}
\State{$x_{q,b} \gets \textsc{LpEstUnbiasedStat}(u)$} \Comment{Static $F_p$-estimation}
\State{$y_{q,b} \gets \textsc{LpEstUnbiased}(v, \frac{\eps}{C_{\ref{lem:inverse:prob:corr}} \log n})$}
\EndFor
\State{$x_q \gets \frac{1}{B} \sum_{b \in [B]} x_{q,b}$} \Comment{Estimator for $\|u\|_p^p$}
\State{$y_q \gets \frac{1}{B} \sum_{b \in [B]} y_{q,b}$}  \Comment{Estimator for $\|v\|_p^p$}
\For{$b \in [B]$}
\State{Let $s^{(b)}_{i}$ be a $(1+\frac{1}{100p})$-approximation of $v_i$} 
\For{$l \in [D]$}
\State{$\widehat{v_i^{(q,b,l)}} \gets \textsc{Counting}(v_i)$}
\EndFor
\For{$d \in [D]$}
\State{$w_{q,b,d} \gets \binom{p}{d}s_i^{{(b)}^{p-d}} \prod_{l \in [d]} \left(\widehat{v_i^{(q,b,l)}} - s_i^{(b)}\right)$} 
\EndFor
\State{$w_{q,b} \gets \sum_{d \in [D]} w_{q,b,d}$} 
\EndFor
\State{$w_q \gets \frac{1}{B} \sum_{b \in [B]} w_{q,b}$}  \Comment{Estimator for $v_i^p$}
\State{$a_q \gets x_q \cdot w_q$}  \Comment{Estimator for $\|u\|_p^p \cdot v_i^p$}
\State{$b_q \gets y_q \cdot u_i^p$}  \Comment{Estimator for $\|v\|_p^p \cdot u_i^p$}
\EndFor
\State{Let $r$ be a $(1+\frac{1}{100p})$-approximation of $\|u\|_p^p$} \Comment{Obtain by $F_p$-estimation of $u$}
\State{Let $o$ be a $(1+\frac{1}{100p})$-approximation for $\|v\|_p^p$} \Comment{Obtain by $F_p$-estimation of $v$}
\State{Let $h_i$ be a $(1+\frac{1}{100p})$-approximation of $v_i$}
\State{$a_0 \gets r \cdot h_i$}
\State{$b_0 \gets o \cdot u_i^p$}
\State{$z \gets \sum_{q=0}^Q \frac{(-1)^q}{q!} \cdot \frac{1}{(a_0+b_0)^{q+1} }\cdot \prod_{l \in [q]} \left((a_l-a_0)+(b_l-b_0)\right)$}
\State{$x \gets \textsc{LpEstUnbiased}(u)$}
\State{$y \gets \textsc{LpEstUnbiased}(v)$}
\State{\Return $2xyz$}
\end{algorithmic}
\end{algorithm}

We show the tail bound of the bivariate Taylor expasion of $f(a,b)$.
\begin{lemma}
\lemlab{lem:bivar:taylor}
Let $Q = \O{\log n}$, $a_0 \in [0.99a,1.01a]$ and $b_0 \in [0.99b,1.01b]$. Suppose we have $a \ge 1$ and $b\ge 1$. Then,
\[\left|\frac{1}{a+b} - \sum_{q=0}^Q \frac{(-1)^q}{q!} \cdot \frac{1}{(a_0+b_0)^{q+1} }\cdot ((a-a_0)+(b-b_0))^q\right| < n^{-c}.\]
\end{lemma}
\begin{proof}
First, we have the bivariate Taylor series of $\frac{1}{a+b}$ to be
\[\frac{1}{a+b} = \sum_{q=0}^\infty \frac{(-1)^q}{q!} \cdot \frac{1}{(a_0+b_0)^{q+1}} \cdot ((a-a_0)+(b-b_0))^q.\]
Now, by our assumption, we have $a \ge 1$ and $b \ge 1$. Since $a_0 \in [0.99a,1.01a]$ and $b_0 \in [0.99b,1.01b]$, we have
\[\left|\frac{(-1)^q}{q!} \cdot \frac{1}{(a_0+b_0)^{q+1}} \cdot ((a-a_0)+(b-b_0))^q\right| \le \frac{1}{a+b} \cdot \left(\frac{1}{40}\right)^q.\]
Hence, we have that for $Q = \O{\log n}$,
\[\sum_{q=Q}^\infty \frac{(-1)^q}{q!} \cdot \frac{1}{(a_0+b_0)^{q+1}} \cdot ((a-a_0)+(b-b_0))^q \le \frac{1}{a+b} \cdot \sum_{q=Q}^\infty \left(\frac{1}{40}\right)^q \le n^{-c}.\]
\end{proof}

Next, we state the guarantee of our truncated bivariate Taylor estimator for $f(a,b)$, given that we have good approximations to both $a$ and $b$.
\begin{lemma}
\lemlab{lem:taylor:inverse}
Let $Q = \O{\log n}$, $a_0 \in [0.99a,1.01a]$, and $b_0 \in [0.99b,1.01b]$. 
Suppose we have $Q$ independent estimates for $a$ and $b$ satisfying $\Ex{\widehat{a^{(l)}}} = a\cdot(1 \pm \eps)$, $\Ex{\widehat{b^{(l)}}} = b\cdot(1 \pm \eps)$ and $\Var{\widehat{a^{(l)}}}\le a^2/5$, $\Var{\widehat{b^{(l)}}}\le b^2/5$. 
We define the Taylor estimator by
\[z = \sum_{q=0}^Q \frac{(-1)^q}{q!} \cdot \frac{1}{(a_0+b_0)^{q+1} }\cdot \prod_{l\in[q]} \left((\widehat{a^{(l)}}-a_0)+(\widehat{b^{(l)}}-b_0)\right).\]
Then, we have $\Ex{z} = \frac{1}{a+b}\cdot(1 \pm \O{\eps \log n})$ and $\Ex{z} = \left(\frac{1}{a+b}\right)^2 \cdot \log^2 n$.
\end{lemma}
\begin{proof}
First, since we have independent estimators such that $\Ex{\widehat{a^{(l)}}} = a\cdot(1 \pm \eps)$, $\Ex{\widehat{b^{(l)}}} = b\cdot(1 \pm \eps)$, so it holds that
\begin{align*}\Ex{\prod_{l\in[q]} \left((\widehat{a^{(l)}}-a_0)+(\widehat{b^{(l)}}-b_0)\right)} &= \prod_{l\in[q]} \Ex{(\widehat{a^{(l)}}-a_0)+(\widehat{b^{(l)}}-b_0)} \\
&= ((a-a_0)+(b-b_0)^q \cdot(1 \pm \eps)^q \\
&=((a-a_0)+(b-b_0))^q \cdot(1 \pm \O{q\eps}) ,
\end{align*}
where $q \le \O{\log n}$. Then, summing up the bias in each term of the truncated Taylor series $z$, we have $\Ex{z} = \frac{1}{a+b}\cdot(1 \pm \O{\eps \log n})$ by \lemref{lem:bivar:taylor}. Moreover, by \lemref{lem:aux:lp:est} we have
\begin{align*}
\Ex{\left((\widehat{a^{(l)}}-a_0)+(\widehat{b^{(l)}}-b_0)\right)^2} &= \Var{\widehat{a^{(l)}} + \widehat{b^{(l)}}} + \Ex{(\widehat{a^{(l)}}-a_0)+(\widehat{b^{(l)}}-b_0)}^2 \\
&= \Var{\widehat{a^{(l)}} + \widehat{b^{(l)}}} + ((a-a_0)+(b-b_0))^2.
\end{align*}
Now, since our estimates for $a$ and $b$ are independent, we have
\begin{align*}
\Ex{\left((\widehat{a^{(l)}}-a_0)+(\widehat{b^{(l)}}-b_0)\right)^2} &= \Var{\widehat{a^{(l)}}} + \Var{\widehat{b^{(l)}}} + ((a-a_0)+(b-b_0))^2 \\
&= \frac{a^2}{5}+\frac{b^2}{5} + \left((a-a_0)+(b-b_0)\right)^2.
\end{align*}
Now, since $a_0 \in [0.99a,1.01a]$ and $b_0 \in [0.99b,1.01b]$, we have
\[\Ex{\left((\widehat{a^{(l)}}-a_0)+(\widehat{b^{(l)}}-b_0)\right)^2} \le \left(\frac{a+b}{2}\right)^2.\]
Now, we have
\begin{align*}\Ex{\left(\prod_{l\in[q]} \left((\widehat{a^{(l)}}-a_0)+(\widehat{b^{(l)}}-b_0)\right)\right)^2} = \prod_{l\in[q]} \Ex{\left((\widehat{a^{(l)}}-a_0)+(\widehat{b^{(l)}}-b_0)\right)^2} 
\le \left(\frac{a+b}{2}\right)^{2q}.
\end{align*}
Notice that $2(a_0+b_0) > a+b$. 
Therefore, we have
\begin{align*}
&\Ex{ \left(\frac{(-1)^q}{q!} \cdot \frac{1}{(a_0+b_0)^{q+1} }\cdot \prod_{l\in[q]} \left((\widehat{a^{(l)}}-a_0)+(\widehat{b^{(l)}}-b_0)\right)\right)^2} \\
\le & \frac{1}{(q!)^2} \cdot \left(\frac{1}{a_0+b_0}\right)^{2q+2} \cdot  \left(\frac{a+b}{2}\right)^{2q} \le \left(\frac{1}{a+b}\right)^2.
\end{align*}
Last, by the AM-GM inequalities, we have
\begin{align*}\Ex{o^2} &= \Ex{\left(\sum_{q=0}^Q \frac{(-1)^q}{q!} \cdot \frac{1}{(a_0+b_0)^{q+1} }\cdot \prod_{l\in[q]} \left((\widehat{a^{(l)}}-a_0)+(\widehat{b^{(l)}}-b_0)\right)\right)^2} \\
&\le \O{Q} \cdot \sum_{q=0}^Q \Ex{\left(\frac{(-1)^q}{q!} \cdot \frac{1}{(a_0+b_0)^{q+1} }\cdot \prod_{l\in[q]} \left((\widehat{a^{(l)}}-a_0)+(\widehat{b^{(l)}}-b_0)\right)\right)^2} \\
&\le \O{Q^2} \left(\frac{1}{a+b}\right)^2.
\end{align*}
which proves our statement as $Q = \O{\log n}$.
\end{proof}

The following lemma proves that we can reduce the variance by an arithmetic mean estimator. 
This will be used to upper bound the variance of our truncated bivariate Taylor estimator.
\begin{lemma}
\lemlab{lem:repetition}
Let $\widehat{a}$ be an estimation of $a$ with $\Ex{\widehat{a}} = a \cdot (1 \pm \eps)$ and $\Var{\widehat{a}} = C$. 
Then if we take the mean of $M$ independent estimations, the resulting quantity has $(1+\eps)$-multiplicative error in expectation and its variance will be $C/M$.
\end{lemma}
\begin{proof}
First, we have the expectation to be
\[\Ex{\frac{1}{M} \sum_{l=1}^M \widehat{a^{(l)}}} = \frac{1}{M} \cdot M\cdot a \cdot(1 \pm \eps) = a \cdot (1 \pm \eps).
\]
Moreover, since they are independent, we have
\[\Var{\frac{1}{M}\sum_{l=1}^M \widehat{a^{(l)}}} = \frac{1}{M^2}\sum_{l=1}^M\Var{ \widehat{a^{(l)}}}= \frac{C}{M}. \]
\end{proof}

Applying the above statement with $M = \polylog(n)$, we prove the correctness of our algorithm:
\begin{lemma}
\lemlab{lem:inverse:prob:corr}
\label{lem:inverse:prob:corr}
With probability $1-\frac{1}{\poly(n)}$, \algref{alg:lp:sample:prob} gives an estimation $\widehat{\frac{1}{p_i}}$ to $\frac{1}{p_i}$ satisfying $\Ex{\widehat{\frac{1}{p_i}}} = \frac{1}{p_i} \cdot\left(1 \pm \O{\frac{\eps}{C_{\ref{lem:inverse:prob:corr}}}}\right)$ and $\Ex{\widehat{\frac{1}{p_i}}^2} = \frac{1}{p_i^2} \cdot \polylog(n)$.
\end{lemma}
\begin{proof}
First, by \lemref{lem:taylor:fi}, \thmref{thm:fp:est:stat} and \lemref{lem:lp:unbiased:cts:frac}, we have the following guarantees of our estimators. For all $q \in [Q]$ and $b \in [B]$, we have
\begin{align*}
    \Ex{x_{q,b}} &= \|u\|_p^p~~\mathrm{and} ~~ \Ex{x_{q,b}^2} =  \|u\|_p^{2p}\\
    \Ex{y_{q,b}} &= \|v\|_p^p\left(1 \pm \O{\frac{\eps}{C_{\ref{lem:inverse:prob:corr}} \log n}}\right) ~~\mathrm{and} ~~ \Ex{y_{q,b}^2} =  \|v\|_p^{2p}\cdot \polylog(n)\\
    \Ex{w_{q,b}} &= v_i^p\cdot\left(1 \pm \O{\frac{\eps}{C_{\ref{lem:inverse:prob:corr}} \log n}}\right) ~~\mathrm{and} ~~ \Ex{w_{q,b}^2} =  v_i^{2p}\cdot \polylog(n).
\end{align*}
Then, by \lemref{lem:repetition}, averaging $B = \O{\polylog(n)}$ such estimators give us:
\begin{align*}
    \Ex{x_{q}} &= \|u\|_p^p ~~\mathrm{and} ~~ \Var{x_{q}} =  \|u\|_p^{2p}/5\\
    \Ex{y_{q}} &= \|v\|_p^p\cdot\left(1 \pm \O{\frac{\eps}{C_{\ref{lem:inverse:prob:corr}} \log n}}\right)~~\mathrm{and} ~~ \Var{y_{q}} =  \|v\|_p^{2p}/5\\
    \Ex{w_{q}} &= v_i^p\cdot\left(1 \pm \O{\frac{\eps}{C_{\ref{lem:inverse:prob:corr}} \log n}}\right) ~~\mathrm{and} ~~ \Var{w_{q}} =  v_i^{2p}/5.
\end{align*}
Therefore, defining $a = \|u\|_p^pv_i^p$ and $b = \|v\|_p^pu_i^p$, for all $q \in [Q]$, we have
\begin{align*}
    \Ex{a_{q}} &= a\cdot\left(1 \pm \O{\frac{\eps}{C_{\ref{lem:inverse:prob:corr}} \log n}}\right) ~~\mathrm{and} ~~ \Var{x_{q}} \le  a^2/5\\
    \Ex{b_{q}} &= b\cdot\left(1 \pm \O{\frac{\eps}{C_{\ref{lem:inverse:prob:corr}} \log n}}\right) ~~\mathrm{and} ~~ \Var{b_{q}} \le  b^2/5.
\end{align*}
Moreover, we have $a_0 = a\cdot(1\pm\frac{1}{100p})$ and $b_0 = b\cdot(1\pm\frac{1}{100p})$. Note that our estimator for $1/p_i$ is
\[z = \sum_{q=0}^Q \frac{(-1)^q}{q!} \cdot \frac{1}{(a_0+b_0)^{q+1} }\cdot \prod_{l \in [q]} \left((a_l-a_0)+(b_l-b_0)\right).\]
Then, by \lemref{lem:taylor:inverse}, we have 
\begin{align*}
    \Ex{z} &= \frac{1}{a+b} \cdot\left(1 \pm \O{\frac{\eps}{C_{\ref{lem:inverse:prob:corr}}}}\right)~~\mathrm{and} ~~ \Ex{z^2} \le  \left(\frac{1}{a+b}\right)^2 \cdot \polylog(n).
\end{align*}
Note that $\frac{1}{a+b} = \frac{1}{p_i}$, we have our desired result.
Moreover, the failure probability follows by union bounding the failure probability of each estimator.
\end{proof}

The following lemma upper bounds the communication cost.
\begin{lemma}
\lemlab{lem:inverse:prob:comm}
With probability $1-\frac{1}{\poly(n)}$, \algref{alg:lp:sample:prob} uses $\left(k^{p-1} + \frac{\sqrt{k}}{\eps}\right)\cdot \polylog(n)$ bits of communication. 
In addition, if $v$ is static, \algref{alg:lp:sample:prob} uses $k^{p-1} \polylog(n)$ bits of communication and gets unbiased estimation (up to $\frac{1}{\poly(n)}$ factors).
\end{lemma}
\begin{proof}
We run $\polylog(n)$ instances of \textsc{LpEstUnbiased} in \algref{alg:lp:sample:prob}, whose communication cost follows from \thmref{thm:lp:est}. Moreover, it takes $\polylog(n)$ instances of \textsc{Counting} to retrieve the $\polylog(n)$ estimations for $v_i$ we needed in \algref{alg:lp:sample:prob}. Hence, we uses $k^{p-1} \cdot \polylog(n)$ bits of communication in total.

If $v$ is static, we can implement the static version of $L_p$ estimation algorithm to estimate $\|v\|_p^p$. 
Then, it only uses $k^{p-1} \polylog(n)$ bits of communication due to \thmref{thm:fp:est:stat}. Since the bias only comes from the error in the continuous $\textsc{Counting}$ algorithm, we have an unbiased estimation in the static version.
\end{proof}

Notice that we only use the continuous difference estimator at granularity $0$, where a constant-factor approximation suffices. Therefore, we state the following corollary showing a constant-fractional bias result for the continuous monitoring setting, and an unbiased result for the static setting, given \lemref{lem:inverse:prob:corr} and \lemref{lem:inverse:prob:comm}.
\begin{corollary}
\corlab{cor:inverse:prob}
\algref{alg:lp:sample:prob} gives an estimation $\widehat{\frac{1}{p_i}}$ to $\frac{1}{p_i}$ such that $\Ex{\widehat{\frac{1}{p_i}}} = \frac{1}{p_i} \left(1\pm\O{\frac{1.01}{C_{\ref{lem:inverse:prob:corr}}}}\right)$ and $\Var{\frac{1}{p_i}} = \frac{1}{p_i^2} \cdot \polylog(n)$ at all times with probability $1-\frac{1}{\poly(n)}$. This protocol uses $k^{p-1} \cdot \polylog(n)$ bits of communication. For static $v$, it uses $k^{p-1}  \cdot \polylog(n)$ bits of communication, and the estimate satisfies $\Ex{\widehat{\frac{1}{p_i}}} = \frac{1}{p_i} \pm n^{-c}$ and $\Var{\frac{1}{p_i}} = \frac{1}{p_i^2} \cdot \polylog(n)$ .
\end{corollary}

\subsubsection{Estimation of \texorpdfstring{$(u_i+v_i)^p - u_i^p$}{ui+vip-uip}}

This section provides the algorithm for estimating $(u_i+v_i)^p - u_i^p$. Since we can obtain the actual value of $u_i$, we consider the difference as a function of $v_i$, i.e., we define $f(x) = (u_i+x)^p - u_i^p$. We expand it at $x = \widehat{v_i}$ by Taylor expansion. Similarly, it is sufficient to truncate at the $\O{\log n}$-th term in the Taylor series, as shown in \algref{alg:lp:diff:index}. 

\begin{algorithm}[!htb]
\caption{Estimator of $(u_i+v_i)^p - u_i^p$}
\alglab{alg:lp:diff:index}
\begin{algorithmic}[1]
\Require{Underlying vector $u$ defined at time $t$, subsequent vector $v$, $u_i$, index $i$}
\Ensure{Unbiased estimator of $(u_i+v_i)^p - u_i^p$}
\State{$s \gets 0$, $Q \gets \O{\log n}$}
\For{$q \in [Q]$}
\State{$\widehat{v_i^{(q)}} \gets \textsc{Counting}(v_i)$}
\EndFor
\State{Let $s_i$ be a $(1+\frac{1}{100p})$-approximation of $v_i$} 
\For{$q \in [Q]$}
\State{$\widehat{(v_i-s_i)^q} \gets \prod_{l \in [q]} \left(\widehat{v_j^{(l)}} - s_i\right)$} 
\State{$s \gets s +\binom{p}{q}(u_i+s_i)^{p-q}\widehat{(v_i-s_i)^q}$} 
\EndFor
\State{\Return $s + \left[(u_i+s_i)^p-u_i^p\right]$}
\end{algorithmic}
\end{algorithm}

The following statement upper bounds the tail error of the truncated Taylor series. 
\begin{lemma}
\lemlab{lem:diff_taylor_u_i_smaller}
Let $Q=\O{\log n}$ and $s_i \in\left[0.99v_i, 1.01v_i \right]$. 
Then 
\[\left\lvert\left[(u_i+v_i)^p-u_i^p\right]-\left[(u_i+s_i)^p-u_i^p\right]-\sum_{q=1}^Q \binom{p}{q}\left[(u_i+s_i)^{p-q}\right](v_i-s_i)^q\right\rvert\le n^{-c}.\]
\end{lemma}
\begin{proof}
Using the Taylor expansion at $\widehat{v_i}$ for the function $f(x)=(u_i+x)^p-u_i^p$ at $s_i$, we have
\[\left[(u_i+v_i)^p-u_i^p\right]=\left[(u_i+s_i)^p-u_i^p\right]+\sum_{q=1}^\infty \binom{p}{q}\left[(u_i+s_i)^{p-q}\right](v_i-s_i)^q.\]
For $q>p$, we have $|v_i-s_i|\le \frac{v_i}{100}$, then
\[\binom{p}{q}(u_i+s_i)^{p-q}(v_i-s_i)^q\le(2e)^{p/2}(u_i+s_i)^{p}\left(\frac{1}{100}\right)^q.\]
Therefore, we have that for $Q=\O{\log n}$, 
\[\sum_{q=Q+1}^\infty\binom{p}{q}\left[(u_i+s_i)^{p-q}\right](v_i-s_i)^q\le (u_i+s_i)^{p} \cdot \sum_{q=Q+1}^\infty(2e)^{p/2}\left(\frac{1}{100}\right)^q \le n^{-c}.\]
Hence, 
\[\left\lvert\left[(u_i+v_i)^p-u_i^p\right]-\left[(u_i+s_i)^p-u_i^p\right]-\sum_{q=1}^Q \binom{p}{q}\left[(u_i+s_i)^{p-q}\right](v_i-s_i)^q\right\rvert\le n^{-c}.\]
\end{proof}

Now, we give an auxiliary result that is used to upper bound the variance.
\begin{lemma}
\lemlab{lem:c:v}
Let $c>1$ be some constant. 
Then we have
\[(u_i+cv_i)^p - u_i^p \le c^p \left((u_i+v_i)^p - u_i^p \right).\]
\end{lemma}
\begin{proof}
We define an auxiliary function as follows
\[f(x) = c^p \left((u_i+x)^p - u_i^p\right) - ((u_i+cx)^p - u_i^p).\]
The derivative of the auxiliary function satisfies
\[f'(x) = c^p p(u_i+x)^{p-1} - cp(u_i+cx)^{p-1} = c p(cu_i+cx)^{p-1} - cp(u_i+cx)^{p-1} \ge 0,\]
for $x>0$, $c>1$, and $p\ge2$. 
Therefore, we have $f(v_i) \ge f(0) = 0$, which gives our result.
\end{proof}

The following lemma proves the correctness of our algorithm.
\begin{lemma}
\lemlab{lem:diff:index}
With probability $1-\frac{1}{\poly(n)}$, \algref{alg:lp:diff:index} outputs an estimation $s$ of $(u_i+v_i)^p - u_i^p$ satisfying $\Ex{s} = (u_i+v_i)^p - u_i^p \pm n^{-c}$ and $\Ex{s^2} = \O{(u_i+v_i)^p \cdot \log^2 n}$. Moreover, if $u_i > v_i$, we have $\Ex{s^2} = \O{((u_i+v_i)^p-u_i^p)^2 \cdot \log n}$
\end{lemma}
\begin{proof}
For $l \in [Q]$, by the guarantee of counting (see \thmref{thm:counting}), we have $\Ex{\widehat{v_i^{(l)}}} = v_i$ and $\Var{\widehat{v_i^{(l)}}} = \left(\frac{v_i}{100p}\right)^2$. Thus, applying \lemref{lem:aux:lp:est} with $a = v_i$, $b=s_i$, and $C = \frac{1}{100p}$, we have $\Ex{\widehat{v_i^{(l)}} - s_i} = v_i - s_i$ and $\Ex{\left(\widehat{v_i^{(l)} - s_i}\right)^2} = \O{\left(\frac{v_i}{50}\right)^2}$. Then, since we obtain $Q$ independent estimation of $v_i$, we have
\[\Ex{\prod_{l \in [q]} \left(\widehat{v_i^{(l)}}-s_i\right)} = \prod_{l \in [q]} \Ex{\widehat{v_i^{(l)}}-s_i} = (v_i - s_i)^q.\]
Moreover, we have
\begin{align*}\Ex{\left(\prod_{l \in [q]} \left(\widehat{v_i^{(l)}}-s_i\right)\right)^2} = \prod_{l \in [q]} \Ex{\left(\widehat{v_i^{(l)}}-s_i\right)^2} = \O{\left(\frac{v_i}{50p}\right)^{2q}}.\end{align*}
Hence, by \lemref{lem:diff_taylor_u_i_smaller} and $Q = \O{\log n}$, we have
\begin{align*}
\Ex{s} &= (u_i+s_i)^p - u_i^p + \sum_{q\in[Q] \backslash \{0\}} \binom{p}{q} (u_i+s_i)^{p-q} \cdot \Ex{\prod_{l \in [q]} \left(\widehat{v_i^{(l)}}-s_i\right)} \\
&= (u_i+s_i)^p - u_i^p + \sum_{q\in[Q] \backslash \{0\}} \binom{p}{q} (u_i+s_i)^{p-q}(v_i - s_i)^q = G_i^p \pm n^{-c}.\end{align*}
Next, applying the AM-GM inequalities, we have
\begin{align*}\Ex{s^2} &= \Ex{\left((u_i+s_i)^p - u_i^p + \sum_{q \in [Q] \backslash \{0\}}\left(\binom{p}{q} (u_i+s_i)^{p-q} \prod_{l \in [q]} \left(\widehat{v_i^{(l)}}-s_i\right)\right)\right)^2} \\ &= \O{Q} \cdot \left(\left((u_i+s_i)^p - u_i^p\right)^2 + \sum_{q \in [Q] \backslash \{0\}}\binom{p}{q}^2 (u_i+s_i)^{2p-2q} \Ex{\prod_{l \in [q]} \left(\widehat{v_i^{(l)}}-s_i\right)^2} \right).\end{align*}
Now, we have
\[\Ex{s^2} = \O{Q} \cdot \left(\left((u_i+s_i)^p - u_i^p\right)^2 + \sum_{q \in [Q] \backslash \{0\}}\binom{p}{q}^2 (u_i+s_i)^{2p-2q} \left(\frac{v_i}{50p} \right)^{2q}\right).\]
Then, based on the same bound as in the analysis of \lemref{lem:taylor:fi}, we have $\Ex{s^2} = \O{(u_i+v_i)^p \cdot \log n}$.
Now, we suppose $u_i > v_i$. 
Consider two adjacent terms in the finite sum, since $s_i = v_i \cdot (1\pm\frac{1}{100p})$, we have
\[\frac{\binom{p}{q}^2 (u_i+s_i)^{2p-2q} \left(\frac{v_i}{50p} \right)^{2q}}{\binom{p}{q+1}^2 (u_i+s_i)^{2p-2q-2} \left(\frac{v_i}{50p} \right)^{2q+2}} = \frac{(q+1)^2(u_i+s_i)^2(50p)^2}{(q-p+1)^2v_i^2} > 2500.\]
Then, the sum is upper bounded by a decreasing geometric sequence. So, we have
\[\sum_{q \in [Q] \backslash \{0\}}\binom{p}{q}^2 (u_i+s_i)^{2p-2q} \left(\frac{v_i}{50p} \right)^{2q} \le \binom{p}{1}^2 (u_i+s_i)^{2p-2} \left(\frac{v_i}{50p} \right)^{2} \cdot \frac{1}{1-1/2500}.\]
Then, we have
\[\Ex{s^2} = \O{Q} \cdot \left(\left((u_i+s_i)^p - u_i^p\right)^2 + \binom{p}{1}^2 (u_i+s_i)^{2p-2} \left(\frac{v_i}{50p} \right)^{2}\right)\]
Now, we look at the Taylor expansion of $f(x) = (u_i+x)^p - u_i^p$ evaluated at $s_i+\frac{v_i}{50p}$:
\[f(s_i+\frac{v_i}{50p}) = (u_i+s_i)^p-u_i^p+\sum_{q=1}^\infty \binom{p}{q}(u_i+s_i)^{p-q}\left(\frac{v_i}{50p} \right)^{q}.\]
Notice that, for $q\ge p$, each pair of adjacent terms consists of one positive term and one negative term. 
Then, as we mentioned earlier, the absolute value of each term is upper bounded by a decreasing geometric series, so the negative term is always dominated by the positive term. 
This means that
\[f(s_i+\frac{v_i}{50p}) \ge (u_i+s_i)^p-u_i^p+ \binom{p}{1}(u_i+s_i)^{p-1}\frac{v_i}{50p}.\]
Therefore, 
\[\Ex{s^2} = \O{Q}\cdot f(s_i+\frac{v_i}{50p})^2 = \O{Q} \cdot ((u_i+s_i+\frac{v_i}{50p})^p-u_i^p)^2.\]
Moreover, since $s_i = v_i\cdot(1\pm \frac{1}{100p})$, we have $(u_i+s_i+\frac{v_i}{50p})^p < \left(u_i+\left(1 + \frac{1}{25p}\right) \cdot v_i\right)^p$. Then, by \lemref{lem:c:v}, we have
\[\left(u_i+\left(1 + \frac{1}{25p}\right) \cdot v_i\right)^p - u_i^p \le \left(1 + \frac{1}{25p}\right)^p \cdot ((u_i+v_i)^p-u_i^p) = \O{1} \cdot ((u_i+v_i)^p-u_i^p) \]
Therefore, we have
\[\Ex{s^2} = \O{Q}\cdot ((u_i+v_i)^p-u_i^p)^2 = \O{((u_i+v_i)^p-u_i^p)^2 \cdot \log n}\]
\end{proof}

The following lemma bounds the communication cost.
\begin{lemma}
\lemlab{lem:diff:comm}
With probability $1-\frac{1}{\poly(n)}$, \algref{alg:lp:diff:index} uses $k \cdot \polylog(n)$ bits of communication.
\end{lemma}
\begin{proof}
We need $\O{\log n}$ instances of \textsc{Counting} to estimate $v_i$. By \thmref{thm:counting}, this needs $k \cdot \polylog(n)$ bits of communication.
\end{proof}

\subsubsection{Putting it All Together}
With the $F_p$ difference estimator introduced in previous sections, we are able to apply the framework to achieve an adversarially robust algorithm for estimating $F_p$.

First, we state an algorithm that provides a $1.01$-approximation to the $F_p$-moment at all times, which is used in the double detector.
\begin{theorem}[\cite{WoodruffZ12}]
\thmlab{thm:fp:cts:const}
Let $S$ be a distributed stream. Let $n$ be the size of the universe. Let $k$ be the number of servers. Let $p\ge 1$. There is an algorithm that gives a $1.01$-approximation to the $F_p$-moment of vector $f$ induced by stream $S$ at all times with probability $1-\frac{1}{\poly(n)}$. This algorithm implements in $k^{p-1}\cdot \polylog(n)$ bits of communications.
\end{theorem}
We then state an algorithm that provides a static $(1+\eps)$-approximation to the $F_p$-moment.

\begin{theorem}[\cite{EsfandiariKMWZ24}]
\thmlab{thm:fp:stat}
Let $n$ be the size of the universe. Let $k$ be the number of servers. Each server holds a vector $f(j) \in \mathbb{R}^n$. Let $p \ge 2$. There is an algorithm that gives a $(1+\eps)$-approximation to the $F_p$-moment with probability $1-\frac{1}{\poly(n)}$. This algorithm runs in $k^{p-1}\cdot \polylog(n)$ bits of communications.
\end{theorem}

Combining the above subroutines gives a $\eps$-double detector.
\begin{corollary}
\corlab{cor:fp:dd}
There is a $\eps$-double detector for $F_p$-estimation problem which runs in $k^{p-1}\cdot \polylog(n)$ bits of communications.
\end{corollary}
\begin{proof}
We run the algorithm in \thmref{thm:fp:cts:const} at all times. It flags when $F$ satisfies $0.99 \cdot 2^a \leq F \leq 1.01 \cdot 2^a$. When it flags, we run the algorithm in \thmref{thm:fp:stat} to give a $(1+\eps)$-approximation with probability $1-\frac{1}{\poly(n)}$. Since there are only $\O{\log n}$ rounds, we only need to run the algorithm in \thmref{thm:fp:stat} for $\O{\log n}$ times, which gives the communication cost.
\end{proof}

Given the $\eps$-double detector and our difference estimator, we have an adversarially robust algorithm for $F_p$-estimation:
\begin{theorem}
\thmlab{thm:fp:ar}
Let $S$ be a distributed stream. Let $n$ be the number of items in the universe. Let $k$ denote the total number of servers. Let $\eps \in (0,1)$ denote some accuracy parameter. There is an adversarially robust distributed streaming algorithm which gives a $(1+\eps)$-approximation to the $F_p$-moment of the vector induced by stream $S$. This algorithm uses $\frac{k^{p-1}}{\eps^2} \cdot \polylog(n,\frac{1}{\eps})$ bits of communication with probability $0.99$.
\end{theorem}
\begin{proof}
By \corref{cor:fp:dd}, we have a $\eps$-double detector that communicates $\frac{k^{p-1}}{\eps^2} \cdot \polylog(n)$ bits. By \thmref{thm:fp:de}, 
there is a distributed $\calB_C(u,v, \eps, \eps)$-continuous difference estimator which uses $\frac{1}{\eps} \cdot k^{p-1} \polylog(n)$ bits of communication. Also, there is a distributed $\calB_S(u,v, \gamma, \eps)$-static difference estimator which uses $\frac{\gamma}{\eps^2} \cdot k^{p-1} \cdot \polylog(n)$ bits of communication. 
Therefore, by \frameref{frame:diff}, there is an adversarially robust algorithm that communicates in $\frac{k^{p-1}}{\eps^2} \cdot \polylog(n,\frac{1}{\eps})$ bits.
\end{proof}

\subsection{Adversarially Robust \texorpdfstring{$L_2$}{L2} Heavy-Hitters}
\seclab{sec:l2:hh}
In this section, we describe our $L_2$ heavy-hitter algorithm for the adversarially robust distributed monitoring model. We state the definition for the $L_2$ heavy-hitter problem as follows.

\begin{definition}[$L_2$ heavy-hitter problem]
Given a frequency vector $f\in\mathbb{R}^n$ and an accuracy parameter $\eps\in(0,1)$, the goal is to output a list $\calL$ that contains all $i\in[n]$ such that $f_i\ge\eps\cdot\|f\|_2$, but contains no such $j\in[n]$ with $f_j\le\frac{\eps}{2}\cdot\|f\|_2$. 
Moreover, for each $i\in\calL$, we require an estimate to $f_i$ with additive error $\frac{\eps}{4}\cdot\|f\|_2$.
\end{definition}

We introduce the result from \cite{HuangXZW25}, which solves the $L_2$ heavy-hitter problem in the non-robust distributed setting.
\begin{theorem}
[\cite{HuangXZW25}]
\thmlab{thm:hh:l2}
There is a distributed monitoring algorithm that solves the $L_2$ heavy-hitter problem using $\frac{k}{\eps^2} \cdot \polylog(n)$ bits of communication.
\end{theorem}

Next, we introduce an algorithm in the adversarially robust setting.
Recall that in our difference estimator framework, we divide the algorithm into $\O{\log n}$ rounds, where in each round, $F_2$ doubles. 
Within each round, we further decompose the stream into difference estimators with different granularities. 
Consider a round $a$, suppose that there is some coordinate $i$ that is a $\eps$-heavy-hitter at the end of the round, i.e., $\left(f_i\right)^2 \geq \eps^2 F_2(1, t_{a+1})$. 
Then, it must be either roughly $\eps^2$-heavy at the beginning of the round or roughly $2^r \eps^2$-heavy with respect to a duration $(t_{a,r-1}, t_{a,r})$ within this round, in which $F_2$ increases by roughly $2^r \eps F_2(1, t_{a+1})$. 
Therefore, we define our difference estimator at granularity $r$ in the $L_2$-heavy-hitter problem as the implementation of the distributed monitoring $L_2$-heavy-hitter algorithm with accuracy $2^r \cdot \eps^2$.
By applying the analysis from our difference estimator framework, we establish the correctness of our adversarially robust $L_2$-heavy-hitter algorithm.

Now, we analyze our adversarially robust algorithm for the distributed monitoring setting.
\begin{theorem}
\thmlab{thm:hh:l2:ar}
Let $S$ be a distributed stream, $n$ be the number of items in the universe, $k$ denote the total number of servers, and $\eps \in (0,1)$ denote some accuracy parameter. 
There is an adversarially robust distributed streaming algorithm that solves the $L_2$ heavy-hitter problem on $S$, which succeeds with high probability and uses $\frac{k}{\eps^2} \cdot \polylog(n,\frac{1}{\eps})$ bits of communication.
\end{theorem}
\begin{proof}
Recall that in \algref{alg:frame}, we have a counter $a$ to count the number of rounds. We have a parameter $r \in [\beta]$, where $\beta = \O{\log \frac{1}{\eps}}$, denoting the granularity of the difference estimators. Suppose there exists $i\in[n]$ such that $\left(f_i\right)^2 \geq \eps^2 F_2(1, t_{a+1})$. Then, the coordinate $i$ must either be $\frac{\eps^2}{32^2}$-heavy with respect to $F_2\left(1, t_a\right)$ or $\frac{2^r \cdot \eps^2}{32^2 \beta^2}$-heavy with respect to $F_2\left(t_{a, r-1}, t_{a, r}\right)$ for some integer $r \in[\beta]$, where we use the convention $t_{a, 0}=t_a$. Hence, running the algorithm in \thmref{thm:hh:l2} with accuracy $\frac{2^r \cdot \eps^2}{32^2 \beta^2}$ with respect to $F_2\left(t_{a, r-1}, t_{a, r}\right)$ detects that $f_i$ is a heavy-hitter, and it uses $\frac{k\beta^2}{2^r\cdot \eps^2}$ bits of communications. Since this is the same accuracy parameter as the difference estimator $\mathcal{B}_{a, r}$, by the same analysis as in \algref{alg:frame}, the communication cost follows.

Additionally, note that if a coordinate $i$ is reported as heavy in an algorithm corresponding to the difference estimator at granularity $r$, then $\mathcal{O}\left(\eps^2\right) \cdot F_2$ of the contribution of coordinate $i$ towards the overall $F_2$ still appears after $i$ is reported as a heavy-hitter. Thus, keeping track of the frequency of $i$ after it is reported using the \textsc{Counting} algorithm, we obtain an estimate of the frequency of $i$ up to an additive $\mathcal{O}(\eps) \cdot F_2$, which solves the $L_2$-heavy-hitter problem. 
\end{proof}

\section*{Acknowledgments}
Honghao Lin is supported in part by a Simons Investigator Award and Office of Naval Research award number N000142112647.  David P. Woodruff is supported in part by a Simons Investigator Award, Office of Naval Research award number N000142112647, and NSF CCF-2335412. 
Shenghao Xie is supported in part by NSF CCF-2335411. 
Samson Zhou is supported in part by NSF CCF-2335411. 
Samson Zhou gratefully acknowledges funding provided by the Oak Ridge Associated Universities (ORAU) Ralph E. Powe Junior Faculty Enhancement Award. 
The work was conducted in part while David P. Woodruff and Samson Zhou were visiting the Simons Institute for the Theory of Computing as part of the Sublinear Algorithms program.

\def\shortbib{0}
\bibliographystyle{alpha}
\bibliography{references}

\end{document}